\documentclass[acmsmall,screen]{acmart}

\usepackage{algorithm}
\usepackage{algorithmic}
\usepackage{amsfonts}
\usepackage{amsthm}
\usepackage{bm}
\usepackage{cleveref}
\usepackage{enumerate}
\usepackage{graphicx}
\usepackage{multirow}
\usepackage{multicol}
\usepackage{thmtools}
\usepackage{thm-restate}
\usepackage{xcolor}
\usepackage{xspace}
\usepackage[vlined,ruled,linesnumbered,algo2e]{algorithm2e}
\usepackage{footnote}
\usepackage{pifont}

\crefname{algorithm}{Algorithm}{Algorithms}

\AtEndPreamble{%
	\theoremstyle{acmdefinition}

}

\newcommand{\E}{\mathbb{E}}

\newcommand{\bparen}[1]{\big(#1\big)}
\newcommand{\Paren}[1]{\left(#1\right)}

\newcommand{\abs}[1]{\lvert #1 \rvert}
\newcommand{\ABS}[1]{\left\lvert #1 \right\rvert}

\newcommand{\floor}[1]{\lfloor #1 \rfloor}

\newcommand{\trho}{\frac{15}{1-\rho}} 
\newcommand{\bsc}{\kappa(p,q,\rho)} 
\newcommand{\Bsc}{\kappa\Paren{p,q,\rho}} 
\newcommand{\btot}{\beta_{\text{tot}}} 
\newcommand{\bmax}{\beta_{\text{max}}} 

\newcommand{\weakINV}{\textrm{\textup{(\ding{78})}}}

\DeclareMathOperator{\rank}{rank}
\DeclareMathOperator{\pos}{pos}
\DeclareMathOperator{\disl}{disl}
\DeclareMathOperator{\score}{score}

\newcommand{\basketsort}{\texttt{BasketSort}\xspace}
\newcommand{\expe}{expected}

\newcommand{\noisysearch}{\texttt{NoisySearch}\xspace}
\newcommand{\rifflesort}{\texttt{RiffleSort}\xspace}
\newcommand{\windowsort}{\texttt{WindowSort}\xspace}
\newcommand{\randomsubset}{\texttt{RandomSubset}\xspace}
\newcommand{\whp}{w.h.p.}

\AtBeginDocument{%
  }

\setcopyright{none}
\renewcommand\footnotetextcopyrightpermission[1]{}

\begin{document}

\pagestyle{plain}

\title{An Optimal Sorting Algorithm for Persistent Random Comparison Faults}

\author{Barbara Geissmann}
\email{barbara.geissmann@fhnw.ch}
\affiliation{%
	\institution{FHNW University of Applied Sciences and Arts Northwestern Switzerland}
	\postcode{5320}
    \city{Windisch}
	\country{Switzerland}
}

\author{Stefano Leucci}
\authornote{Corresponding author.}
\email{stefano.leucci@univaq.it}
\affiliation{%
	\institution{University of L'Aquila}
	\postcode{6710}
	\city{L'Aquila}
	\country{Italy}	
}
\author{Chih-Hung Liu}
\authornotemark[1]
\email{chliu@ntu.edu.tw}
\affiliation{%
	\institution{National Taiwan University}
	\postcode{106319}
	\city{Taipei}
	\country{Taiwan}
}

\author{Paolo Penna}
\email{paolo.penna@iohk.io}
\affiliation{
     \institution{IOG}
     \postcode{8046}
     \city{Zurich}
     \country{Switzerland}
}

\date{}


\begin{abstract}
We consider the problem of sorting $n$ elements subject to persistent random comparison errors. 
In this problem, each comparison between two elements can be wrong with some fixed (small) probability $p$, and comparing the same pair of elements multiple times always yields the same result.
Sorting perfectly in this model is impossible, and the objective is to minimize the dislocation of each element in the output sequence, that is, the difference between its position in the sequence and its true rank.

In this paper, we present the first $O(n \log n)$-time sorting algorithm that guarantees both $O(\log n)$ maximum dislocation and $O(n)$ total dislocation with high probability when $p < \frac{1}{4}$. 
This settles the time complexity sorting with persistent comparison errors in the given range of $p$ and shows that comparison errors do not increase its computational difficulty.
Indeed, $\Omega(n \log n)$ time is necessary to archive a maximum dislocation of $O(\log n)$ (even without comparison errors).
Moreover,  we further prove that no algorithm can guarantee a maximum dislocation of $o(\log n)$ with high probability, nor a total dislocation of $o(n)$ in expectation.

To develop our sorting algorithm, we solve two related sub-problems in the persistent error comparisons model, which might be of independent interest. More precisely, we that $O(\log n)$ time suffices to find a position in which to insert a new element $x$ in an almost-sorted sequence $S$ of $n$ elements having dislocation at most $d = \Omega(\log n)$, so that the dislocation of $x$ in the resulting sequence is $O(d)$ with high probability (which can be equivalently thought as the problem of estimating the rank of $x$ in $S$).
We also show that the maximum (resp.\ total) dislocation of an approximately sorted sequence $S$ of $n$ elements can be lowered to $O(\log n)$ (reps. $O(n)$) in $O(n d)$ time, with high probability, where $d$ is an upper bound on the maximum dislocation of $S$.
\end{abstract}

\maketitle

\newpage

\section{Introduction}\label{sec:introduction}


We study the problem of \emph{approximately sorting} $n$ distinct elements under \emph{persistent} random comparison faults.
In this fault model,  
the outcome of each comparison is wrong with some fixed error probability $p<1/2$, and correct with probability $1-p$. 
These comparison faults are independent over all possible pairs of elements, but they are persistent, i.e., comparing the same pair of element multiple times always results in the same outcome.
In other words, the persistent random comparison fault model does not allow resampling.

Our sorting problem is related to the well-known \emph{feedback arc set for tournaments} (FAST) problem~\cite{Alon06,AilonCN08,Kenyon-MathieuS07,CharbitTY07,AlonLS09}: given a tournament, i.e., a directed graph in which every pair of nodes have exactly one arc between them, compute a permutation of the nodes that minimizes the number of arcs violating the permutation, i.e., the number of arcs whose tail appears before the head in the permutation. 
By viewing elements as nodes and comparison outcomes as arcs, an acyclic tournament defines a total order of the elements and vice versa. From this viewpoint, comparison faults are arcs directed in the ``wrong direction'' in the graph.

For the \emph{non-persistent} random comparison faults model, where the same pair of elements can be compared multiple times and the errors (or lack of thereof) in the returned outcomes are independent, many combinatorial problems have been studied in 80s and 90s~\cite{Pelc02}.
In particular, Feige et al.~\cite{FeigeRPU94} developed an algorithm for the \emph{noisy binary search problem} based on \emph{noisy binary search trees}. 
In this problem, we are given a sorted sequence of $n$ elements plus an additional element $x$ as inputs, and we want to find the rank of $x$ in the sequence.
The algorithm of \cite{FeigeRPU94} developed uses $O(\log \frac{n}{Q})$ comparisons to report the correct rank with success probability of at least $1-Q$. Moreover, the authors used their noisy binary search tree as a building block to design algorithms for (i) selecting the $k$-th smallest out of $n$ elements, and (ii) sorting $n$ elements. These latter algorithms require $O(n\log \frac{\max\{k,\,n-k\}}{Q})$ and $O(n\log \frac{n}{Q})$ comparisons, respectively, and succeed with a probability of at least $1-Q$.
The same work also proves matching lower bounds for the noisy binary search problem, the problem of selecting the $k$-th smallest element in an (unsorted) input sequence, and that of sorting an input sequence.
Recently, 30 years after Feige et al's seminal work, Gu and Xu~\cite{GuX23} derived tight non-asymptotic bounds on the number of comparisons for sorting and binary search. 

In contrast, for \emph{persistent} random comparison faults, it is \emph{impossible} to perfectly sort an input sequence even with a (positive) constant success probability. 
In this situation, one seeks instead an \emph{approximately} sorted sequence, where the quality of the output is evaluated by using the notion of \emph{dislocation}. 
The dislocation of an element in a sequence is the absolute difference between its position in the sequence and its true rank. 
The \emph{maximum dislocation} of a sequence is the maximum among the dislocations of all the elements in the sequence, and the \emph{total dislocation} of a sequence is the sum of the dislocations of all its elements.

Bravernman and Mossel \cite{BravermanM08} worked on computing \emph{maximum likelihood permutations}, which correspond to solving the FAST problem on a tournament graph sampled from the probability distribution induced by the comparison errors. They proved that the maximum dislocation and the total dislocation of a maximum likelihood permutation are  $O(\log n)$ and $O(n)$, respectively, with \emph{high probability}, i.e., with a probability of at least $1-\frac{1}{n}$.
They also gave an algorithm that takes $O\bparen{n^{3+f(p)}}$ time to compute a maximum likelihood permutation with high probability, where $f(p)$ is a function of $p$. Unfortunately, $f(p)$ is roughly $\frac{110525}{(1-2p)^4}$, as it can be seen by inspecting the analysis in \cite{BravermanM08},  which renders the algorithm unpractical.
Later, Klein et al.~\cite{KleinPSW11} proposed an $O(n^2)$-time sorting algorithm called \emph{BucketSort} 
that achieves a maximum dislocation of $O(\log n)$ with high probability provided that $p \le \frac{1}{20}$, while no bound was given for the total dislocation.
The running time was subsequently improved to $\widetilde{O}(n \sqrt{n})$ for $p < \frac{1}{16}$ in \cite{GeissmannLLP20} while achieving logarithmic maximum dislocation w.h.p.\ and linear total dislocation in expectation.


To sum up, all the existing works considering \emph{persistent} random comparison faults achieve a maximum dislocation of $O(\log n)$ (w.h.p.) and require $\Omega(n \sqrt{n})$ time to approximately sort $n$ elements, while, in the case of \mbox{non-persistent} random comparison faults, sorting (with no dislocation, w.h.p.) requires only $O(n\log n)$ time.
This begs the following question:
\smallskip
\begin{description}
  \item[Question 1:] How much time is needed to (approximately) sort $n$ elements in the presence of \emph{persistent} random comparison faults? 
\end{description}
\smallskip

To answer the above question, one might hope of employing a strategy similar to that used by the sorting algorithm of \cite{FeigeRPU94} for non-persistent random comparison faults, which iteratively extends a sorted sequence with new elements by performing noisy binary searches. 
However, due to the persistent nature of our random comparison faults, one cannot maintain a perfectly sorted sequence. Then, the following question also arises: 

\smallskip
\begin{description}
  \item[Question 2:] How can we \emph{quickly} find the approximate rank of an element in an approximately sorted sequence in the presence of \emph{persistent} random comparison faults? 
\end{description}
\smallskip

Furthermore, while inserting an element in its correct position (i.e., in the position corresponding to the element's rank) in a perfectly sorted sequence results in a sequence that is still sorted, the same does not necessarily hold in an approximate sense. This is because the maximum dislocation of an approximately sorted sequence might increase following the insertion of one or more elements. To circumvent this problem, it would be helpful to locally rearrange the elements of an approximately sorted sequence in a way that reduces the resulting dislocation.

 
\smallskip
\begin{description}
  \item[Question 3:] Can the (maximum and/or total) dislocation of an approximately sorted sequence be quickly improved (up to some lower limit) in the  \emph{persistent} random comparison faults model? 
\end{description}
\smallskip

Finally, since the best maximum and total dislocations achieved by existing works with high probability are $O(\log n)$ and $O(n)$, respectively, one might wonder whether these bounds are (asymptotically) tight or if they can be improved (w.h.p).
\smallskip
\begin{description}
    \item[Question 4:] What are the best possible (asymptotic) maximum and total dislocations that can be achieved with high probability in the presence of \emph{persistent} random comparison faults? 
\end{description}
\smallskip

\subsection{Our contributions}

In this work, we tackle the above questions and we develop a sorting algorithm called \emph{\rifflesort} that takes $O(n\log n)$ time and, if $p < \frac{1}{4}$, achieves a maximum dislocation of $O(\log n)$ and a total dislocation of $O(n)$ with high probability.
We also exploit the intrinsic random nature of comparison faults to derandomize \rifflesort while maintaining the same asymptotic guarantees, as long as the comparison errors do not depend on the position of the elements in the input sequence. 

On the negative side, we prove that no algorithm can achieve maximum dislocation $o(\log n)$ with high probability, nor total dislocation $o(n)$ in expectation.
Moreover, if $o(n \log n)$ time was sufficient to compute a sequence having a maximum dislocation of $d = O(\log n)$, there would exist an algorithm that sorts $n$ elements in the absence of comparison faults using only $o(n\log n)+O\bparen{\frac{n}{\log n}\cdot ((\log n)(\log\log n))}=o(n\log n)$ time, contradicting the classical $\Omega(n \log n)$ lower bound for comparison-based sorting algorithms.
Hence, \rifflesort is simultaneously optimal, in an asymptotic sense, in terms of maximum dislocation, total dislocation, and running time. 
This answers Question 1 for constant $p < \frac{1}{4}$ and Question 4.
Handling all constant error probabilities in $p \in [0, \frac{1}{2})$ in time $O(n \log n)$ remains an interesting open problem.
\Cref{tb-recurrent} compares the state-of-the-art works with our result.


\begin{table}
	\centering
  \caption{Summary of the state-of-the-art. Here, $f(p)$ is some function of $p$. The upper bound on $f(p)$ resulting from the derivation in \cite{BravermanM08} is $\frac{110525}{(1-2p)^4}$. The notation $\widetilde{O}(g(n))$ is a shorthand for $O(g(n) \text{polylog} n)$.}	\label{tb-recurrent}
	\renewcommand{\arraystretch}{1.2}
	\begin{tabular}{|c|cc|c|}
		\hline
		\textbf{Time}	& \textbf{Maximum Dislocation} & \textbf{Total Dislocation} & \textbf{Reference} \\ \hline
		$O(n^{3+f(p)})$	 & $O(\log n)$ \whp & $O(n)$ \whp &  \cite{BravermanM08} \\ 
		$O(n^2)$	& $O(\log n)$ \whp & $O(n\log n)$ \whp  &  \cite{KleinPSW11} 
		\\
    $\widetilde{O}(n \sqrt{n})$ w.h.p. & $O(\log n)$ w.h.p & $O(n)$ expected & \cite{GeissmannLLP20} \\
	   \hline \hline
		$O(n\log n)$	& $O(\log n)$ \whp & $O(n)$ \whp & \multirow{2}{*}{\textbf{this work}} \\ 
		Any & $\Omega(\log n)$ \whp & $\Omega(n)$ \expe &  \\ \hline
	\end{tabular} 
\end{table}

To design \rifflesort we develop two sub-routines, \noisysearch and \basketsort, both of which can be of independent interest.
\noisysearch takes $O(\log n)$ time to report an approximate rank of an external element $x$ w.r.t.\ an approximately sorted sequence having dislocation at most $d$ when only noisy comparisons between $x$ and the elements in the sequence are allowed. The reported rank differs from the true rank by at most $O(d)$ with high probability.
\basketsort re-sorts an approximately sorted sequence having $m$ elements and maximum dislocation at most $d$ and returns, with high probability, a sequence having maximum dislocation $O(\log m)$ and total dislocation $O(m)$ in time $O(d \cdot m)$.
We note that \noisysearch works for any constant error probability $p<\frac{1}{2}$, while \basketsort requires $p < \frac{1}{4}$, i.e., the constraint on $p$ of \rifflesort is due to the analogous constraint in \basketsort. These algorithms answer Question 2 and Question 3, for their respective ranges of $p$. 

%


\subsection{Other related work}

Computing with errors is often considered in the framework of a two-person game called \emph{R\'{e}nyi-Ulam Game}. In this game, a questioner tries to identify an unknown object $x$ from a universe $U$ by asking yes-or-no questions to a responder, but some of the answers might be wrong. The case in which $U = \{1, \dots, n \}$, the questions are of the form ``is $s > x$?'', and each answer is independently incorrect with probability $p < \frac{1}{2}$ has been considered by Feige et al.~\cite{FeigeRPU94}, as we already mentioned. Ben-Or and Hassidim \cite{Ben-OrH08} then showed how to find $s$ using an optimal amount of queries up to additive polylogarithmic terms.
The variant in which responder is allowed to adversarially lie up to $k$ times has been proposed by Rényi~\cite{Renyi61} and Ulam~\cite{Ulam76}, which has then been solved by Rivest et al.~\cite{RivestMKWS80} using only $\log n + k \log log n + O( k \log k )$ questions, which is tight.
Among other results, near-optimal strategies for the distributional version of the game have been devised in \cite{Damaschke16}.
For more related results, we refer the interested reader to \cite{Pelc02} for a survey and to \cite{Cicalese13} for a monograph.

Regarding other sorting algorithms not mentioned before, Braverman et al.~\cite{BravermanMW16} studied algorithms which use \emph{a bounded number of rounds} for some ``easier'' versions of sorting, e.g., distinguishing  the top $k$ elements from the others. 
The comparisons to be performed in each round are decided a priori, i.e., they do not depend on the results of the comparisons in the round.
In each round, a fresh set of comparison outcomes is generated, and each round consists of $\delta \cdot n$ comparisons. 
They evaluate the quality of an algorithm by estimating the number of ``misclassified'' elements and also consider a variant in which errors result in missing (rather than incorrect) comparison results.

Ajtai et al.~\cite{AjtaiFHN16} provided algorithms using \emph{subquadratic} time (and number of comparisons) when errors occur only between  elements whose difference is at most some fixed threshold. Also,  Damaschke~\cite{Damaschke16} gave a subquadratic-time algorithm when the number $k$ of errors is known in advance. 

Kunisky et al.\ \cite{KuniskyInference} consider the information-theoretic problem of detecting whether a given tournament is generated by choosing each arc direction uniformly at random, or according to the outcomes of independent noisy comparisons. They show that this can be done w.h.p.\ when $p = \frac{1}{2} - \omega(n^{-3/4})$, and this is tight. Moreover, for $p = \frac{1}{2} - \omega(n^{-3/4})$, it is possible to recover the correct ranking with probability $1 - o(1)$, which is also tight. In the latter case, an approximate ranking w.r.t.\ to some \emph{ranking alignment measure} can be found in polynomial time. Another recent paper \cite{KuniskySWY25} considers similar detection and recovery problems in the case in which a noisy tournament is planted as a subgraph in an otherwise randomly directed \emph{host} graph, where the arcs of both the planted tournament and the host graph only exist with some fixed probability.


	Finally, we point out that this manuscript is an extended and improved version of two conference papers~\cite{GeissmannLLP17,GeissmannLLP19}, which found applications in the context of learning theory \cite{HopkinsKLM20}, graph clustering \cite{BianchiP21}, knowledge and data engineering \cite{BerettaNTV23}, and database management \cite{AddankiGS21}.
  The major improvement from \cite{GeissmannLLP17,GeissmannLLP19} is the novel \basketsort algorithm, which supersedes \windowsort~\cite{GeissmannLLP17} in \rifflesort, and allows us to relax the previous constraint $p<\frac{1}{16}$ on the probability of error to $p < \frac{1}{4}$. 
  Following the publication of ~\cite{GeissmannLLP17}, we developed the aforementioned $\widetilde{O}(n \sqrt{n})$-time algorithm \cite{GeissmannLLP20}, which is now superseded by \rifflesort. 

%



\subsection{Organization of the paper}

\Cref{sec:preliminaries} describes the error model and introduces some basic notation.
\Cref{sec:algorithm,sec:noisy-search,sec:baskset-sort} are organized in a top-down fashion: \Cref{sec:algorithm} presents the main algorithm \rifflesort by using \noisysearch and \basketsort as black-box subroutines; then \Cref{sec:noisy-search} and \Cref{sec:baskset-sort} present the details of \noisysearch and \basketsort, respectively.
\Cref{sec:lower_bound} discusses our lower bounds, and \Cref{sec:derandomization} concludes the paper by showing how \rifflesort can be derandomized when errors happen at least with some constant probability and the comparison errors do not depend on the position of the elements in the input sequence.
For the sake of self-containment, Appendix~\ref{apx:random_subset} describes an algorithm to choose a uniformly random subset of a given cardinality from a universe of $n$ elements while using only $O(n)$ random bits, which is needed by our derandomization strategy of \Cref{sec:derandomization}.

\section{Preliminaries}\label{sec:preliminaries}

\paragraph*{Basic notation} Given a set or a sequence $S$ of elements and an element $x$ (not necessarily in $S$), we define the \emph{rank} of $x$ w.r.t.\ $S$ as $\rank(x,\,S)= 1+\abs{\{ y \in S : y < x\}}$.  Moreover, if $S$ is a sequence and $x \in S$, we use $\pos(x,\,S) \in [1, |S|]$ to denote the \emph{position} of $x$ in $S$.
Furthermore, let $\disl(x,\,S)= |\pos(x,S)-\rank(x,S)|$ denote the \emph{dislocation} of $x$ in $A$, and let $\disl(S)= \max_{x\in S} \disl(x,\,S)$ denote the \emph{maximum dislocation} of the sequence $S$. 
Throughout the paper we use $\ln z$ and $\log z$ to refer to the natural and the binary logarithm of $z \in \mathbb{R} \setminus \{0\}$, respectively.

\paragraph*{Our error model for approximate noisy sorting}
In the approximate noisy sorting problem, we consider a set $S$ of elements possessing a true linear order that can only be observed through noisy comparisons. Given any two distinct elements $x, y \in S$, we write $x < y$ (resp.\ $x > y$) to denote that $x$ is smaller (resp.\ larger) than $y$ according to the true order. 
For each unordered pair $\{x,y\}$ of distinct elements, there exists \emph{an error probability} $p_{\{x,y\}}$ such that a comparison between  the two elements reports the correct order relation (i.e., $x < y$ or $x > y$) with probability $1-p_{\{x,y\}}$ and the complementary, erroneous, outcome with probability $p_{\{x,y\}}$. 
We write $x \prec y$ or $y \succ x$ (resp.\ $x \succ y$ or $x \prec y$) to denote that $x$ is reported to be smaller (resp.\ larger) than $y$. 
Comparison errors between different pairs of elements happen independently while the comparison outcome of each pair of elements is \emph{persistent} so that comparing the same pair of elements multiple times always returns the same relative order of the elements (i.e., resampling is not allowed). 

We say that the elements in $S$ are subject to persistent random comparison faults (or errors) with probability at most $p$ (resp.\ at least $q$) if $p_{\{x,y\}} \le p$ (resp.\ $p_{\{x,y\}} \ge q$) for all pairs $\{x,y\}$.
Throughout the paper, $q$ and $p$ will be some constants with $0 \le q \le p < \frac{1}{2}$, and we might impose further constraints on the range of allowed values.

An algorithm that approximately sorts $S$ takes as input a sequence of the elements in $S$ and outputs an ``approximately sorted'' sequence $S'$, i.e., permutation of the input sequence, while only performing noisy comparisons.
Each comparison requires constant time, and the goal is that of \emph{quickly} computing a sequence $S$ with small dislocation.
Notice that, in general, the order in which the elements appear in the input permutation of $S$ might depend on the comparison errors, i.e., one can imagine the following process: all the comparison errors between distinct pairs of elements are randomly decided a-priori (independently and according to the probabilities described above), then an adversary is given access to the elements in $S$, to their true order, and to all the noisy outcomes of the comparisons, and chooses the input sequence to be used by the algorithm.

\paragraph*{Our error model for approximate noisy binary searching}
As discussed in the introduction, we also consider the problem of estimating the rank $\rank(x, S)$ of an element $x$ in an input sequence $S$ that satisfies $\disl(S) \le d$, where $d$ is an integer input parameter.

In this problem, only comparisons between $x$ and the elements in $S$ are allowed,  and these comparisons are noisy: errors occur with probability at most $p$ for some constant $p \in [0, \frac{1}{2})$, errors are independent, and the comparison outcomes are persistent. Here, the order of the elements in the sequence $S$ cannot depend on the comparison errors between $x$ and the elements in $S$.

\section{Our Sorting Algorithm}~\label{sec:algorithm}

In this section, we design an algorithm that approximately sorts an input sequence $S$ of $n$ elements 
subject to persistent random comparison faults with probability at least $q$ and at most $p$ with $p \le \frac{9q+1}{8q+11}$.
The algorithm uses $O(n\log n)$ worst-case time, which is asymptotically optimal, and returns a permutation of $S$ having maximum dislocation $O(\log n)$ and total dislocation $O(n)$, with high probability. 
For the case $q=0$, i.e., when no lower bound  on the error probability is known, the above condition becomes $p < \frac{1}{11}$, while for the case $p=q$, which is the one that has been considered in \cite{BravermanM08, KleinPSW11, GeissmannLLP20}, the algorithm can handle any constant $p < \frac{1}{4}$. 

To some extent, the flavor of our algorithm is a mixture of \emph{Insertion Sort} and \emph{Merge Sort}.
On the one hand, our algorithm extends a partially sorted sequence by finding the (approximate) ranks of new elements and inserting them in the computed positions; on the other hand, our algorithm inserts such elements in batches so that the size of the intermediate sequences doubles at each iteration.
Due to the persistent comparison faults, it is not possible to guarantee that these intermediate sequences will be correctly sorted. Instead, our algorithm ensures that such sequences 
are kept approximately sorted w.h.p. More precisely, since inserting a batch of elements causes the dislocation of the sequence to increase multiplicatively, we counteract this effect by post-processing the resulting sequence.

In order to implement our algorithm, we answer the following three key questions:
\begin{enumerate}
	\item How do we quickly find an approximate rank of a new element to be inserted in an approximately sorted sequence?
%
	\item How do we bound the increase in dislocation resulting from simultaneously inserting a batch of elements in an approximately sorted sequence?
	\item How can the sequence resulting from inserting a batch of elements be post-processed to counteract the increase in dislocation?
\end{enumerate}

To address the first question, we design a search algorithm called \noisysearch that finds an approximate rank of a new element with respect to an approximately sorted sequence of $m$ elements in $O(\log m)$ time, w.h.p. The computed approximate rank differs from the true rank of the new element by at most a constant times the maximum dislocation of the sequence. The details of this algorithm are given in \Cref{sec:noisy-search} and its guarantees can be summarized as follows:

\begin{restatable*}{theorem}{thmnoisysearch}
	\label{thm:noisy_binary_search}
	Let $S$ be a sequence of $m$ elements, let $d$ be an integer with $d \ge \max\{\disl(S),\ln m\}$, and let $x \not\in S$ be an additional element.
    Under our error model for approximate noisy binary searching,
    there exists an algorithm $\noisysearch(S, x, d)$ 
    that runs in $O(\log m)$ worst-case time and returns an index $\tau_x \in [ \rank(x,S) - \alpha d,  \rank(x,S) + \alpha d]$ with probability at least $1- O(m^{-6})$, where $\alpha \ge 2$ is a constant that depends only on $p$.
\end{restatable*}

For the second question, we connect the batch insertion to \emph{an urn experiment} that is seemingly unrelated to sorting and can be stated as a basic probability problem (see \Cref{lem:draw_no_long_monocolor} in the analysis of the algorithm). 

Finally, to restore the dislocation of a sequence, we develop a new sorting algorithm for persistent random comparison faults called \basketsort.
This algorithm takes a sequence of $m$ elements and an upper bound $d$ on the maximum dislocation of such sequence as input and computes, in $\Theta(m d)$ time, a new sequence with a maximum dislocation of  $O(\log m)$ and a total dislocation of $O(m)$ w.h.p.\footnote{
  Since the running time of \basketsort depends on the maximum dislocation of its input sequence, directly applying \basketsort on $S$ cannot achieve a running time of $O(n\log n)$.} This algorithm is described in \Cref{sec:baskset-sort}, which formally proves the following:

\begin{restatable*}{theorem}{thmbasketsort}
  \label{thm:basketsort}
    Consider a set of $m$ elements subject to persistent random comparison faults with probability at least $q$ and at most $p$, for some constants $p,q$ with $0 \le q \le p < \frac{9q+1}{8q+11}$.
    Let $S$ be a sequence of these $m$ elements (possibly chosen as a function of the errors), and let $w_S \ge \disl(S)$.
  $\basketsort(S, w_S)$ is a deterministic algorithm that approximately sorts $S$ in $O(m w_S)$ worst-case time.
    The maximum dislocation and the total dislocation of the returned sequence are at most $\bmax \cdot \ln m$ and $\btot \cdot m$, respectively, with probability $1 - O(m^{-3})$, where $\bmax \ge 2$ and $\btot$ are constants that come from \Cref{lemma:bs_maximum-dislocation} and \Cref{lem:total_dislocation_high}, respectively, and depend only on $p$ and $q$.
\end{restatable*}

\subsection{Algorithm description}
\label{sec:rifflesort_description}

In the rest of this section, we assume that \Cref{thm:noisy_binary_search} and \Cref{thm:basketsort} hold (as they will be proven in \Cref{sec:noisy-search,sec:baskset-sort}, respectively), we let $\alpha$, $\bmax$, and $\btot$ be the corresponding constants in the two theorems, and we define $\gamma = 226 \alpha$. We now  describe our main sorting algorithm, which we name \emph{\rifflesort} and whose pseudocode is given in \Cref{alg:rifflesort}.

\begin{algorithm2e}[t]
	\caption{\rifflesort\unskip$(S)$}
	\label[algorithm]{alg:rifflesort}
    $S' \gets $ A permutation of $S$ chosen u.a.r.\ from the set of all possible permutations\label{ln:rs_shuffle}\;
    \tcp{The above step is only needed if the order of the elements in $S$ can depend on the errors (otherwise, it suffices to use $S' = S$).}

    \BlankLine 

    $T_0, T_1,\ldots, T_k \gets$ a random partition of the elements in  $S'$ such that $|T_0| = \lceil \sqrt{n} \rceil$, $|T_i| = 2^{i-1} \cdot \lceil \sqrt{n} \rceil$ for $i=1, \dots, k-1$, and $1 \le |T_k| \le 2^k \cdot \lceil \sqrt{n} \rceil$\;

    \BlankLine

    $\widetilde{S}_0 \gets$ Subsequence of $S'$ containing the elements in $T_0$\;
    $S_0 \gets \basketsort(\widetilde{S}_0,\, |\widetilde{S}_0|)$\label{ln:rf_basketsort_S0}\;

    \BlankLine

    \For{$i=1, \dots, k$}
    {
        \ForEach{$x\in T_i$}
        {
            $r_x\gets \noisysearch(S_{i-1},\, x, \, \bmax \cdot \ln n )$\label{ln:rf_noisysearch}\;
        }
    
        $\widetilde{S}_i \gets $ Simultaneously insert each element  $x \in T_i$ into position $r_x$ of $S_{i-1}$ (break ties according to $S'$)\label{ln:rf_simultaneous_insert}\;
        $S_i \gets \basketsort(\widetilde{S}_i,\, \gamma \cdot \bmax  \cdot \ln n)$\label{ln:rf_basketsort_Si}\;
	}
    \BlankLine
		\Return $S_k$\;
\end{algorithm2e}

The algorithm starts by shuffling the elements in $S$ in order to obtain a new sequence $S'$ in which the position of the elements is guaranteed not to depend on the comparison errors.\footnote{If this is already known to be the case for $S$, we can choose $S'=S$.} 
Next, \rifflesort partitions the elements in $S'$ into $k+1$ subsets $T_0, T_1, \dots, T_k$, such that $|T_0| = \lceil \sqrt{n} \rceil$, $|T_i| = 2^{i-1} \cdot \lceil \sqrt{n} \rceil$ for $i=1, \dots, k-1$, and $1 \le |T_k| \le 2^k \cdot \lceil \sqrt{n} \rceil$.

More precisely, we let $k$ be the smallest integer that satisfies $\lceil \sqrt{n} \rceil + \lceil \sqrt{n} \rceil\cdot  \sum_{i=1}^k 2^{i-1} \ge n$, namely $k = \left\lceil \log \frac{n}{\lceil \sqrt{n} \rceil} \right\rceil < 1 + \frac{1}{2} \log n$. Then, we pick $T_0$ as a subset of $\lceil \sqrt{n} \rceil$ elements chosen u.a.r.\ from those in $S'$ (without replacement); $T_i$, for $i=1,\dots,k-1$, as a subset of $2^{i-1}\cdot \lceil \sqrt{n} \rceil$ elements chosen u.a.r.\ from $S' \setminus \bigcup_{j=0}^{i-1} T_j$; and $T_k$ as the set of the remaining elements, i.e., $T_k = S' \setminus \bigcup_{j=0}^{k-1} T_j$.

As a preliminary step, \rifflesort runs \basketsort on the subsequence $\widetilde{S}_0$ of $S$ that contains the element in $T_0$ in order to get an approximately sorted sequence $S_0$. This call to \basketsort uses $|\widetilde{S}_0|$ as a trivial upper bound on $\disl(\widetilde{S}_0)$.
Then, for $i=1, 2,\ldots, k$, \rifflesort  performs \noisysearch on $S_{i-1}$ to compute an approximate rank $r_x$ for each element $x$ in $T_i$.
These elements are simultaneously inserted into $S_{i-1}$ according to their approximate ranks, thus obtaining a sequence $\widetilde{S}_i$. If ties arise (i.e., if multiple elements share the same approximate rank), they are broken according to the order of the tied elements in $S'$. This ensures that the order of the elements in $\widetilde{S}_i$ does not (indirectly) depend on the comparisons errors involving elements in $T_{i+1}, \dots, T_{k}$.
Finally, \basketsort is used to re-sort $\widetilde{S}_i$ into a new sequence $S_i$ with a smaller maximum dislocation, before moving to the next iteration.
Once all iterations are complete, the algorithm returns $S_k$.

From a high-level perspective, we have that, in each iteration, the maximum dislocation of $S_{i-1}$ is at most $\bmax\cdot \ln n$,
while inserting the elements of $T_i$ into $S_{i-1}$ according to the approximate ranks computed by \noisysearch increases the maximum dislocation to at most $\gamma\cdot \bmax\cdot\ln n$. Finally, \basketsort decreases the maximum dislocation back to at most $\bmax\cdot \ln n$, where the quantities $\bmax$ and $\gamma$ are constants that depend only on $p$ and $q$. In particular, 
$\gamma$ stems from \noisysearch and our analysis of the simultaneous insertion step of \rifflesort (see \Cref{lem:merge_constant_disl_increase}).

\subsection{Algorithm analysis}

\noindent We start by analyzing the running time of \rifflesort.

\begin{lemma}\label{lem:rifflesort_runtime}
	The worst-case running time of \Cref{alg:rifflesort} is $O(n \log n)$.
\end{lemma}
\begin{proof}
    Clearly, it takes $O(n \log n)$ time to compute $S'$ and the random partition $T_0$, $\ldots$, $T_k$.\footnote{The exact complexities of these steps depend on whether we are allowed to sample uniformly random integers in a range in $O(1)$ time. Even if this is not the case, we can still independently sample $n-1$ integers, where the $i$-th such integer is sampled u.a.r.\ from $\{0, \dots, n-i\}$, in $O(n \log n)$ time with probability at least $1-n^{-2}$ using a simple rejection strategy (see the proof of \Cref{lemma:randomsubset_numbits} for details).
    This allows us to compute $S'$ using, e.g., the Fisher-Yates shuffle (see, e.g., \cite{Knuth98fisheryates}) on $S$.
    Moreover, the partition $T_0, \dots, T_k$ can be found using $O(n)$ random bits with probability at least $1-n^{-2}$ (see \Cref{sub:rifflesort_n_bits} and \Cref{apx:random_subset}).
    To maintain a worst-case upper bound on the running time, we can handle the unlikely event in which $O(n \log n)$ bits do not suffice by simply stopping the algorithm and returning any arbitrary permutation of $S$. This will not affect the high-probability bounds in presented in the sequel.} 
  Moreover, by \Cref{thm:basketsort}, the first call to \basketsort requires $O(| T_0 |^2) = O(\lceil \sqrt{n} \rceil^2) = O(n)$ time. 
We can therefore restrict our attention to the generic $i$-th iteration of the for loop.
  From \Cref{thm:noisy_binary_search} we have that each invocation of \noisysearch on $S_{i-1}$ takes time $O(\log |S_{i-1}|) = O(\log n)$, leading to a total search time of $O(\abs{T_i} \log n)$.
Then, the insertion trivially takes $O(\abs{S_{i-1}}+\abs{T_i})=O(\abs{S_i})$ time.
    Furthermore, by \Cref{thm:basketsort}, $\basketsort(S_i,\, \gamma \cdot  \bmax  \cdot \ln n)$ requires $O(\abs{S_i}\log n)$ time (recall that both $\gamma$ and $\bmax$ are constants depending only on $p$ and $q$).
To sum up, since  $\abs{T_i} \le 2^{i-1} \lceil \sqrt{n} \rceil$ and $\abs{S_i} \le \sum_{j=0}^i \abs{T_j}=O(2^i\sqrt{n})$ for $i=1,\dots, k$,
  the $i$-th iteration takes time $O(\abs{T_i}\log n+\abs{S_i}+\abs{S_i}\log n)=O(|S_i| \log n) = O(2^i\sqrt{n} \log n)$, and thus the total running time of \Cref{alg:rifflesort} is $O(n\log n)+ O\left( \big(\sqrt{n} \log n \big) \cdot \sum_{i=1}^k 2^i \right) =O(n\log n + 2^k \sqrt{n} \log n) = O(n \log n)$, where we used $2^k < 2^{1 + \frac{1}{2} \log n} = 2 \sqrt{n}$.
\end{proof}

We now analyze the maximum and the total dislocation of the returned sequence.
We start by upper bounding the maximum dislocation resulting from the batch insertion and, to this aim, we consider a thought experiment involving urns and randomly drawn balls.

\begin{lemma}\label{lem:draw_no_long_monocolor}
  Consider an urn containing $N$ balls, $M \ge \frac{N}{2}$ of which are white while the others are black.
  The balls are iteratively drawn from the urn without replacement until the urn is empty. 
  Let $\ell$ be an integer that satisfies $8 \log N \le \ell \le \frac{M}{432}$.
  If $M \ge 64$, the probability that there exists some contiguous subsequence of at least $54\ell$ drawn balls containing $\ell$ or fewer white balls is at most $N^{-6}$.
\end{lemma}
\begin{proof}
  Let $b_i$ denote the $i$-th drawn ball. 
We consider the sequence of drawn balls and we prove an upper bound the probability that,
  given any position $i \le \left\lceil \frac{N}{2} \right\rceil$, at most $\ell$ balls in $b_{i}, \dots, b_{i+54\ell-1}$ are white. The same bound immediately holds for sequences stating from any position $i > \left\lceil \frac{N}{2} \right\rceil$ since we can apply similar arguments by considering the sequence of drawn balls in reverse order. We distinguish two cases.

 %
 %
 %

  If $i \le \frac{M}{4}$ then, for any integer $0 \le j \le \frac{M}{4} $, the probability that $b_{i+j}$ is white (regardless of the colors of the other balls in $b_0, \dots, b_{i+\lfloor M/4 \rfloor }$) is at least
  $
    \frac{M - M/2}{N} \ge \frac{M/2}{2M} = \frac{1}{4} > \frac{1}{10}
  $.


  If $\frac{M}{4} < i \le \left\lceil \frac{N}{2}\right\rceil$, then let $X$ be the number of white balls in $\{b_1, \dots, b_{i-1} \}$.
  Since $X$ is distributed as a hypergeometric random variable of parameters $N$, $M$, and $i-1$ we have that $\E[X] = \frac{(i-1)M}{N} \le \frac{M}{2}$.

  By using the tail bound (see, e.g., \cite{Skala13}) $\Pr(X \ge \E[X] + t(i-1)) \le e^{-2t^2 (i-1)}$ for $t \ge 0$, we obtain: 
\begin{align*}    
	\Pr\left(X \ge \frac{3M}{4}\right) &=
	\Pr\left(X \ge \frac{M}{2} + \frac{M}{4}\right) \le
  \Pr\left(X \ge \E[X] + \frac{M}{4(i-1)} \cdot (i-1)\right) \\ 
  &\le e^{-2 \frac{M^2}{16 (i-1)^2} (i-1)}
  = e^{-\frac{M^2}{16(i-1)}}
  \le e^{- \frac{M}{16} } 
  \le e^{-4} < \frac{1}{50},
\end{align*}
  where we used $i-1 \le \frac{N}{2} \le M$ and $M \ge 64$.

  Consider now the case $X < \frac{3M}{4}$, which happens with probability at least $\frac{49}{50}$.
  In this case, for any integer $j \le \frac{M}{8} $, the probability that $b_{i+j}$ is white, regardless of the specific colors of the other balls in $b_1, \dots, b_{i+\lfloor\frac{M}{8}\rfloor}$, is at least:
  $
    \frac{M-(\frac{3M}{4}+ \frac{M}{8})}{N} \ge
    \frac{M/8}{2M} = \frac{1}{8}
  $.
  Using the union bound on the events $X \ge \frac{3M}{4}$ and ``$b_{i+j}$ is black'', we have that $b_{i+j}$ is white with probability at least $1- \frac{7}{8} - \frac{1}{50} > \frac{1}{10}$.

  Therefore, using $54 \ell \le \frac{M}{8}$, we have that, in both the first case and the second case, the probability that at most $\ell$ balls in $b_{i}, \dots, b_{i+54\ell-1}$ are white is at most:
\begin{align*}    
	\sum_{j=0}^{\ell} \binom{54\ell}{j} \left( \frac{1}{10} \right)^j \left( \frac{9}{10} \right)^{54\ell-j} 
  &\le (\ell+1)  \binom{54\ell}{\ell} \left(\frac{9}{10} \right)^{54\ell}   
	\le (\ell+1)  \left(\frac{54e\ell}{\ell}\right)^\ell \left(\frac{9}{10}\right)^{54\ell} \\
  &\le (\ell+1)  \left(54e \left(\frac{9}{10}\right)^{54} \right)^\ell 
	< \frac{\ell+1}{2^\ell},
\end{align*}    
where we used the inequality $\binom{\eta}{z} \le \left( \frac{e \eta}{z} \right)^z$.

By using the union bound on all $i$ and the inequality $8 \log N \le \ell \le \frac{M}{432} < N-1$, we can upper bound the sought probability as:
\[
  N \cdot \frac{\ell+1}{2^\ell} < \frac{N^2}{2^{8 \log N}} = N^{-6}. \qedhere
\]
%
\end{proof}

\Cref{lem:draw_no_long_monocolor} allows us to show that if $A$ and $B$ are sequences containing randomly selected elements, the sequence obtained by simultaneously inserting element $x$ of $B$ into $A$ in the position returned by \noisysearch is likely to have a dislocation that is at most a constant factor larger than the dislocation of $A$. This is formalized as follows.

\begin{lemma}\label{lem:merge_constant_disl_increase}
    Let $A$ be a sequence of $m \ge \max\{64, \sqrt{n}\}$ elements chosen u.a.r.\ from $S$, 
    and let $B$ be a set of at most $m$ elements chosen u.a.r.\ from $S \setminus A$ (both without replacement).
    Let $d$ be an integer such that  $\max\{\disl(A), \log n\}\leq d \le \frac{\sqrt{n}}{2160 \alpha}$ and, for each $x \in B$,
    let $r_x \in \{1, \dots,  m \}$ be an integer that satisfies $|r_x - \rank(x, A)| \le \alpha d$.
    The sequence $\widetilde{S}$ obtained from simultaneously inserting each element $x \in B$ into position $r_x$ of $A$ (breaking ties arbitrarily) has a maximum dislocation of at most $\gamma d$ with probability at least $1-O(n^{-3})$. 
\end{lemma}
\begin{proof}
Let $S^*$ be the correctly sorted sequence of the elements in $S$
and observe that, since the elements in $A$ and $B$ are selected u.a.r.\ from $S$, we can relate their distribution in $S^*$ with the distribution of the drawn balls in the urn experiment of Lemma~\ref{lem:draw_no_long_monocolor}: 
the urn contains $N=|A|+|B|$ balls each corresponding to an element in $A \cup B$, a ball is white if it corresponds to one of the $M=m$ elements of $A$, and black if it corresponds to one of the at most $m$ elements of $B$.	

    Then, by Lemma~\ref{lem:draw_no_long_monocolor} with $\ell= 2 \lfloor 2 \alpha d \rfloor +2$, we have that all contiguous subsequences of $S^*$ with least $54 \ell$ elements contain more than $\ell$ elements from $A$ with probability at least $1- m^{-6} \ge 1 - n^{-3}$.
    Hence, with the aforementioned probability, all contiguous subsequences of $S^*$ containing up to $\ell$ elements from $A$ have a length of at most $54 \ell \le 216 \alpha d + 108$. In the rest of the proof we assume that this is the case, and we argue that this implies that the dislocation of $\widetilde{S}$ is at most $226 \alpha d = \gamma d$.

    For an element $y \in A \cup B$, let $X_y$ denote the set of all elements $x \in B$ such that $|\rank(x, A) - \rank(y, A)| \le 2\alpha d$.
    We separately study the dislocation $\disl(y, \widetilde{S})$ for $y \in A$ and for $y \in B$.

    We consider the case $y \in A$ first.
    Observe that inserting element $x \in B$ can affect the dislocation of $y$ only if one of the following two (mutually exclusive) conditions holds: (i) $x < y$ and $r_x \ge \pos(y, A)$, or (ii) $x > y$ and $r_x \le \pos(y, A)$. Indeed, all other elements in $B$ are guaranteed to be in the correct relative order w.r.t.\ $y$ in $\widetilde{S}$.
    
    \noindent If (i) holds, we have:
    \[
        \rank(y, A) \ge \rank(x, A) \ge r_x - \alpha d \ge \pos(y, A) - \alpha d \ge \rank(y, A) - (\alpha +1) d.
    \]
    If (ii) holds, we have:
    \[
        \rank(y, A) \le \rank(x, A) \le r_x + \alpha d \le \pos(y, A) + \alpha d \le \rank(y, A) + (\alpha +1) d. 
    \]

    Hence, all the elements $x$ that can affect the dislocation of $y$ following their insertion are contained in $X_y$,
    and we have $\disl(y, \widetilde{S}) \le d + |X_y|$.

    We now consider the case $y \in B$.
    Observe that if were to insert only element $y$ in position $r_y$ of sequence $A$, the dislocation of $y$ in the resulting sequence would be upper bounded by $\alpha d$. By simultaneously inserting all elements in $B$ into $A$ (with arbitrary tie-breaking), the above upper bound on the dislocation of $y$ can only be affected by the elements in $x \in B \setminus \{y\}$ that satisfy one of the following two (mutually exclusive) conditions: (i) $x < y$ and $r_x \ge r_y$, or (ii) $x > y$ and $r_x \le r_y$.

    \noindent If (i) holds, we have:
    \[
        \rank(y, A) \ge \rank(x, A) \ge  r_x - \alpha d \ge r_y - \alpha d \ge \rank(y, A) - 2\alpha d.
    \]
    If (ii) holds, we have:
    \[
        \rank(y, A) \le \rank(x, A) \le  r_x + \alpha d \le r_y + \alpha d \le \rank(y, A) + 2\alpha d.
    \]

    Hence, we have $\disl(y, \widetilde{S}) \le \alpha d + |X_y|$. Then, regardless of whether $y \in A$ or $y \in B$,  $\disl(y, \widetilde{S}) \le \alpha d + |X_y|$ and we can focus on upper bounding $|X_y|$.

    Let $z_y^-$ be the element of $A$ with $\rank(z_y^-, A) = \rank(y,A) - \lfloor 2\alpha d \rfloor -1$; if no such element exists, then let $z_y^-$ be the first element of $S^*$.
    The above choice of $z_y^-$ ensures that $z_y^- \le x$ for all $x \in X_y$. This is trivial if $z_y^-$ is the first element of $S^*$ while, in the complementary case, it follows from $\rank(z_y^-, A) < \rank(y,A) - 2\alpha d  \le \rank(x,A)$, where we the last inequality is due to the definition of $X_y$.
    Symmetrically, let $z_y^+$ be the element of $A$ with $\rank(z_y^+, A) = \rank(y,A) + \lfloor 2\alpha d \rfloor$ or, if no such element exists, let $z_y^+$ be the last element of $S^*$. We have $x \le z_y^+$ for all $x \in X_y$. Indeed, this is trivial when $z_y^+$ is last element of $S^*$, and it follows from $z^+_y \in A$ and $\rank(x, A) \le \rank(y,A) + \lfloor 2\alpha d \rfloor = \rank(z^+_y, A)$ in the complementary case.

    Then, since all $x \in X_y$ satisfy $z_y^- \le x \le z_y^+$, we have that all elements in $X_y$ must appear in the contiguous subsequence $\overline{S}$ of $S^*$ having $z_y^-$ and $z_y^+$ as its endpoints. Such a sequence $\overline{S}$ contains at most  $2 \lfloor 2 \alpha d \rfloor +2$ elements of $A$, hence our assumption guarantees that $\overline{S}$ contains at most $216 \alpha d + 108$ elements.
    We conclude that the dislocation of $y$ in $\widetilde{S}$ is at most $\alpha d + |X_y| \le \alpha d + |\overline{S}| \le 217 \alpha d + 108 \le 226 \alpha d = \gamma d$, where we used $\alpha d \ge 12$ since $\alpha \ge 2$ and $d \ge \log n \ge \log m \ge 6$. 
\end{proof}

We can now apply \Cref{lem:merge_constant_disl_increase}, \Cref{thm:noisy_binary_search}, and \Cref{thm:basketsort} together to analyze the maximum dislocation and the total dislocation of the sequence returned by \rifflesort.

\begin{lemma}\label{lem:rifflesort_dislocation}
    Let $n$ be larger than some sufficiently large constant (depending on $p$ and $q$).
With probability at least $1- O\left( \frac{\log n}{n \sqrt{n}} \right)$, the maximum dislocation and the total dislocation of the sequence returned by \rifflesort are at most $\bmax\cdot \log n$ and at most $\btot\cdot n$, respectively.
\end{lemma}
\begin{proof}

We assume that $\frac{\sqrt{n}}{\ln n} \ge 2160 \alpha \bmax$, which is satsified by all values of $n$ that are larger than some constant depending on $\alpha$ and $\bmax$, both of which, in turn, depend only on $p$ and $q$.

	
For $i = 1, \dots, k$, we say that the $i$-th iteration of \Cref{alg:rifflesort} is \emph{good} if the sequence $S_i$ computed at the end of the $i$-th iteration has a maximum dislocation of at most $\bmax\cdot \ln n$ and a total dislocation of at most $\btot \cdot n$.
    As a corner case, we say that iteration $0$ is good if the above conditions are satisfied by the sequence $S_0$ computed in line~\ref{ln:rf_basketsort_S0} of \Cref{alg:rifflesort}.

    Since $S_0$ is the sequence returned by invoking $\basketsort$ on $\widetilde{S}_0$ using the trivial upper bound of $|\widetilde{S}_0|$ on $\disl(\widetilde{S}_0)$,  \Cref{thm:basketsort} ensures that $S_0$ is good with probability at least $1-O(n^{-1.5})$.

    We now focus on a generic iteration $i\ge 1$ and show that, assuming that iteration $i-1$ is good, iteration $i$ is also good with probability at least $1-O(n^{-1.5})$.
    Since iteration $i-1$ is good, we have $\disl(S_{i-1})\leq \bmax \cdot \ln n$. Moreover, for each element $x \in T_i$, we have that the comparison errors between $x$ and the elements in $S_{i-1}$ do not depend on the specific elements in $S_{i-1}$ nor on their order. Then, \Cref{thm:noisy_binary_search} together with the union bound ensure that all the approximate ranks $r_x$, for $x \in T_i$, computed by the calls to $\noisysearch(S_{i-1}, x, \bmax \cdot \ln n)$ (see line~\ref{ln:rf_noisysearch}) satisfy $|r_x - \rank(x, S_{i-1})| \le  \alpha \bmax \cdot \ln n$ with probability at least $1 - O( |T_i| \cdot |S_{i-1}|^{-6} ) = 1 - O( n^{-2})$. Whenever this happens, \Cref{lem:merge_constant_disl_increase} with $d = \bmax \cdot \ln n$ shows that the sequence $\widetilde{S}_i$ computed in line~\ref{ln:rf_simultaneous_insert} has a maximum dislocation of at most $\gamma \cdot \bmax \cdot \ln n$ with probability at least $1-O(n^{-3})$.
    If such an upper bound on the maximum dislocation of $\widetilde{S}_i$ does indeed hold, \Cref{thm:basketsort} further implies that the call to $\basketsort(\widetilde{S}_i, \gamma \cdot  \bmax \cdot \ln n)$ in the $i$-th iteration (see line~\ref{ln:rf_basketsort_Si}) returns a sequence $S_i$ having a maximum (resp.\ total) dislocation of at most $\bmax\cdot \ln |\widetilde{S}_i| \le \bmax \cdot \ln n$ (resp.\ $\btot |\widetilde{S}_i| \le \btot n$) with probability at least $1-O\Paren{\frac{1}{|\widetilde{S}_i|^{-3}}} =  1 - O\left(n^{-1.5}\right)$, i.e., that iteration $i$ is good.
    By using the union bound, we have that all the required conditions for iteration $i$ to be good occur with probability at least $1 - O(n^{-2}) - O(n^{-3}) - O(n^{-1.5}) = 1 - O(n^{-1.5})$.


    Since there are $k = O(\log n)$ iterations, yet another application of the union bound, implies that all iterations are good with a probability of at least $1 - O\left( \frac{\log n}{n \sqrt{n}} \right)$.
    The claim follows from the fact that the $k$-th iteration is good with the above probability. 
\end{proof}

To conclude, \Cref{lem:rifflesort_runtime} and \Cref{lem:rifflesort_dislocation} directly lead to the main result of this section.

\begin{theorem}\label{thm:rifflesort}
    Consider a set of $n$ elements subject to random persistent comparison faults with probability at least $q$ and most $p$, for some constants $p,q$ with $0 \le q \le p < \frac{9q+1}{8q+11}$. Let $S$ be a sequence of these elements (possibly chosen as a function of the errors).
  \rifflesort is a randomized algorithm that approximately sorts $S$ in $O(n \log n)$  worst-case time. The maximum dislocation and the total dislocation of the returned sequence are at most $O(\log n)$ 
  and at most $O(n)$, 
    respectively, with probability at least $1 - O\left( \frac{\log n}{n \sqrt{n}} \right)$.
\end{theorem}

\section{Noisy Binary Search}\label{sec:noisy-search}

In this section we tackle the problem of efficiently compute an approximate rank of an ``external'' element with respect to an approximately sorted sequence in presence of persistent random comparison faults.
More precisely, given a sequence $S = \langle s_1, \dots, s_m \rangle$ of $m$ elements, an integer $d$ that satisfies $d \geq \max\{\disl(S),\,\ln m\}$, and an element $x\not\in S$, we want to compute, in time $O(\log m)$, a value $\tau_x$ such that $\abs{\tau_x-\rank(x,\,S\cup\{x\})}=O(d)$ with high probability.

Recall from \Cref{sec:preliminaries} that only comparisons between $x$ and the elements in $S$ are allowed, that errors happen independently with probability at most $p$, and that the order of the elements in $S$ does not depend on the comparison errors.


\subsection{High-level ideas}

As mentioned in the introduction, the seminal paper of Feige et al.~\cite{FeigeRPU94} presented a noisy binary search tree for \emph{non-persistent} random comparison faults that allows to compute $\rank(x, S \cup \{x\})$ w.h.p.\ in time $O(\log n)$. 
However, this approach does not readily generalize to our model due to two main difficulties:
\begin{itemize}
\item Feige et al.'s algorithm heavily relies on the fact that comparing the same pair of elements multiple times yields independent outcomes.
\item The input sequence $S$ of \cite{FeigeRPU94} is perfectly sorted while, in our case, $S$ is only \emph{approximately sorted}. This is motivated by the fact that, when comparison faults are persistent, it is impossible to compute a correctly sorted sequence, as we discuss in Section~\ref{sec:lower_bound}.
\end{itemize}

  To circumvent the first difficulty, we ensure that $x$ is never compared twice with any element $y$ in $S$. Rather, any further comparison between $x$ and $y$ is replaced by a comparison (or multiple comparisons) between $x$ and some nearby proxy element(s) for $y$.


  To tackle the second difficulty, we use the upper bound $d$ on the maximum dislocation which implies that an element in a generic position $i$ is larger than any element in a position preceding $i-d$ and smaller than any element in a position following $i+d$. 
Then, if we group the elements into blocks of size at least $d$, all elements in $i$-th group will be guaranteed to be larger (resp.\ smaller) than all the elements in the $(i-2)$-th (resp.\ $(i+2)$-th) group. 
We can then perform two parallel noisy binary searches: one involving only even-numbered groups and one involving one odd-numbered groups. 

The next subsection describes the details of our construction. 

\subsection{Construction of Noisy Binary Search Trees}\label{sub:tree_construction}

Let $k$ be the smallest odd integer that satisfies $k \ge 32  \frac{1-p}{(1-2p)^2}$, and let $c = 250k$.
For the sake of simplicity, we assume that $m = 2cd \cdot 2^h$ where $h$ is a non-negative integer. At the end of this section, we argue that this assumption can be removed.



We start by partitioning the sequence $S$ of $m$ elements into $2\cdot 2^h$ ordered \emph{groups} $g_0, g_1, \ldots$ where group $g_i$ contains the elements in positions  $i\cdot cd+1$, \dots, $ (i+1) cd$.
Then, we further partition these $2 \cdot 2^h$ groups into two ordered sets $G_0$ and $G_1$, where $G_0$ contains the groups $g_i$ with even $i$, and $G_1$ the groups $g_i$ with odd $i$. Notice that $|G_0| =|G_1|= 2^h$. 

For each $j\in \{0,1\}$, we define a \emph{noisy binary search tree} $T_j$ on $G_j$. Let $\eta = 2 + 4 \lceil \log_{7} m \rceil$. The tree $T_j$ is a binary tree of height $h+\eta$ in which the first $h+1$ levels (i.e., those containing vertices at depths $0$ to $h$) are complete and the last $\eta$ levels consist of $2^h$ paths of $\eta$ vertices, each emanating from a distinct vertex on the $(h+1)$-th level.

We index the leaves of $T_j$ from $0$ to $2^h-1$, we use $h(v)$ to denote the depth of vertex $v$ in $T_j$,
and we refer to the vertices $v$ at depth $h(v) \ge h$ as  \emph{path-vertices}.
We also associate each vertex $v$ of $T_j$ with one \emph{interval} $I(v)$, i.e., a set of contiguous positions, as follows: for a leaf $v$ with index $i$, $I(v)$ consists of the positions in $g_{2i+j}$, for a non-leaf path-vertex $v$ having $u$ as its only child, we set $I(v)=I(u)$, and for an internal vertex $v$ having $u$ and $w$ as its left and right children, respectively, we define $I(v)$ as the interval containing all the positions between $\min I(u)$ and $\max I(w)$ (inclusive).

\begin{figure}[t]
	\centering
	\includegraphics[width=.9\textwidth]{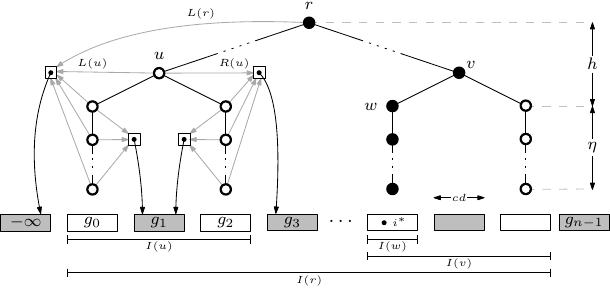}
	\caption{An example of the noisy tree $T_0$. On the left side, the shared pointers $L(\cdot)$ and $R(\cdot)$ are shown. Notice how $L(r)$ (and, in general, all the $L(\cdot)$ pointers on the leftmost side of the tree) points to the special $-\infty$ element.  Good vertices are shown in black while bad vertices are white. Notice that, since $i^* \in I(w)$, we have $T^* = T_0$ and hence all the depicted vertices are either good or bad.}
	\label{fig:noisy_tree}
\end{figure}

Moreover, each vertex $v$ of $T_j$ has a reference to two \emph{shared pointers} $L(v)$ and $R(v)$ to positions in $\{1, \dots, n\} \setminus \bigcup_{g_i \in S_j} g_i$. Intuitively, $L(v)$ (resp.\ $R(v)$) will always point to the positions of $S$ occupied by elements that are \emph{smaller} (resp.\ \emph{larger}) than all the elements $s_i$ with $i \in I(v)$.
For each leaf $v$,  $L(v)$ initially points to $\min I(v) - d - 1$ while $R(v)$ initially points to $\max I(v) + d$.
A non-leaf path-vertex $v$ shares both its pointers with the corresponding pointers of its only child, while a non-path vertex $v$ shares its left pointer $L(v)$ with the left pointer of its left child, and its right pointer $R(v)$ with the right pointer of its right child. See Figure~\ref{fig:noisy_tree} for an example.

We sometimes allow $L(v)$ to point to positions smaller than $1$ and $R(v)$ to point to positions larger than $m$. 
We consider all the elements $s_i$ with $i \le 0$ (resp.\ $i > m$) to be copies of a special $-\infty$ (resp.\ $+\infty$) element such that $-\infty < x$ and $-\infty \prec x$ in every observed comparison (resp.\ $+\infty > x$ and $+\infty \succ x$).

\subsection{Random Walks on Noisy Binary Search Trees}\label{sub:noisy_search_on_tree}
To compute an approximate rank $\tau_x$ for the query element $x$, our algorithm will perform a discrete-time random walk on each of $T_0$ and $T_1$. These random walks rely on a \emph{test} operation defined as follows:

\begin{definition}[test operation]
A \emph{test} of element $x$ with a vertex $u$ is performed by (i) comparing $x$ with the elements $s_{L(u)}, s_{L(u)-1}, \dots, s_{L(u)-k+1}$, (ii) comparing $x$ with the elements $s_{R(u)}, s_{R(u) +1}, s_{R(u)+k-1}$, (iii) decrementing $L(u)$ by $k$, and finally (iv) incrementing $R(u)$ by $k$. Vertex $u$ \emph{passes} the test if the majority result among the comparisons performed in step (i) reports $x$ as the larger element, and the majority result among the comparisons performed in step (ii) reports $x$ as the smaller element. Otherwise, $u$ \emph{fails} the test.
\end{definition}

For a fixed $j\in\{0, 1\}$, the random walk on $T_j$ proceeds as follows. 
At time $0$, i.e., before the first step, the \emph{current} vertex $v$ coincides with the root $r$ of $T_j$.
Then, at each time step, $x$ is tested with all the children of the current vertex $v$, and we \emph{walk} from $v$ to the next vertex (possibly $v$ itself) according to the following rules:
\begin{itemize}
	\item If \emph{exactly one} child $u$ of $v$ \emph{passes} the test, we walk from $v$ to $u$.
	\item If \emph{all} the children of $v$ \emph{fail} the test, we walk from $v$ to its parent, if $v$'s parent exists (if the parent of $v$ does not exist, then $v=r$, and we walk to $r$ itself);
	\item If \emph{more than one} child of $v$ \emph{passes} the test, we walk from $v$ to itself.
\end{itemize}

Fix an upper bound  $\rho = 120 \lceil \ln m \rceil$ on the total number of steps. 
The walk stops as soon as one of the following two conditions is met: 
\begin{description}
	\item[Success:] The current vertex $v$ is a leaf of $T_j$, in which case we say that the random walk returns $v$;
	\item[Timeout:] The $\rho$-th time step is completed and the success condition is not met.
\end{description}

\subsection{Analysis}

We now analyze the behavior of our random walks and establish the main result of this section.
The key ingredient of the analysis is to show that one of the two random walks, i.e., the one whose noisy binary search tree contains the leaf that corresponds to the true rank of $x$ in $S$, is likely to succeed, and that any walk that succeeds returns either such a leaf, or a nearby one.

Let $i^* = \rank(x, S \cup \{x\}) = \rank(x, S)$, let $T^*$ be the unique tree in $\{ 
T_0, T_1 \}$ such that $i^*$ belongs to the interval of a leaf in $T^*$, and let $T'$ be the other tree.

\begin{definition}[good/bad vertex]
A vertex  $v$ of $T^*$ is \emph{good} if $i^* \in I(v)$, it is \emph{bad} if either $i^* < \min I(v) - cd$ or $i^* > \max I(v) + cd$, and it is \emph{neutral} if it is neither good nor bad.
\end{definition}

Notice that all the vertices in $T^*$ are either good or bad, that the intervals corresponding to the vertices in $T^*$ at the same depth are pairwise disjoint, and that the set of good vertices in $T^*$ is exactly a root-to-leaf path in $T^*$.
Furthermore, all path-vertices of $T'$ are either bad or neutral.
Moreover, in both $T'$ and $T^*$, all the children of a bad vertex are also bad.
The following two lemmas provide bounds on the probability that a good/bad vertex passes a test. 

\begin{lemma}
	\label{lemma:test_good}
	Consider a test performed on a good vertex $u$. The probability that $u$ fails the test is at most $\frac{2}{50}$.
\end{lemma}
\begin{proof}
For $i=0, \dots, k-1$, let $X_i$ (resp.\ $Y_i$) be an indicator random variable that is $1$ iff the comparison between $x$ and $s_{L(u)-i}$ reports $x$ as the larger element (resp.\ the comparison between $x$ and $s_{R(u)+i}$ reports $x$ as the smaller element). Let $X = \sum_{i=0}^{k-1} X_i$ and $Y = \sum_{i=0}^{k-1} Y_i$ and notice that for a test to fail at least one of $X > \frac{k}{2}$ and $Y > \frac{k}{2}$ must hold.

Consider the comparison between $x$ and $s_{L(u)-i}$ for some $i=0, \dots, k-1$. Since $i^* \in I(u)$ and the pointer $L(u)$ only gets decremented, we  have $L(u) - i \le L(u) \le \min I(u) - d - 1 \le i^* -d - 1$. This, combined with the fact that $S$ has dislocation at most $d$, implies  $\rank(s_{L(u)-i}, S) \leq L(u) - i + d \le  i^* - 1$ and hence $s_{L(v)-i} < x$.

We conclude that the probability to observe $S_{L(v) - i} \prec x$  is at least $1-p$, and $X$ stochastically dominates a binomial random variable $X'$ of parameters $k$ and $1-p$. Thus, we can use a Chernoff bound to write:
\begin{align*}
    \Pr\left( X \le  \frac{k}{2} \right)
    &\le \Pr\left( X' \le \left(1- \frac{1-2p}{2-2p} \right) \cdot \E[X'] \right)
    \le \exp\left( - \frac{1}{2} \cdot \left( \frac{1-2p}{2-2p}\right)^2 \cdot k(1-p)  \right) \\
    & \le e^{-4} < \frac{1}{50}.
\end{align*}
A similar argument shows that $\Pr(Y \le \frac{k}{2}) < \frac{1}{50}$ once we observe that $R(u) + i \ge \max I(u) + d \ge i^* + d$, which implies $\rank(s_{R(u)+i}, S) \ge i^*$ and hence $s_{R(u)+i} > x$.
 The claim follows by using the union bound on the events $X \le \frac{k}{2}$ and $Y \le \frac{k}{2}$.
\end{proof}

\begin{lemma}
	\label{lemma:test_bad}
	Consider a test performed on a bad vertex $u$. The probability that $u$ passes the test is at most $\frac{1}{50}$.
\end{lemma}
\begin{proof}
Since $u$ is a bad vertex, either $i^* < \min I(u) - cd$ or $i^* > \max I(u) - cd$. We consider the first case, as the second one is symmetric.

Let $X_i$ be an indicator random variable that is $1$ iff the comparison between $x$ and $s_{L(u)-i}$ reports $x$ as the smaller element and define $X = \sum_{i=0}^{k-i} X_i$. A necessary condition for the test to succeed is $X \le \frac{k}{2}$.

Focus on the comparison between $x$ and $s_{L(u)-i}$ for some $i=0, \dots, k-1$ and notice that, before the considered test is performed, the pointer $L(u)$ must have been decremented (by $k$) at most $\rho-1$ times. Using $k\rho = 120 k \lceil \ln m \rceil < (c-2) \lceil \ln m \rceil \le (c-2) d$, we can write:
\begin{equation*}
L(u) - i \ge
\min I(u) - d - 1  - k (\rho - 1) - i \ge
\min I(u) - d  -  k \rho  >
i^* + cd  - d - (c-2)d =
i^* + d.
\end{equation*}

Since $S$ has dislocation at most $d$, this implies  $\rank(s_{L(u)-i}, S) \geq L(u) -i -d > i^*$. Therefore, $s_{L(u) - i} > x$ and $s_{L(u)} \succ x$ is observed with probability at least $1-p$. Hence, $X$ stochastically dominates a binomial random variable $X'$ of parameters $k$ and $1-p$, and we can use a Chernoff bound to write:
\begin{equation}
\Pr\left(X \le \frac{k}{2}\right) \le \Pr\left( X \le \left( 1 - \frac{1-2p}{2-2p} \right) \E[X'] \right)
    \le \exp\!\left( - \frac{1}{2} \cdot \left( \frac{1-2p}{2-2p}\right)^2 \! \cdot kp  \right)
    \le e^{-4} < \frac{1}{50}. \tag*{\qedhere}
\end{equation}
\end{proof}

\begin{definition}\label{def:improving_step}
A step on $T_j$ from vertex $v$ to vertex $u$ is called \emph{improving} if either of the following two conditions holds:
	\begin{itemize}
		\item $u$ is a good vertex and $h(u) > h(v)$, implying that $v$ is also good; or
		\item $v$ is a bad vertex and $h(u) < h(v)$.
	\end{itemize}
\end{definition}

Intuitively, each improving step is making progress towards identifying the interval containing the true rank $i^*$ of $x$, while each non-improving step undoes the progress of at most one improving step.

\begin{lemma}
	\label{lemma:improving_step}
	Each step performed during the walk on $T^*$ is improving with probability at least $\frac{9}{10}$.
\end{lemma}
\begin{proof}
	Consider a generic step performed during the walk on $T^*$ from a vertex $v$.
	If $v$ is a good vertex, then exactly one child $u$ of $v$ in $T^*$ is good (notice that $v$ cannot be a leaf). 
  By Lemma~\ref{lemma:test_good}, the test on $u$ succeeds with probability at least $\frac{48}{50}$. Moreover, there can be at most one other child $w \neq u$ of $v$ in $T^*$. If $w$ exists, then it must be bad and, by Lemma~\ref{lemma:test_bad}, the test on $w$ fails with probability at least $\frac{49}{50}$.
  By using the union bound on the complementary probabilities, we have that process walks from $v$ to $u$ with probability at least $\frac{47}{50}$.
	
  If $v$ is a bad vertex, then all of its children are also bad. Since $v$ has at most $2$ children and, by Lemma~\ref{lemma:test_bad}, a test on a bad vertex fails with probability at least $\frac{49}{50}$, we have that all the tests fail with probability at least $\frac{48}{50}$. In this case, the process walks from $v$ to the parent of $v$ (notice that, since $v$ is bad, it cannot be the root of $T^*$).
\end{proof}

As a consequence of our choice for the group size, no two tests will ever compare $x$ with the same element from $S$, as the following lemma shows.
\begin{lemma}
	\label{lemma:independent_comparisons}
	During the random walk on $T_j$, element $x$ is compared to each element in $S$ at most once.
\end{lemma}
\begin{proof}
During the random walk, each pointer moves from its initial position by at most $k\rho$ positions (since each of the at most $\rho$ steps of the random walk either does not affect the pointer or it increases/decreases the pointer by $k$). 
However, the initial distance between any two distinct pointers is more than $2k\rho$, indeed: 
\[
cd - 2d - 1 \ge (c-3)d = (250k - 3) d > 240 kd = 2 k \cdot 120 d \ge 2 k\cdot 120 \lceil \ln m \rceil = 2k \rho,
\] where we used the fact that $d$ is an interger with $d \ge \ln m$, which implies $d \ge \lceil \ln m\rceil$.
The claim follows by observing that $x$ is only compared to the elements in $S$ via test operations.
\end{proof}

The following lemmas show that the walk on $T^*$ is likely to return a good vertex, while the walk on $T' \neq T^*$ is likely to either timeout or to return a non-bad (good or neutral) vertex, i.e., a vertex whose corresponding interval contains positions that are close to the true rank of $x$ in $S$.

\begin{lemma}
	\label{lemma:timeout}
	The walk on $T^*$ timeouts with probability at most $m^{-6}$.
\end{lemma}
\begin{proof}
	For $t=1,\dots, \rho$, let $X_t$ be an indicator random variable that is equal to $1$ iff the $t$-th step of the walk on $T^*$ is improving. If the $t$-th step is not performed then let $X_t = 1$.
	
	Notice that if, at any time $t'$ during the walk, the number $X^{(t')} = \sum_{t=1}^{t'} X_t$ of improving steps exceeds the number of non-improving steps by at least $h + \eta$, then the success condition is met.
	This means that a necessary condition for the walk to timeout is $X^{(\rho)} - (\rho - X^{(\rho)}) < h + \eta$, which is equivalent to $X^{(\rho)}  < \frac{h + \eta + \rho}{2}$.
	
	By Lemma~\ref{lemma:independent_comparisons} and by Lemma~\ref{lemma:improving_step}, we know that 
	each $X_t$s corresponding to a performed step satisfies $P(X_t = 1) \ge \frac{9}{10}$, regardless of whether the other steps are improving.
	We can therefore consider the following experiment: at every time step $t=1,\dots, \rho$, we flip a coin that lands on heads with probability $q = \frac{9}{10}$, we let $Y_t = 1$ if this happens, and $Y_t = 0$ otherwise.
	Defining $Y^{(\rho)}  = \sum_{t=1}^{\rho} Y_t$, we have that the probability of having $X^{(\rho)}  < \frac{h + \eta + \rho}{2}$ is at most the probability of having $Y^{(\rho)}  < \frac{h + \eta + \rho}{2}$. 
	Since $h < \log m < 2 \ln m$, we have $h + \eta <  2 +  4 \lceil \log_{7} m \rceil + 2 \ln m \le 8 \lceil \ln m \rceil < \frac{\rho}{10}$, and we can use a Chernoff bound to write:
	\begin{align*}
		\Pr\left(X^{(\rho)}  < \frac{h + \eta + \rho}{2}\right) & \le \Pr\left(Y^{(\rho)} < \frac{h + \eta + \rho}{2}\right) \le \Pr\left(Y^{(\rho)} < \frac{11\rho}{20}\right)\\ 
		 & = \Pr\left(Y < \frac{11 \E[Y]}{20q} \right) \le \Pr\left(Y < \left(1- \frac{1}{3}\right) \E[Y] \right)\\
		 & \le e^{-\frac{\rho q}{18}} =  e^{-\frac{\rho}{20}} \le e^{-6 \ln m} = m^{-6}. \tag*{\qedhere}
	\end{align*}
\end{proof}

\begin{lemma}
	\label{lemma:bad_vertex_returned}
  Let $j \in \{0,1\}$. A walk on $T_j$ returns a bad vertex with probability at most $m^{-7}$.
\end{lemma}
\begin{proof}
	Notice that, in order to return a bad vertex $v$, the walk must
	first reach a vertex $u$ at depth $h(u)=h$, and then traverse the $\eta$ vertices of the path rooted in $u$ having $v$ as its other endpoint.
	
	We now bound the probability that, once the walk reaches $u$, it will also reach $v$ before walking back to the parent of $u$.
	Notice that all the vertices in the path from $u$ to $v$ are associated to the same interval, and hence they are all bad. 
  Since a test on a bad vertex succeeds with probability at most $\frac{1}{50}$ (see Lemma~\ref{lemma:test_bad}) and tests are independent (see Lemma~\ref{lemma:independent_comparisons}), the sought probability can be upper-bounded by considering a random walk on $\{0 ,\dots, \eta+1 \}$ that: (i) starts from $1$, (ii) has one absorbing barrier on $0$ and another on $\eta+1$, and (iii) for any state $i \in [1, \eta]$ has a probability of transitioning to state $i+1$ of $\frac{1}{50}$ and to state $i-1$ of $\frac{49}{50}$.
	Here state $0$ corresponds to the parent of $u$, and state $i$ for $i>0$ corresponds to the vertex of the $u$--$v$ path at depth $h+i-1$ (so that state $1$ corresponds to $u$ and state $\eta+1$ corresponds to $v$).
	
	The probability of reaching $v$ in $\rho$ steps is at most the probability of being absorbed in $\eta$ (in any number of steps), which is (see, e.g., \cite[pp.~344--346]{feller1957introduction}):
	\[
		\frac{  \frac{49/50}{1/50} - 1 }{ \left(\frac{49/50}{1/50}\right)^{\eta+1} - 1 } =
\frac{48}{49^{\eta+1} -1}
< \frac{1}{49^\eta} 
  \le \frac{1}{7^4 \cdot 7^{2 \cdot 4 \log_7 m}}
= \frac{1}{7^4\, m \cdot m^7}
< \frac{1}{\rho m^7}.
	\]
	
	
	Since the walk on $T_j$ can reach a vertex at depth $h$ at most $\rho$ times, by the union bound we have that the probability of returning a bad vertex is at most $m^{-7}$.
\end{proof}

\noindent We are now ready to prove the main result of this section.

\thmnoisysearch
\begin{proof}
	We compute the index $\tau_x$ by performing two random walks on $T_0$ and on $T_1$, respectively.
	If any of the walks returns a vertex $v$, then we return any position in the interval $I(v)$ associated with $v$ (if both walks return some vertex, we choose one arbitrarily). If both walks timeout, then we return an arbitrary position.
	
  From Lemma~\ref{lemma:timeout} the probability that both walks timeout is at most $m^{-6}$ (as the walk on $T^*$ timeouts with at most this probability).
  Moreover, by Lemma~\ref{lemma:bad_vertex_returned}, the probability that at least one of the two walks returns a bad vertex is at most $2m^{-7} < m^{-6}$.
	
  Overall, vertex $v$ exists and it is not bad with probability at least $1 - 2m^{-6}$. When this happens we can use $\tau_x \in [\min I(v), \max I(v)]$ and that $\max I(v) - \min I(v) < cd$ to write:
	\[
	i^* \ge \min I(v) - cd  > \max I(v) - 2cd \ge \tau_x - 2cd,
	\]
	and
	\[
	i^* \le \max I(v) + cd < \min I(v) + 2cd \le \tau_x + 2cd.
	\]
	
	To conclude the proof, it suffices to notice that the random walk requires at most $\rho = O(\log n)$ steps, that each step requires constant time, and that it is not necessary to explicitly construct $T_0$ and $T_1$ beforehand.
	Instead, it suffices to maintain a partial tree consisting of all the vertices visited by the random walk so far: vertices (and the corresponding pointers) are \emph{lazily} created and appended to the existing tree whenever the walk visits them for the first time.
\end{proof}

We conclude this section by arguing that our initial assumption that $m = 2cd \cdot 2^h$ for some integer $h \ge 1$ can be relaxed.
Indeed, if this is not already the case and $m<4cd$, we can simply return an arbitrary position $\tau_x \in [1, m+1]$ to trivially satisfy $|\tau_x - \rank(x,S)| \le 4c d$.
If $m > 4cd$, we can satisfy the requirement by finding an estimate $\tau'_x$ of the rank of $x$ in a padded version $S'$ of $S$ using a suitable upper bound $d'$ on the dislocation of $S'$.
We build $S'$ by appending up to $m-1$ dummy elements which always compare larger than $x$, and we choose $d' = 2d$ to ensure 
$d' = 2d \ge 2\max\{ \disl(S), \ln m \} \ge \max\{ \disl(S'), \ln m^2 \} > \max\{ \disl(S'), \ln |S'| \}$, where we used $\disl(S) = \disl(S')$.
After finding $\tau'_x$, we return $\tau_x = \min\{\tau'_x, m+1\}$, which does not increase the error of the estimated rank.

We remark that Theorem~\ref{thm:noisy_binary_search} as stated works regardless of our assumption. In fact, by inspecting its proof and taking into account the above discussion, we can give a (crude) general upper bound for the value of $\alpha$ as
$2 \cdot (2 c) = 1000 k$ (where the extra factor of $2$ accounts for the possibly doubled value of $d$).
Since $k$ is defined as the smallest odd integer that is at least $32  \frac{1-p}{(1-2p)^2}$ and $ \frac{1-p}{(1-2p)^2} \ge 1$, we have $k < 34 \frac{1-p}{(1-2p)^2}$ from which we obtain $\alpha < 34000 \frac{1-p}{(1-2p)^2}$.

\section{BasketSort}\label{sec:baskset-sort}


In this section, we describe and analyze a sorting algorithm called \emph{\basketsort} that re-sort an approximately sorted sequence into a sequence having logarithmic maximum dislocation and linear total dislocation, w.h.p. 
At a high level, \basketsort maintains a \emph{window size} as a tentative upper bound on the dislocation of each element.
Then, \basketsort iteratively shrinks the window size by a multiplicative \emph{shrinking factor} while preserving such an upper-bound until the window size becomes $O(\log n)$.
In the shrinking process, \basketsort partitions the current sequence into \emph{baskets} containing as many elements as the current window size, compares the elements in each basket with the elements in the neighboring baskets, and re-orders the elements according to the outcomes of those comparisons.

Our algorithm works when the input elements are subject to persistent random comparison errors with a probability of error of at least $q$ and at most $p$ with $p \le \frac{9q+1}{8q+11}$. 
The rest of this section is organized as follows. In \Cref{sub:basket-description} we describe our algorithm, while in \Cref{sub:basket-maximum-dislocation} and \Cref{sub:basket-total-dislocation} we prove that the maximum dislocation and the total dislocation of the sequence returned by \basketsort are logarithmic and linear w.h.p., respectively. 

\subsection{Algorithm description}\label{sub:basket-description}

The input to the algorithm is a sequence $S$ of $m$ distinct elements and an upper bound $w_S$ on the dislocation of $S$, i.e., $w_S\geq \disl(S)$.
We assume w.l.o.g.\ that $w_S \le m$ (since otherwise we can set $w_S = m$) and that $w_S \ge 1$ (since otherwise $S$ is already sorted).
To simplify the notation, we further assume w.l.o.g.\ that $S$ contains the integers between $1$ and $m$ (so that $\rank(x,S)=x$).

The algorithm uses a constant \emph{shrinking rate} parameter $\rho \in [\frac{1}{2}, 1)$ that is chosen as a function of the lower and upper bounds $q$ and $p$ on the probability of comparison errors.
Although we set $\rho=\frac{1}{2} \Paren{1+\frac{8pq + 10(p-q)}{1-p-q}}=\frac{1}{2} + \frac{4pq + 5(p-q)}{1-p-q}$ in \Cref{alg:bs}, $\rho$ only has to be a parameter in $\Paren{\frac{8pq + 10(p-q)}{1-p-q},1}$, and  we still express the achieved maximum and total dislocations, the success probability, and the running time of the algorithm as functions of $\rho$ throughout our analysis.
Notice that the above interval is non-empty as long as $p< \frac{9q+1}{8q + 11}$. For $q=0$, the condition becomes $p < \frac{1}{11}$, while $p=q$ yields $p < \frac{1}{4}$.

The pseudocode of our algorithm is shown in \Cref{alg:bs}. 
The algorithm iteratively ``sorts'' the sequence of elements while maintaining a (tentative) upper bound $w$, called the \emph{window size}, to their dislocation.
The algorithm terminates in $O(\log_{1/\rho} w_S)$ rounds, where the $i$-th round considers a window size $w$ of roughly $\rho^{i-1} w_S$ and permutes an associated sequence $S_w$ to obtain a new sequence $S_{\lfloor \rho w \rfloor}$ with a (tentative) upper bound on the maximum dislocation of $\rho w$. The initial sequence $S_{\lfloor w_S \rfloor}$ is exactly $S$, while the output of the algorithm is the sequence $S_0$.

To compute $S_{\lfloor \rho w \rfloor}$, the algorithm uses $S_w$ to find an estimate $\tau_w(x)$ of the true rank of each element $x \in S$, and then sorts the elements in non-decreasing order of $\tau_w(\cdot)$.
The quantity $\tau_w(x)$ is obtained as follows: we partition $S_w$ into consecutive \emph{baskets} $B_1, B_2, \dots$ of $w$ elements (except possibly for the last basket when $m$ is not a multiple of $w$) 
and we restrict our attention to the union $B$ of the elements that are ``at most $3$ baskets away'' from the baskets $B_i$ containing $x$.
Note that if the dislocation of $S_w$ is at most $w$, then all elements not in $B$ are in the correct relative order w.r.t.\ $x$ and $\rank(x, S) = \max\{(i-4)w, 0\} + \rank(x, B)$, so we are left with the problem of estimating $\rank(x, B)$. 
We say that $x\in B$ has a \emph{score} $score(x)$ equal to the number of elements $y\in B\setminus\{x\}$ with $y\prec x$, we sort $B$ according to the scores of the elements to obtain a sequence $A$, and we use the position of $x$ in $A$ as a proxy for $\rank(x, B)$, i.e., we define $\tau_w(x) = \max\{0,i-4\}\cdot w + \pos(x,A)$.

Clearly, we cannot hope for the tentative upper bound $w$ on the dislocation of $S_w$ to hold with high probability until $w=0$. However, as we will show in \Cref{sub:basket-maximum-dislocation}, $\disl(S_w) \le w$ holds with high probability for $w = \Omega(\log m)$, and the combined effect of all subsequent rounds only affects the dislocation by at most a multiplicative constant.

\savenotes
\begin{algorithm2e}[t]
\caption{\basketsort\unskip($S$,\,$w_S$)}\label[algorithm]{alg:bs}
\small
    $\rho \gets  \frac{1}{2} + \frac{4pq + 5(p-q)}{1-p-q}$\tcp*{$\rho$ is a \emph{shrinking rate} in $\left(\frac{8pq + 10(p-q)}{1-p-q}, 1\right)$ (see \Cref{sub:basket-description})}
  $S_{\lfloor w_S \rfloor} \gets S$\;
  $w \gets \lfloor w_S \rfloor$\;
\While{$w \ge 1$}
{
  $B_1, \dots, B_{\lceil \frac{m}{w}\rceil} \gets$ Partition $S_w$ into $\lceil \frac{m}{w}\rceil$  \emph{baskets} of consecutive elements such that each of the first  $\lceil \frac{m}{w}\rceil-1$ baskets contains exactly $w$ elements and the last basket contains at most $w$ elements\;
    \For{$i=1, \dots, \lceil \frac{m}{w}\rceil$\label{ln:bs_for}}
  {
		$B \gets  \bigcup_{j = \max\{1,i-3\}}^{\min\{i+3,\lceil m/w\rceil\}}  B_j$\;
        \lForEach{$x \in B$}{ $\score(x) \gets | \{ y \in B \setminus{x} \mid y \prec x \} |$}        
         
        $A \gets $Sort\footnote{\label{fn:stable}We assume that the used sorting algorithm is stable.} the elements in $B$ in non-decreasing order of $\score(\cdot)$\;

		\lForEach{$x \in B_i$}{$\tau_w(x) \gets \max\{0,i-4\}\cdot w + \pos(x,A)$\label{ln:compute_tau}}
  }
  $S_{\lfloor \rho w \rfloor} \gets$ Sort\textsuperscript{\ref{fn:stable}} the elements in $S_w$ in non-decreasing order of $\tau_w(\cdot)$\;
  $w \gets \lfloor \rho w \rfloor$\;
}
  \Return $S_0$\;
\end{algorithm2e}

\noindent We start by analyzing the running time of \basketsort.\spewnotes

\begin{lemma}\label{lem:basketsort-time}
  The running time of \basketsort is $O(\frac{m\cdot w_S}{1-\rho})$.
\end{lemma}
\begin{proof}
    When the window size is $w$, each iteration of the for loop at line~\ref{ln:bs_for} of \Cref{alg:bs} requires time $O(w^2)$. Since there are $\lceil \frac{m}{w}\rceil$ iterations, the overall time spent for window size $w$ is $O(mw)$. 
    The claim follows by summing the above complexity over all the $O(\log_{1/\rho} w_S)$ considered window sizes, which decrease geometrically.
  In particular, the window size considered in the $i$-th iteration is at most $\rho^{i-1} w_S$, and we can upper bound the overall running time with
  $\sum_{i=1}^{+\infty} O(m \rho^{i-1} w_S)=O(\frac{m\cdot w_S}{1 - \rho})$.
   \end{proof}

In the rest of this section, we use $w$ to denote a window size considered by the algorithm and $\sigma_w(x)$ as a shorthand for $\pos(x, S_w)$. The following three lemmas relate the three quantities $\sigma_w(\cdot)$, $\tau_w(\cdot)$ and $\sigma_{\lfloor \rho w \rfloor}(\cdot)$.

\begin{lemma}\label{lem:diff-computed_rank-initial}
	For all elements $x$,	$\lvert \tau_{w}(x)-\sigma_{w}(x)\rvert< 4w$.
\end{lemma}
\begin{proof}
Fix an element $x$ and let $B_i$ be the basket containing $x$.
We have $\tau_w(x)=\max\{0, i-4\} \cdot w + \pos(x, A)$, where $\pos(x, A)$ is the position of $x$ in $A$.
Since $A$ contains at most $\min\{i+3, 7\} \cdot w$ elements, we have:
\[
    (i-4) \cdot w < \tau_w(x) \le \max\{0, i-4\} \cdot w + \min\{i+3, 7\} \cdot w = (i+3) \cdot w.
\]
Therefore, using $ (i-1) \cdot w < \sigma_w(x) \le i \cdot w$, we have:
\[
    - 4w = (i-4) \cdot w - i \cdot w < \tau_w(x) - \sigma_w(x)  < (i+3) \cdot w - (i-1)\cdot w = 4w. \qedhere
\]
\end{proof}

\begin{lemma}\label{lem:diff-computed_position-computed_rank}
	For all elements $x$,	$\lvert \sigma_{\lfloor \rho w \rfloor}(x)-\tau_{w}(x)\rvert< 4w$.
\end{lemma}
\begin{proof}
Consider an element $x$. 
By \Cref{lem:diff-computed_rank-initial},
all elements $y$ with $\sigma_w(y)\leq \tau_w(x)-4w$ satisfy $\tau_w(y)<\tau_w(x)$, and all elements $z$ with $\sigma_w(z)\geq \tau_w(x)+4w$ satisfy $\tau_w(z)>\tau_w(w)$.
This shows that there are at least $\tau_w(x)-4w$ elements $y$ with $\tau_w(y)<\tau_w(x)$ and at least $n-\tau_w(x)-4w$ elements $z$ with $\tau_w(z)>\tau_w(x)$, implying that $\tau_w(x)-4w<\sigma_{\lfloor \rho w \rfloor}(x)<\tau_w(x)+4w$.
\end{proof}

\begin{lemma}\label{lem:move-one-round}
	For all elements $x$, $\lvert \sigma_{\lfloor \rho w \rfloor}(x)-\sigma_{w}(x)\rvert<8w$.
\end{lemma}
\begin{proof}
Using \Cref{lem:diff-computed_rank-initial}, \Cref{lem:diff-computed_position-computed_rank}, and the triangle inequality, we have:
\[
    | \sigma_{\lfloor \rho w \rfloor}(x)-\sigma(x) | \le |\sigma_{\lfloor \rho w \rfloor}(x) - \tau_w(x)| + |\tau_w(x) - \sigma(x) | < 4w + 4w = 8w. \qedhere 
\]
\end{proof}

\begin{lemma}\label{lem:move-afterward}
  For all elements $x$, $\lvert \sigma_{w}(x)-\sigma_{0}(x)\rvert < \frac{8w}{1-\rho}$. 
\end{lemma}
\begin{proof}
  Let $r+1$ be the number of rounds of \basketsort executed from (and including) the round with window size $w$ to completion and let $w_0, w_1, w_2, \dots, w_r$ be the corresponding window sizes so that $w_0 = w$, $w_r \ge 1$, and $\lfloor \rho w_r \rfloor = 0$. Let $w_{r+1} = 0$.
By \Cref{lem:move-one-round}, we know that
$|  \sigma_{w_i}(x) - \sigma_{w_{i+1}}(x) | < 8w_i$ for all $i=0,\dots,r$.
Since $w_i \le \rho^i w$, a repeated application of the triangle inequality yields:
\begin{align*}
    | \sigma_w(x) - \sigma_0(x) | 
    &= \left| \sum_{i=0}^{r}  \left( \sigma_{w_i}(x) - \sigma_{w_{i+1}}(x) \right) \right| 
  \le \sum_{i=0}^{r} \left| \sigma_{w_i} (x) - \sigma_{w_{i+1}}(x)\right| \\
    &< \sum_{i=0}^{r}  8w_i 
    < 8 \sum_{i=0}^{+\infty} \rho^i w = \frac{8w}{1-\rho}. \tag*{\qedhere}
\end{align*}
\end{proof}

\subsection{Analysis of the Maximum Dislocation}\label{sub:basket-maximum-dislocation}

We are now ready to analyze the maximum dislocation of the sequence returned by \basketsort.

Given an element $x \in S$ and a window size $w$, we say that an element is \emph{good} w.r.t.\ $w$ if its dislocation in $S_w$ is at most $w$. We avoid specifying $w$ when it is clear from context.
For the analysis of the maximum dislocation, it would suffice to prove that all elements are good w.r.t.\ all window sizes  $w = \Omega(\log m)$ considered by $\basketsort$ with a high enough probability.
However, the above requirement turns out to be too strict for our analysis of the total dislocation. We therefore relax such requirement in a way that is convenient for the sequel. Formally, given an element $x \in S$ and a window size $w$, we say that $x$ is \emph{happy} if $x$ is larger than all the elements lying apart by more than $2w$ positions on the left of $x$ in $S_w$ and $x$ is smaller than all the elements lying apart by more than $2w$ positions on the right of $x$ in $S_w$.
Notice that an element can be any combination of happy and/or good; however, if all elements are good, then they are also all happy (the converse is not necessarily true).

Fix a window size $w$ and a basket $B_i$ (of $S_w$). 
Consider the union $B$  (resp.\ $B^+$) of $B_i$ with the $3$ (resp.\ $6$) closest buckets to its left and to its right, if they exist. 
Formally, $B = \bigcup_{j = \max\{1,i-3\}}^{\min\{i+3,\lceil m/w\rceil\}}  B_j$ (which matches the definition in the pseudocode of \basketsort) and $B^+ = \bigcup_{j = \max\{1,i-6\}}^{\min\{i+6,\lceil m/w\rceil\}}  B_j$.
In other words, $B$ contains all elements $x$ with $\sigma_w(x) \in \{(i-4)w+1,\ldots,(i+3)w\}$ and $B^+$ contains all the elements $x$ with $\sigma_w(x) \in \{(i-7)w+1,\ldots, (i+6)w\}$.

\begin{lemma}
\label{lem:B_containment}
If all elements in $B^+$ are good and all elements in $B$ are happy, then $\{ (i-3)w+1, \dots, (i+2)w \} \subseteq B \subseteq \{ (i-5)w +1, \dots, (i+4)w \}$.
\end{lemma}
\begin{proof}
  We start by proving $\{ (i-3)w+1, \dots, (i+2)w \} \subseteq B$.
  Let $x \in \{(i-3)w+1, \dots, (i+2)w \}$.
If $x \in B^+$ then $x$ is good and $x + (i-4)w  +1 \le x - w \le \sigma_w(x) \le x + w \le (i+3)w$, which implies $x \in B$.

The complementary case  $x \not\in B^+$ cannot happen, as shown by the following two sub-cases.
If $x \not\in B^+$ and $\sigma_w(x) \le (i-7)w$, let $y$ be the first element of $B$, i.e., $\sigma_w(y) = (i-4)w+1$. Since $y$ is good, we must have $|  y- ((i-4)w+1)  | \le w$ and hence $y \le (i-3)w + 1 \le x$.
However, since $y$ is happy and $\sigma_w(y)- \sigma_w(x) > 2w$, we must have $x < y$, which is impossible.

If $x \not\in B^+$ and $\sigma_w(x) > (i+6)w$, let $z$ be the last element of $B$, i.e., $\sigma_w(z) = (i+3)w$. Since $z$ is good, we must have $| (i+3)w - z  | \le w$ and hence $z \ge  (i+2)w \ge x$.
However, since $z$ is happy and $\sigma_w(x) - \sigma_w(z) > 2w$, we must have $z < x$, which is impossible.

  We now prove $B \subseteq \{ (i-5)w+1, \dots, (i+4)w\}$. For $x \in B$, it trivially holds that $0 \le x \le m$ and, since $x$ is good, $|x - \sigma_w(x)| \le w$, which implies $(i-5)w + 1 \le \sigma_w(x)- w \le x \le \sigma_w(x) + w \le (i+4)w$.
\end{proof}

For the following lemmas, it is convenient to define the quantity 
\[\Bsc = \frac{\Paren{\rho (1-p-q) - 10(p-q) -8pq}^2 }{72},\]
which depends on the constants $p$ and $q$ and on the choice of $\rho$. 
Notice that, since $\rho \in \left( \frac{8pq + 10(p-q)}{1-p-q}, 1 \right)$ by hypothesis, we always have $\bsc \in (0, \frac{1}{3})$.
We also let $I = S \cap \{ (i-5)w +1, \dots, (i+4)w \}$, and we partition $I$ into three subsets $I_L$, $I_M$ and $I_R$, defined as $I_L= S \cap \{ (i-5)w+1, \dots, (i-3)w \}$,  $I_M= S \cap \{ (i-3)w+1, \dots, (i+2)w \}$, and $I_R = S \cap \{ (i+2)w+1, \dots, (i+4)w \}$.

\begin{lemma}\label{lem:score-difference-1}
Consider two elements $x\in B_i$ and $y\in B \setminus\{x\}$ such that $y\leq x-\rho \frac{w}{2}$ (resp.\ $y\geq x+ \rho \frac{w}{2}$) and $y \in I_M$.
If all elements in $B^+$ are good and all elements in $B$ are happy, then the probability that $\score(x) \le \score(y)$ (resp.\ $\score(x) \ge \score(y)$) is at most $\exp\left(-w\cdot \bsc\right)$.
\end{lemma}
\begin{proof}
We will only analyze the case $y\leq x-\rho \frac{w}{2}$, as the complementary case $y\geq x+\rho \frac{w}{2}$ is symmetric.

In order to compare $\score(x)$ with $\score(y)$, we study how the outcome of the comparison between $x$ and each element $z \in I \setminus \{x,y\}$ differs from the outcome of the comparison between $y$ and $z$. 
To this aim, we encode the differences in the comparison results using a collection of random variables $X_z$ which have values in $\{-1, 0, -1\}$ and are defined as follows:
\[
X_z = \begin{cases}
	-1 & \mbox{if $y \prec z \prec x$};\\
	0 & \mbox{if $y \prec z$ and $x \prec z$, or  $y \succ z$ and $x \succ z$};\\
	+1 & \mbox{if $y \succ z \succ x$}.
\end{cases}
\]
Moreover, we let $X_x$ be a random variable that equals $-1$ if $y \prec x$ and $+1$ if $y \succ x$, and we define: 
\[
    \widetilde{X}_z =
    \begin{cases}
        X_z                 & \text{ if } z \in I_M; \\ 
        \max\{0,\, X_z\}    & \mbox{ if } z \in I_L \cup I_R
    \end{cases}.
\] 
By \Cref{lem:B_containment}, we have $B\subseteq I=I_L\cup I_M\cup I_R$, and according to the above definitions, we have:
\begin{equation}
\label{eq:score_partition}
\score(y) - \score(x) =  \sum_{z\in B\setminus\{y\}} X_z \le    \sum_{z \in I_M \setminus \{y\}} \widetilde{X}_z + \sum_{z \in I_L \cup I_R} \widetilde{X}_z.
\end{equation}
In the rest of the proof, we will upper bound the probability that the right-hand side of~\eqref{eq:score_partition} is non-negative.
We start by considering the elements $z \in I_M \setminus \{ y \}$:
\begin{itemize}
    \item If $z=x$, then $\Pr(X_z=-1)=1-p_{\{x,y\}}$ and $\Pr(X_z=1)=p_{\{x,y\}}$. Therefore,  $\E[X_z]=-(1-p_{\{x,y\}})+p_{\{x,y\}}=2p_{\{x,y\}}-1 \le 2p-1$;
    \item If $z<y$, then $\Pr(X_z=-1)=p_{\{z,y\}}(1-p_{\{z,x\}})$, $\Pr(X_z = 1) = (1-p_{\{z,y\}})p_{\{z,x\}}$, and $\E[X_z]=-p_{\{z,y\}}(1-p_{\{z,x\}})+(1-p_{\{z,y\}})p_{\{z,x\}} = p_{\{z,x\}} - p_{\{z,y\}} \le p-q$; 
    \item If $z>x$, then $\Pr(X_z=-1)=(1-p_{\{z,y\}})p_{\{z,x\}}$, $\Pr(X_z = 1) = p_{\{z,y\}}(1-p_{\{z,x\}})$, and $\E[X_z]=-(1-p_{\{z,y\}})p_{\{z,x\}}+p_{\{z,y\}}(1-p_{\{z,x\}}) = p_{\{z,y\}} - p_{\{z,x\}} \le p-q$; 
    \item If $y<z<x$, then $\Pr(X_z=-1) = (1-p_{\{z,y\}})(1-p_{\{z,x\}})$ and $\Pr(X_z = 1) = p_{\{z,y\}} p_{\{z,x\}}$. Therefore, $\E[X_z] = -(1-p_{\{z,y\}})(1-p_{\{z,x\}})+p_{\{z,y\}}p_{\{z,x\}}=p_{\{z,y\}} + p_{\{z,x\}} -1 \le 2p-1$.
\end{itemize}
Note that $2p-1<0$ and that $p-q \ge 0$. 
    Since the number of elements $z \in I_M$ such that $y < z \le x$ is at least $x-y$, we have:
    \begin{align*}
        \E\left[\sum_{z\in I_M\setminus\{y\}} \widetilde{X}_z\right] &=
    \sum_{\substack{ z\in I_M\setminus\{y\}  \\ z \in \{y+1, \dots, x\} }} \E[X_z] + \sum_{\substack{ z\in I_M\setminus\{y\}  \\ z \not\in \{y+1, \dots, x\} }} \E[X_z] \\
        &\le
        (x-y) \cdot  (2p-1)  + (|I_M| - (x -y)) \cdot (p-q) \\
        &\le \rho \frac{w}{2} \cdot (2p-1) + (5w - \rho \cdot \frac{w}{2})(p-q)
        = \rho \frac{w}{2} (p+q-1) + 5w(p-q).
    \end{align*}

    We now consider the elements $z \in I_L\cup I_R$. 
    If $z < y$, then $\Pr(X_z = 1) = p_{\{z,y\}} (1-p_{\{z,x\}}) \le pq$, while if $z > x$, then $\Pr(X_z = 1) = (1-p_{\{z,y\}}) p_{\{z,x\}} \le pq$.
  Since $\E[\widetilde{X}_z] =\E[\max\{0,\, X_z\}]= \Pr(X_z=1) \le pq$, we have:
\[
\E\left[\sum_{z\in I_L \cup I_R} \widetilde{X}_z\right] 
= \sum_{z\in I_L \cup I_R}  \E[\widetilde{X}_z]
\le 4wpq. 
\]

To upper bound the probability that $\widetilde{X} = \sum_{z \in I_M \setminus \{y\}} \widetilde{X}_z + \sum_{z \in (I_L \cup I_R) \setminus \{y \}} \widetilde{X}_z$ is non-negative, we observe that all the $X_z$ are independent random variables with values in $\{-1, 0, 1\}$ and hence the same holds for the random variables $\widetilde{X}_z$. Moreover:
\begin{align*}
    \E[\widetilde{X}] &= \E\left[\sum_{z\in I_M\setminus\{y\}} \widetilde{X}_z\right] +  \E\left[\sum_{z\in I_L \cup I_R } \widetilde{X}_z\right] 
    \le  \rho \cdot \frac{w}{2} (p + q -1) +  5w(p-q) + 4wpq \\
    &= - \frac{w}{2} \left( \rho(1-p-q)  - 10(p-q) - 8pq \right),
\end{align*}
which is negative as soon as $\rho > \frac{8pq + 10(p-q)}{1-p-q}$, as per our hypothesis.

The claim follows by using Hoeffding's inequality:\begin{align*}
    \Pr\left(\widetilde{X}\geq 0\right)  &=\Pr\left(\widetilde{X} -\E[\widetilde{X}]\geq -\E[\widetilde{X}]\right)
    \leq  \exp\left(-\frac{\E[\widetilde{X}]^2}{2(\abs{I}-1)}\right) \leq \exp\left(-\frac{\E[\widetilde{X}]^2}{18 w-2}\right) \\
    & < \exp\left( -w \cdot \frac{(\rho (1-p-q) - 10(p-q) -8pq)^2 }{4 \cdot 18} \right)
    = \exp\left( -w \cdot \bsc \right). \qedhere
\end{align*}
\end{proof}

\begin{lemma}\label{lem:score-difference-2}
Consider two elements $x\in B_i$ and $y\in B \setminus\{x\}$ such that $y \not\in I_M$ and $y<x$ (resp.\ $y> x$).
If all elements in $B^+$ are good and all elements in $B$ are happy, then the probability that $\score(x) \le \score(y)$ (resp.\ $\score(x) \ge \score(y)$) is at most $\exp\left(-w\cdot \bsc\right)$.
\end{lemma}
\begin{proof}
    We focus on the case $y < x$ since the complementary case is symmetric. 

    Let the random variables $X_z$ and $\widetilde{X}_z$ for $z \in I\setminus \{y\}$ be defined as in the proof of \Cref{lem:score-difference-1} so that $\score(y) - \score(y) = \sum_{z \in B \setminus \{y\}} X_z \le \sum_{z \in I_M} \widetilde{X}_z + \sum_{z \in {I_L \cup I_R} \setminus \{y\}} \widetilde{X}_z$.

        Notice that, since $y \not\in I_M$ and all elements in $B$ are good, we have $y \le (i-3)w$ and $x \ge (i-1)w + 1 - w = (i-2)w + 1$ and hence there are at least $w$ elements $z$ such that $y < z < x$.
        Since the same bounds of \Cref{lem:score-difference-1} apply for $\E[X_z]$ when $z \in I_m$, we have:
        \[
            \E\left[\sum_{z \in I_M} \widetilde{X}_z \right] \le (x-y) \cdot (2p-1) + (|I_M| - (x-y)) \cdot (p-q)
            \le w (2p-1) + 4w(p-q).
        \]

    We now consider the elements $z \in (I_L\cup I_R) \setminus \{y\}$. 
    If $z < y$ or $z > x$, then $\Pr(\widetilde{X}_z = 1) = pq \le p^2$. If $y < z < x$, then $\Pr(\widetilde{X}_z =1) = p_{\{z,y\}} p_{\{z,x\}} \le p^2$. Thus, we have $\E[\sum_{z \in (I_L \cup I_R)} \widetilde{X}_z] \le 4w p^2$ and, defining $\widetilde{X} = \sum_{z \in I \setminus \{y\}} \widetilde{X}_z$, we obtain:
    \begin{align*}
        \E[\widetilde{X}] &\le w (2p-1) + 4w(p-q) + 4wp^2
        = -\frac{w}{2} \left(  2(1 - 2p) - 8(p-q) - 8p^2  \right) \\
        &\le -\frac{w}{2} \left(  \rho(1 - 2p) - 10(p-q) - 8p^2 + 2(p-q) \right)
        \le -\frac{w}{2} \left(  \rho(1 - 2p) - 10(p-q) - 8pq \right),
    \end{align*}
    where we used $-8p^2 + 2(p-q) > 8pq$ for $q \le p<\frac{1}{4}$, which follows from observing that, whenever $p \ge q \ge \frac{1}{4}$, we have $\frac{8pq + 10 (p-q)}{1-p-q} \ge 1$ and that this inequality violates the hypothesis on the value of $\rho$.
    Then, $\E[\widetilde{X}] < 0$ since $\rho > \frac{8pq + 10(p-q)}{1-p-q}$. 
    
    The claim follows by applying Hoeffding's inequality on $\widetilde{X}$ similarly to the corresponding part in the proof of \Cref{lem:score-difference-1}.
\end{proof}

The following lemma provides a sufficient condition for $\tau_w(x)$ to be a good approximation of $\rank(x, S_w) = x$.

\begin{lemma}\label{lem:scoresort-computed_rank}
Consider an element $x \in B_i$.
  If all elements in $B^+$ are good and all elements in $B$ are happy, then $\ABS{x-\tau_{w}(x)}\leq \rho \cdot \frac{w}{2}$ with probability at least
$
 1-9w\cdot\exp\left(-w \cdot f\left(p, \rho \right)\right)
$.
\end{lemma}
\begin{proof}
    By construction (see Line~\ref{ln:compute_tau} of \Cref{alg:bs}),  $\tau_w(x) = \max\{0,i-4\} \cdot w + \pos(x, A)$, and, since $x$ is happy,
$x=\max\{0,i-4\}\cdot w+\mbox{rank}(x,B)$. Therefore, we have $\lvert x-\tau_w(x)\rvert=\lvert\rank(x,B)-\pos(x,A)\rvert$.

\Cref{lem:score-difference-1,lem:score-difference-2} ensure that, for any element $y \in B$ with $y \le x-\rho \cdot \frac{w}{2}$ (resp.\ $y \ge x+ \rho \cdot \frac{w}{2}$), the probability that 
$\score(y) < \score(x)$ (resp.\ $\score(y) > \score(x)$) does not hold is at most $\exp(-w\cdot \bsc)$. 
Using the union bound over all the (at most $9w$) elements $y\in B$, we have that all elements $y \in B$ with $|x-y| > \rho \frac{w}{2}$ are in the correct relative order w.r.t.\ $x$ in $A$ with
probability at least $1 - 9w\cdot\exp(-w\cdot \bsc)$.
    Whenever this is the case, $ \rank(x, B) - \rho \frac{w}{2} \le \pos(x,A) \le \rank(x,B) + \rho \frac{w}{2}$ and the claim follows by combining such inequality with $|x - \tau_w(x)| = |\rank(x,B) - \pos(x, A)|$.
\end{proof}

We now show that if all elements are good w.r.t.\ $w$, they are also likely to all be good w.r.t.\ $\lfloor \rho w \rfloor$.

\begin{lemma}\label{lem:dislocation-probability}
Consider a round of \basketsort having window size $w$. 
  If all elements in $S_w$ are good  w.r.t.\ $w$, then 
  all elements in $S_{\lfloor \rho w \rfloor}$ are good w.r.t.\  $\rho w$ with probability at least $1-9mw\cdot\exp\left(-w \cdot \bsc\right)$. 
\end{lemma}
\begin{proof}
Notice that, since all elements in $\sigma$ are good  w.r.t.\ $w$, they are also happy w.r.t.\ $w$.
Then, using \Cref{lem:scoresort-computed_rank} and the union bound, we have that all elements $x$ satisfy $\lvert x - \tau_w(x)\rvert\leq \rho \frac{w}{2}$ with probability at least $1-9mw\cdot\exp\left(-w \cdot \bsc\right)$.

Suppose that this is indeed the case and focus on a single element $x$.
For all elements $y$ with $y< x - \rho w$, we have $\tau_w(y)\leq y+\rho \frac{w}{2} <  x-\rho \frac{w}{2}\leq\tau_w(x)$. Similarly, for all elements $y$ with $y> x+ \rho w$, we have $\tau_w(y)\geq y- \rho \frac{w}{w} > x+\rho \frac{w}{2}\geq\tau_w(x)$.
In other words, $x-\rho w \leq \sigma_{\lfloor \rho w \rfloor }(x) \leq x+\rho w$, i.e., $\lvert\sigma_{\lfloor \rho w \rfloor }(x)-x\rvert\leq \rho w$.
  Since both $x$ and $\sigma_{\lfloor \rho w \rfloor }(x)$ are integers, $\lvert\sigma_{\lfloor \rho w \rfloor }(x)-x\rvert\leq \lfloor \rho w \rfloor$ implying that $x$ is good w.r.t.\ $\lfloor \rho w \rfloor$.
\end{proof}

We can use the above result to give a lower bound to the probability that all elements remain good for all sufficiently large window sizes.

\begin{lemma}
  \label{lem:all_elements_good_until_wcrit}
  Consider an execution of \basketsort with $w_S \ge \disl(S)$. 
    With probability at least $1-O\big(\frac{\log m}{m^4 \log \frac{1}{\rho}}\big)$,
   all elements $x \in S$ are good w.r.t.\ all window sizes $w$ with $w \ge \frac{6}{\bsc}\ln m$ considered by the algorithm. 
\end{lemma}
\begin{proof}
We start by focusing on a generic round having window size $w \ge \frac{6}{\bsc} \ln m$.
If all elements are good w.r.t.\ $w$, then all elements are also happy w.r.t.\ $w$ and \Cref{lem:dislocation-probability} ensures that the maximum dislocation of $\sigma_{\lfloor \rho w \rfloor}$ is at most $\lfloor \rho w \rfloor$ with probability at least 
\[
1-9mw\cdot\exp\left(-w \cdot f\left(p, \rho\right)\right) 
\ge 1 - 9mw \cdot \exp(-6 \ln m)
  \ge 1 - 9m^2 \cdot \frac{1}{m^6} = 1 - \frac{9}{m^4},
\]
and, whenever this happens, all elements are alo good w.r.t.\ $\lfloor \rho w \rfloor$.

  Since all elements are good w.r.t.\ the window size $w_S \le m$, we have that the probability that all elements remain good w.r.t.\ all the $O(\log_{1/\rho} m)$ subsequent window sizes $w$ such that $w \ge \frac{6}{\bsc} \ln m$ is at least $1- O\left(\frac{\log_{1/\rho} m}{m^4}\right) = 1-O\left( \frac{\log m}{m^4 \log \frac{1}{\rho}} \right)$.
\end{proof}

Whenever all elements are good w.r.t.\ all window sizes $w = \Omega(\log m)$, the dislocation of the sequence returned by \basketsort is $O(\log m)$, which allows us to prove the following:

\begin{lemma}\label{lemma:bs_maximum-dislocation}
  The maximum dislocation of the sequence returned by \basketsort with $w_S \ge  \disl(S)$ is at most $\frac{9}{\rho(1-\rho)} + \frac{54}{\rho (1-\rho) \cdot \bsc} \ln m$
    with probability at least $1-O\big(\frac{\log m}{m^4 \log \frac{1}{\rho} }\big)$. 
\end{lemma}
\begin{proof}
  Let $w$ be the smallest window size $w$ that satisfies $w \ge \frac{6}{\bsc} \ln m$.\footnote{If no such window size exists, then $w_S < \frac{6}{\bsc} \ln m$, and, by the algorithm's precondition and by \Cref{lem:move-afterward}, the maximum dislocation of the returned sequence is at most $(1+\frac{8}{1-\rho}) w_S  < \frac{6}{\bsc} + \frac{48}{(1-\rho) \cdot \bsc}$, which satisfies the claim.}
  By \Cref{lem:all_elements_good_until_wcrit}, all elements are good w.r.t.\ $w$ with probability at least $1 - O(\frac{\log m}{m^4 \log \frac{1}{\rho}})$.
  When this is indeed the case, the maximum dislocation of $\sigma_w$ is at most $w <  \frac{1}{\rho} + \frac{6}{ \rho \cdot \bsc} \ln m$, where we used $\lfloor \rho w \rfloor < \frac{6}{\bsc} \ln m$.

  By \Cref{lem:move-afterward}, each element moves by at most $\frac{8 w}{1-\rho} < \frac{8}{\rho(1-\rho)} +  \frac{48}{\rho(1-\rho) \cdot \bsc} \ln m$ positions in all subsequent rounds. Therefore, the maximum dislocation of the permutation returned by \basketsort is at most
  $\frac{9}{\rho(1-\rho)} + \frac{54}{\rho (1-\rho) \cdot \bsc} \ln m$ with the aforementioned probability.
%
%
\end{proof}

\subsection{Analysis of the Total Dislocation}\label{sub:basket-total-dislocation}

In this section we prove that the total dislocation of the sequence returned by \basketsort is $O(m)$ with high probability.
We first define a condition, denoted by \weakINV{}, which allows us to analyze the expected dislocation of a single element.
Roughly speaking, if an element $x$ satisfies~\weakINV{} w.r.t.\ a window size $w$, then  $x$ still satisfies ~\weakINV{} until a window size of $O(\log w)$ with probability $1-O(1/w^2)$ (\Cref{lem:weak-pr}).
Then, in \Cref{subsub:constant-dislocation}, we make use of the above fact to prove that the expected dislocation of each element is $O(1)$.
In order to do so, 
our analysis partitions the $O(\log m)$ rounds of \basketsort into $O(\log\log m)$ phases,
where each phase extends from a round having some window size $w$ to a round having a window size in  $O(\log w)$.
Finally, \Cref{subsub:deterministic_total_dislocation}, uses the constant upper bound on the expected dislocation of each element to prove that the total dislocation is $O(m)$ with high probability. 

We start by formally defining \weakINV{}. Given an element $x$ and a window size $w$, $x$ satisfies \weakINV{} w.r.t.\ $w$ if \emph{both} the following conditions hold:
\begin{enumerate}[(a)]
  \item All elements $y$ with $|\sigma_w(y) - x| \le \trho w$ are good; and
 %
  \item All elements $y$ with $|\sigma_w(y) - x| \le (\trho - 3)w$ are happy.

\end{enumerate}	

%
%

\begin{lemma}\label{lem:weak_part_a}
  Let $x$ be an element that satisfies \weakINV{} w.r.t.\ $w$ and assume that all elements $z$ with
  $|\sigma_w(z) - x| \le (\trho \rho + 8)w$ 
  satisfy $|\tau_{w}(z)-z|\leq \rho \frac{w}{2}$.
  Then, $x$ satisfies condition~(a) of \weakINV{} w.r.t.\ $\lfloor \rho w \rfloor$.
  Moreover, all elements $z$ such that $|\tau_w(z) - x| \le (\trho \rho + 4) w$ are good w.r.t.\ $\floor{\rho w}$. 
\end{lemma}
\begin{proof}
  Let $Y$ denote the set of all elements $y$ for which $|\tau_w(y) - x| \le (\trho \rho + 4)w$.
For any $y\in Y$, by \Cref{lem:diff-computed_rank-initial}, we have $|\tau_{w}(y)-\sigma_{w}(y)| < 4w$, thus implying that
\[
  x - \Paren{\trho\rho+8}w \le \tau_{w}(y)-4w\ < \sigma_{w}(y) < \tau_{w}(y)+4w \le x + \Paren{\trho\rho+8}w.
\]
That is, for any $y\in Y$, we have $\abs{\sigma_w(y)-x}\leq  \Paren{\trho\rho+8}w$. 
  In particular, the hypotheses of the lemma imply that any $y \in Y$ satisfies $|\tau_w(y) - y| \le  \rho \frac{w}{2}$.
  Moreover, since $\frac{15}{1-\rho}\rho + 8 \le  \frac{15}{1-\rho} - 3$, condition~(a) of \weakINV{} w.r.t. $w$ for element $x$ ensures that $y$ is both good and happy w.r.t.\ $w$. Observe further that $Y$ contains all elements $y'$ with
  $|\sigma_w(y') - x| \le \trho \floor{\rho w}$, which are those considered by condition~(a) of \weakINV{} w.r.t.\ $\floor{\rho w}$ for element $x$. Indeed, for any $y'$ with $\abs{\sigma_{\lfloor \rho w \rfloor}(y')-x}\leq  \trho\cdot \floor{\rho w} $, by \Cref{lem:diff-computed_position-computed_rank},
\[
  |x - \tau_w(y')| \le |x - \sigma_{\lfloor \rho w \rfloor}(y')| +  |\sigma_{\lfloor \rho w \rfloor}(y') - \tau_w(y') | \le \trho\cdot \floor{\rho w}  + 4 w \le \Paren{\trho\rho+4}w. 
\]

  Hence, to prove the claim, it suffices to show that all elements in $Y$ are good w.r.t.\ $\floor{\rho w}$. %
  In the rest of the proof we focus on an element $y \in Y$ and show that any element $z$ such that $z  < y - \rho w$ (resp.\ $z > y + \rho w$) must appear before (resp.\ after) $y$ in $S_{\floor{\rho w}}$, thus implying that $y$ is good w.r.t.\ $\floor{\rho w}$.
  We distinguish two cases depending on the value of $|\sigma_w(z) - x|$.

  \paragraph*{Case 1:}  This case considers $|\sigma_w(z) - x| > (\trho \rho + 8 )w$. 
  Using the triangle inequality, the fact that $y$ is good w.r.t.\ $w$, $|\tau_w(y) - y| \le \rho \frac{w}{2}$, and the definition on $Y$, we have: 
    \begin{align*}
        |\sigma_w(y)-x| &\le |\sigma_w(y)-y| + |y - \tau_w(y)| + |\tau_w(y) - x| \\
        & \le w + \rho \frac{w}{2}  + \Paren{\trho \rho + 4}w
    = \Paren{\trho \rho + \frac{\rho}{2} + 5}w.
    \end{align*}
%
%

  Then, all elements $z$ with $\sigma_w(z)<x-(\trho \rho + 8 )w$  (resp.\ $\sigma_w(z)>x+(\trho \rho + 8 )w$) are necessarily smaller than (resp.\ larger than) $y$.
  Indeed, $\sigma_w(z)<x-(\trho \rho + 8 )w \le \sigma_w(y) - (3 - \frac{\rho}{2})w < \sigma_w(y) - 2w$
  (resp.\  $\sigma_w(z)>x+(\trho \rho + 8 )w \ge\sigma_w(y) + (3 - \frac{\rho}{2})w > \sigma_w(y) + 2w$), 
  and the fact that $y$ is happy w.r.t.\ $w$ implies $z < y$ (resp.\ $z>y$).

  From the above discussion, we know that if $z < y - \rho w$, then we must have  $\sigma_w(z)<x-(\trho \rho + 8 )w$ (since if $\sigma_w(z)>x+(\trho \rho + 8 )w$, then $z>y$).  Using \Cref{lem:diff-computed_rank-initial}, we obtain
  $
    \tau_w(z) < \sigma_w(z) + 4w < x - (\trho \rho + 4 )w \le \tau_w(y),
  $
  implying that $\sigma_{\floor{\rho w}}(z) < \sigma_{\floor{\rho w}}(y)$.
  Symmetrically, if $z > y + \rho w$, we must have $\sigma_w(z) > x + (\trho \rho + 8 )w$, which implies $\tau_w(z) > \sigma_w(z) - 4w > x + (\trho \rho + 4w) \ge \tau_w(y)$ and hence $\sigma_{\floor{\rho w}}(z) > \sigma_{\floor{\rho w}}(y)$.

%
%
%
%
%
%
%
%
%
%
%
%
  


  \paragraph*{Case 2:} This case considers $|\sigma_w(z) - x| \le (\trho \rho + 8 )w$. Notice that $z \in Y$ and hence $|\tau_{w}(z)-z|\leq \rho \frac{w}{2}$.
%
%
%
  If $z < y - \rho w$, then $\tau_w(z) \le z + \rho \frac{w}{2} < y - \rho \frac{w}{2} \le \tau_w(y)$, where the last inequality holds since $|\tau_w(y)  -y | \le \rho \frac{w}{2}$. 
  Using $\tau_w(z) < \tau_w(y)$, we have $\sigma_{\floor{\rho w}}(z) < \sigma_{\floor{\rho w}}(y)$ as desired.
  Symmetrically, if $z > y + \rho w$, then $\tau_w(z) \ge z - \rho \frac{w}{2} > y + \rho \frac{w}{2} \ge \tau_w(y)$, which implies $\sigma_{\floor{\rho w}}(z) > \sigma_{\floor{\rho w}}(y)$ as desired.
  %
  %
 %
%
%
\end{proof}

\begin{lemma}\label{lem:weak}
  Let $x$ be an element that satisfies \weakINV{} w.r.t.\ $w$ 
  and assume that all elements $z$ with 
  $| \sigma_w(z) - z | \le x - (\trho \rho +8)w$
  satisfy $|\tau_{w}(z)-z|\leq \rho \frac{w}{2}$.
  Then, $x$ satisfies \weakINV{} w.r.t.\ $\lfloor \rho w \rfloor$.
\end{lemma}
\begin{proof}
  By \Cref{lem:weak_part_a}, $x$ satisfies condition~(a) of \weakINV{} w.r.t.\ $\lfloor \rho w \rfloor$, and therefore we only need to prove that $x$ also satisfies condition~(b). 

  Let $T$ be the set of all elements $y$ with $|\sigma_{\floor{\rho w}}(y) - x| \le (\trho - 3) \lfloor \rho w \rfloor$. 
  For any element $y\in T$, we have $\abs{\sigma_w(y)-\sigma_{\floor{\rho w}}(y)}\leq 8w$ (see \Cref{lem:move-one-round}), and thus $|\sigma_w(y) - x| \le  (\trho - 3) \rho w + 8w \le   (\trho - 3)w$, implying that $y$ is both happy and good w.r.t. $w$ (since $x$ satisfies \weakINV{}).

  Focus on an element $y \in T$. We want to show that any element $z$ such that $\sigma_{\floor{\rho w}}(z) < \sigma_{\floor{\rho w}}(y) - 2\floor{\rho w}$ (resp.\ $\sigma_{\floor{\rho w}}(z) > \sigma_{\floor{\rho w}}(y) + 2\floor{\rho w}$) must be smaller (resp.\ larger) than $y$ (w.r.t.\ the true order of the elements). 
  This immediately implies that $y$ is happy w.r.t.\ $\lfloor \rho w \rfloor$, hence the claim follows.

  In the rest of the proof, we only argue about elements $z$ such that $\sigma_{\floor{\rho w}}(z) < \sigma_{\floor{\rho w}}(y) - 2\floor{\rho w}$ as symmetric arguments hold for $\sigma_{\floor{\rho w}}(z) > \sigma_{\floor{\rho w}}(y) - 2\floor{\rho w}$.


We start with some preliminary observations: first of all, recall that all elements $y'$ with
  $|\sigma_{\floor{\rho w}}(y') - x| \le \trho \floor{ \rho w}$ are good w.r.t.\ $\floor{\rho w}$ by \Cref{lem:weak_part_a}, and, since the set of these elements is a superset of $T$, $y$ is good w.r.t. $\floor{\rho w}$. 
%
%
%
%
  Moreover, we can upper bound $\sigma_{\floor{\rho w}}(z)$ by writing $\sigma_{\floor{\rho w}}(z) < \sigma_{\floor{\rho w}}(y) - 2\floor{\rho w}  \le x+(\trho-5) \floor{ \rho w}$, and thus all $z$ with $| \sigma_{\floor{\rho w}}(z) - x | \le  \trho \floor{\rho w}$ immediately satisfy $z < y$. %
%
  Indeed, any such $z$ is good w.r.t.\ $\floor{\rho w}$ as $x$ satisfies condition~(a) of \weakINV{}, and thus $z \le \sigma_{\floor{\rho w}}(z) + \floor{\rho w} < \sigma_{\floor{\rho w}}(y) - \floor{\rho w} \le y$.

  As a consequence, it remains to analyze the case $| \sigma_{\floor{\rho w}}(z) - x | >  \trho \floor{\rho w}$, and we henceforth further restrict ourselves to the elements $z$ for which $\sigma_{\floor{\rho w}}(z) < x - \trho \floor{\rho w}$.\footnote{Since $y \in T$, we have $\sigma_{\floor{\rho w}(y)} + 3\floor{\rho w} \le x  + \trho \floor{\rho w}$. Thus, any element $z$ with $\sigma_{\floor{\rho w}}(z) > x + \trho \floor{\rho w}$ satisfies $\sigma_{\floor{\rho w}}(z) \ge \sigma_{\floor{\rho w}(y)} + 3 \floor{\rho w} > \sigma_{\floor{\rho w}(y)} -2 \floor{\rho w} $, and does not need to be considered.}

  By \Cref{lem:move-one-round}, we have that $\sigma_w(z) < \sigma_{\floor{\rho w}}(z) + 8w < x-(\trho\rho-8) w$. To conclude the proof we separately consider three cases depending on the value of $\sigma_w(z)$.

  \paragraph*{Case 1:} The first case considers  $\sigma_w(z) < x - \trho w$. 
   From $\sigma_w(y) \ge x - (\trho -3)w$,  we have $x - \trho \le \sigma_w(y) - 3w$, i.e., $\sigma_w(z)<\sigma_w(y) - 3w$,   which implies $z < y$ since $y$ is happy.

  \paragraph*{Case 2:} The second case considers $x - \trho w \le \sigma_w(z) < x - (\trho\rho+1)w$. In this case, $z$ is good w.r.t.\ $w$, as ensured by condition (a) of \weakINV{} from the hypothesis of this lemma and by $\trho>\trho\rho+1$.
   Thus, we have $z \le \sigma_w(z) + w < x - \trho  \rho w \le x - (\trho -3) \rho w - \floor{\rho w}  \le \sigma_{\floor{\rho w}}(y) - \floor{\rho w} \le y$, where we relied on the definition of $T$ 
    and the fact that $y$ is good w.r.t.\ $\floor{\rho w}$.
 
  \paragraph*{Case 3:} The third and final case considers  $x- (\trho\rho+1)w \le \sigma_w(z) < x-(\trho\rho-8) w$.  By condition~(a) of \weakINV{}, $z$ is good w.r.t.\ $w$, 
    and thus $x - (\trho \rho +2)w \le z < x -(\trho\rho-9)w$. %
%
  Moreover, by the assumption of the lemma, $x - (\trho \rho +2 + \frac{\rho}{2})w \le \tau_w(z) \le x - (\trho \rho - 9 - \frac{\rho}{2})w$, and further by \Cref{lem:weak_part_a}, $z$ is good w.r.t.\ $\floor{\rho w}$, implying that  $|\sigma_{\floor{\rho w}}(z) - z | \le \floor{\rho w}$.
%
%
Since $\abs{\sigma_{\floor{\rho w}}(y) - y}\leq \floor{\rho w}$ (as explained before) and $\sigma_{\floor{\rho w}}(z)<\sigma_{\floor{\rho w}}(y) -2\floor{\rho w}$, we have that $z\leq  \sigma_{\floor{\rho w}}(z)+\floor{\rho w} < \sigma_{\floor{\rho w}}(y)-\floor{\rho w}\leq y$. 
%
%
%
%
%
\end{proof}

\begin{lemma}\label{lem:weak-pr}
  Let $w$ and $w^*$ be two window sizes satisfying
    $\frac{4}{\bsc} \cdot \ln \lfloor \rho w \rfloor \le w^* \le w$ and $w  \ge 4$.
	If \emph{\weakINV{}}\ holds for an element $x$ w.r.t.\ $w$,
    then \emph{\weakINV{}}\ holds for $x$ w.r.t.\ all window sizes that are at least $w^*$ and at most $w$ with probability at least $ 1 - \frac{69120}{(1-\rho)^3 w^2}$.     
\end{lemma}
\begin{proof}
  Suppose that \weakINV{} holds for $x$ w.r.t.\ a generic window size $w'$ such that $w^* \le w' \le w$.  
  The probability that \weakINV{} fails for $x$ w.r.t.\ $\lfloor \rho w' \rfloor$ is at most the probability that the hypothesis of  Lemma~\ref{lem:weak} does not hold, i.e., the probability $P(w')$ that there exists an element $y$ with $\sigma_{w'}(y)\in [x-(\trho \rho +8)w', x+(\trho \rho +8)w']$ such that $|\tau_{w'}(y)-y|> \rho\frac{w'}{2}$.
  
  Recall the definition of $B_i$, $B$ and $B^+$ in \Cref{sub:basket-description} and  \Cref{sub:basket-maximum-dislocation}.
  In order to apply \Cref{lem:scoresort-computed_rank}, we consider the basket $B_i$ of $\sigma_{w'}$ that contains $y$, i.e., $\max\{(i-1)\cdot w',\,0\}+1\leq \sigma_{w'}(y)\leq \min\{i\cdot w',m\}$.
  \Cref{lem:scoresort-computed_rank} requires that the up to $7w'$ elements in $B$ are happy w.r.t.\ $w'$ and that the up to $13w'$ elements in $B^+$ are good w.r.t.\ $w'$.
  By the definition of $B$ and $B^+$, any element $z$ in $B$ satisfies $\abs{\sigma_{w'}(z)-\sigma_{w'}(y)}\leq 4w'$ and any element $z'$ in $B^+$ satisfies $\abs{\sigma_{w'}(z')-\sigma_{w'}(y)}\leq 7w'$.
  Since $\sigma_{w'}(y)\in [x-(\trho \rho +8)w', x+(\trho \rho +8)w']$ and $[x-(\trho \rho +8+7)w', x+(\trho \rho +8+7)w']= [x-\trho w', x+\trho w']$, the fact that \weakINV{} holds w.r.t.\ $w'$ immediately implies that the hypotheses of \Cref{lem:scoresort-computed_rank} are met.
  
  Therefore, by \Cref{lem:scoresort-computed_rank}, Lemma~\ref{lem:weak}, and by using the union bound over the (at most) $2\Paren{\trho \rho +8}w'+1$ elements $y$ with  $\sigma_{w'}(y)\in [x-(\trho \rho +8)w', x+(\trho \rho +8)w']$, we have: 
	 \begin{align*}
	 	P(w') & \leq \Paren{2\Paren{\trho \rho +8}w'+1}\cdot 9w' \exp(-w'\cdot \bsc)\\
	 	       & \leq \Paren{2\Paren{\trho \rho +8}w'+1} \cdot 9w' \exp(- 4 \ln \lfloor \rho w \rfloor)\\
              & \leq \frac{18\Paren{\trho \rho +9}(w')^2}{\lfloor \rho w \rfloor^4}
	 	      < \frac{256 \cdot 18\Paren{\trho}(w')^2}{w^4}
              = \frac{ 69120 (w')^2}{(1-\rho) w^4},
    \end{align*}
    where $\trho \rho + 9 < \trho$ and $\lfloor \rho w \rfloor^4 \le (\rho w - 1)^4 \le \frac{\rho^4 w^4}{16} \le \frac{w^4}{256}$ since $\lfloor \rho w\rfloor \ge  \lfloor \frac{4}{2} \rfloor = 2$.

    Let $r$ denote the number of iterations of \basketsort having a window size between $w$ (included) and $w^*$ (excluded). The probability that condition \weakINV{} fails for $x$ w.r.t.\ at least one such window size is at most:
  \[
      \sum_{i=0}^{r-1}  \frac{69120 (\rho^i w)^2}{(1-\rho) w^4}
    \le \frac{69120}{(1-\rho) w^2}\cdot \sum_{i=0}^{\infty} \rho^{2i}
    = \frac{69120}{(1-\rho) (1-\rho^2) w^2}
    = \frac{69120}{(1-\rho)^3 w^2}. \qedhere
%
%
  \]
\end{proof}

\subsubsection{A linear upper bound in expectation}\label{subsub:constant-dislocation}

We prove a linear upper bound on the expected total dislocation by proving a constant upper bound on the expected dislocation of each element $x$.

Within this subsection, we assume that $w_S \ge 4$ (since otherwise \Cref{lem:move-afterward} ensures that the dislocation of $x$ is at most $3 + \frac{24}{1-\rho}$).
We define $w_0 = w_S$, and for an integer $i\ge 1$, we let $w_i =  \frac{4}{\bsc} \cdot \ln w_{i-1}$.
Moreover, let $z$ be the smallest index for which $w_{z+1} < 4$ or $w_{z+1} > \frac{w_z}{2}$, i.e., $\frac{4}{\bsc} \ln w_z > \frac{w_z}{2}$. 
Notice that we have $w_{i} \ge 2 w_{i+1}$ for $i=0, \dots, z-1$, $w_z \ge 4$, and $\frac{4}{\bsc} \ln w_z \ge \frac{w_z}{2}$.
Therefore, $w_z \le \frac{8}{\bsc} \ln w_z \le \frac{8}{\bsc} \sqrt{w_z}$, and hence $w_z \le \frac{64}{\bsc^2}$.

We group the rounds of \basketsort into $z+1$ phases as follows: phase $0$ consists of the sole round with window size $w_S$, while phase $i \in \{1, \dots, z\}$ consists of all rounds that consider a window size that is smaller than $w_{i-1}$ and at least $w_i$. Since $\rho \ge \frac{1}{2}$, our choice of $w_i$ ensures that each phase contains at least one round. 

Given an element $x$, we say that $x$ satisfies \weakINV{} \emph{until phase $i$} if $x$ satisfies \weakINV{} w.r.t.\ all window sizes $w$ with $w \ge w_i$. Notice that, since the dislocation of the input sequence is at most $w_S$, all elements are both good and happy w.r.t.\ $w_0 = w_S$ and hence $x$ always satisfies \weakINV{} until phase $0$.
We say that $x$ \emph{violates} \weakINV{} in phase $i \ge 1$ if $x$ does not satisfy \weakINV{} until phase $i$ but it $x$ satisfies \weakINV{} until phase $i-1$.

\begin{lemma}
\label{lem:prob_x_violates_weakinv_in_a_phase}
    The probability that an element $x$ violates \weakINV{} in phase $i \in \{1, \dots, z\}$ is at most $\frac{69120}{w_{i-1}^2}$.
\end{lemma}
\begin{proof}
If $x$ violates $\weakINV{}$  in phase $i$, then $x$ satisfies \weakINV{} until phase $i-1$.
Therefore, we let $w$ be the smallest window size with $w \ge w_{i-1}$, and we assume that $x$ satisfies \weakINV{} w.r.t.\ $w$.
Observe that our choice of $w$ ensures that $\rho w < w_{i-1}$.
Let $w^*$ be the smallest window size such that $w^* \ge w_i$.
We have $w^* \ge w_i \ge \frac{4}{\bsc} \cdot \log w_{i-1} \ge \frac{4}{\bsc} \cdot \ln \lfloor \rho w \rfloor$.

Then, \Cref{lem:weak-pr} ensures that $x$  satisfies \weakINV{}  w.r.t.\ all window sizes between $w$ and $w^*$ with probability at least $1- \frac{69120}{(1-\rho)^2 w^2} \ge 1- \frac{69120}{(1-\rho)^2 w_{i-1}^2}$.
\end{proof}

\begin{lemma}
    \label{lem:expected_dislocation_one_element}
    The expected dislocation of each element $x$ in $S_0$ is at most 
    $\frac{1152}{\bsc^2 (1-\rho)} + \frac{311040}{(1-\rho)^4}$.
\end{lemma}
\begin{proof}
Consider a phase $i \in \{0, \dots, z\}$, let $w$ be the smallest window size such that $w \ge w_i$, and observe that $w \le \frac{w_i}{\rho} \le 2w_i$. 
If $x$ satisfies $\weakINV{}$ until phase $i$, then $x$ is good w.r.t.\ $w$. 
Hence, $|x - \sigma_w(x)| \le w$ and, by \Cref{lem:move-afterward} and the triangle inequality, we have:
\[
    |x - \sigma_0(x)| \le |x - \sigma_w(x) | + | \sigma_w(x) - \sigma_0(x) |
    < w + \frac{8w}{1-\rho} < \frac{9w}{1-\rho} \le  \frac{18 w_i}{1-\rho}. 
\]

Recall that $x$ satisfies \weakINV{} in phase $0$.
Then, either $x$ satisfies $\weakINV{}$ until phase $z$, or it violates \weakINV{} in some phase $i \in \{1, \dots, z\}$.
In the former case, we have $|x - \sigma_0(x)| \le \frac{18 w_z}{1-\rho} \le \frac{18 \cdot 64}{\bsc^2 (1-\rho)} = \frac{1152}{\bsc^2 (1-\rho)}$.
In the latter case, \Cref{lem:prob_x_violates_weakinv_in_a_phase} ensures that the probability that $x$ violates $\weakINV$ in phase $i$ is at most:    
    \[
        \frac{69120}{(1-\rho)^3 w_{i-1}^2}
        \le 
        \frac{8640}{(1-\rho)^3 w_{i-1}} \cdot 2^{i-z},
    \]
    where $w_{i-1} \ge 2^{z-i+3}$ follows from $w_z \ge 4$ and $w_{i} \ge 2 w_{i+1}$.
    Moreover, we have $|x - \sigma_0(x)| \le \frac{18 w_{i-1}}{1-\rho}$ since $x$ satisfies $\weakINV{}$ until phase $i-1$.

    Overall, the expected dislocation of $x$ is at most:
    \begin{align*}
        \frac{1152}{\bsc^2 (1-\rho)} + \frac{8640 \cdot 18}{(1-\rho)^4} \sum_{i=1}^z 2^{i-z}
        < 
        \frac{1152}{\bsc^2 (1-\rho)} + \frac{311040}{(1-\rho)^4}. \tag*{\qedhere}
    \end{align*}

\end{proof}

\subsubsection{A linear upper bound with high probability}\label{subsub:deterministic_total_dislocation}
After proving that the expected dislocation of each element is constant, we additionally prove that the total dislocation is linear with high probability.
To this end, we first prove that the final position of an element only depends on its nearby elements and on the outcome of their comparisons. 
For a window size $w$, let $C_w(h, \ell)$ denote a collection of data including (i) the set of all elements $z$ for which $|\sigma_{w}(z)  - h| \le \ell$, (ii) their positions in $S_w$, and (iii) the (random variables denoting the) results of the comparisons between them.

\begin{lemma}\label{lem:dependent}
    For any element $x$ and any window size $w$, $\sigma_0(x)$ is a function of \linebreak $C_w(\sigma_w(x), \frac{36}{1-\rho} w)$.
\end{lemma}
\begin{proof}
    Let $w$ be a window size considered by \basketsort and define $w' = \lfloor \rho w \rfloor$.
    Consider a generic element $y$ belonging to basket $B_j$ in the round corresponding to window size $w$.
    The value $\tau_w(y)$ is a function of $j$, of the elements in the baskets $B_{j-3}, \dots, B_{j+3}$ (if they exist), of their order in $S_w$,\footnote{The order of the elements is  used, e.g., to break ties in the sorting steps of \basketsort.} and of the comparison results between $y$ and these elements.
    Let $B_i$ be the basket containing $x$. 
    The position $\sigma_{w'}(x)$ is a function of $i$, of the elements $y$ belonging to the baskets $B_j$ with $|j-i| \le 8$ in the iteration corresponding to window size $w$, and of their values $\tau_w(y)$.
 Indeed, for any element $y'$ that belongs to some basket $B_j$ with $j < i-8$ we have $\sigma_w(y') < \sigma_w(x) - 8w$ and Lemma~\ref{lem:diff-computed_rank-initial} ensures that $y'$ precedes $x$ in $S_{w'}$ since
   $
    \tau_w(y') < \sigma_w(y') + 4w < \sigma_w(x) - 4w < \tau_w(x)
  $. 
    A symmetric argument shows that if $j > i+8$ then $y'$ must follow $x$ in $S_{w'}$.
 
    Since all elements  $y$ in $B_{i-11}, \dots, B_{i+11}$ satisfy $|\sigma_{w}(y) - \sigma_{w}(x)| \le 12w$, $\sigma_{w'}(x)$ is a function of $C_w(\sigma_w(x), 12w)$.
    Moreover, using $|\sigma_{w}(y) - \sigma_{w'}(y)| < 8w$ (see Lemma~\ref{lem:move-one-round}), we also obtain that $\sigma_{w'}$ is a function of $C_w(\sigma_{w'}(x), 20w)$.

    Let $w_k = w, w_{k-1} = \lfloor \rho w \rfloor, w_{k-2}, \dots, w_{1}$ be the window sizes considered by \basketsort from the round corresponding to window size $w$ to termination, and let $w_0 = 0$.
    Notice that, unlike the proof of \Cref{subsub:constant-dislocation}, here the window sizes $w_i$ increase with $i$.
    We now prove by induction on $i = 0, \dots, k$ that $\sigma_0(x)$ is a function of $C_{w_i}(\sigma_{0}(x), \frac{28}{1-\rho} w_i)$.
    The base case $i=0$ is trivial, since $C_{0}(\sigma_0(x), 0)$ includes the position $\sigma_0(x)$ of $x$ in $S_0$.
    Assume then that the inductive claim holds up to some $i < k$.
    By induction hypothesis, $\sigma_0(x)$ is a function of $C_{w_i}(\sigma_0(x), \frac{28}{1-\rho} w_i)$, and, by \Cref{lem:move-one-round} and triangle inequality, all elements $x'$ considered in $C_{w_i}(\sigma_0(x), \frac{28}{1-\rho}  w_i)$ satisfy $| \sigma_{w_{i+1}}(x') - \sigma_0(x) | \le |\sigma_{w_{i+1}}(x') - \sigma_{w_i}(x') | + |\sigma_{w_i}(x') - \sigma_{0}(x)| < 8 w_{i+1} + \frac{28}{1-\rho}  w_i$.
    By the discussion above, the position $\sigma_{w_{i+1}}(x')$ of any such element $x'$ is a function of $C_{w_{i+1}}(\sigma_{w_i}(x'), 20w_{i+1})$, and hence the position $\sigma_{w_{i+1}}(x')$ of \emph{all} the elements $x'$ is a function of $C_{w_{i+1}}(\sigma_0(x),28 w_{i+1} +  \frac{28}{1-\rho} w_i )$.
    Using $w_i \le \rho w_{i+1}$, we have $28 w_{i+1} + \frac{28}{1-\rho} w_i \le ( 28 + \frac{28 \rho}{1-\rho})w_{i+1} = \frac{28}{1-\rho}w_{i+1}$.

    To conclude the proof, observe that, for any window size $w$, the elements in $C_w(\sigma_0(x), \frac{28}{1-\rho} w)$ are a subset of those in $C_w\Paren{\sigma_w(x), \Paren{\frac{28}{1-\rho} + \frac{8}{1-\rho}}w}$, as guaranteed by Lemma~\ref{lem:move-afterward}.
%
%
%
\end{proof}

We now combine \Cref{lem:all_elements_good_until_wcrit}, \Cref{lem:expected_dislocation_one_element}, and \Cref{lem:dependent}, to prove that our linear upper bound on the total dislocation of the sequence returned by \basketsort also holds with high probability.

\begin{lemma}\label{lem:total_dislocation_high}
    The total dislocation of $S_0$ is 
    at most $\left( \frac{1152}{\bsc^2 (1-\rho)} + \frac{311040}{(1-\rho)^4} + 1 \right) \cdot m$
    with probability at least $1 - O\left( \frac{\ln m}{(1-\rho) \bsc} \right) \cdot \exp\left( - \Omega\left(\frac{m (1-\rho)^3 \bsc}{\ln^3 m} \right) \right) - O\left( \frac{\log m}{m^4 \log \frac{1}{\rho}} \right).$ 
\end{lemma}
\begin{proof}
We start by computing a first lower bound to the probability that an execution of \basketsort with initial window size $w_S = w$, input sequence $S = S_w$, and $\disl(S) \le w$  returns a sequence having total dislocation $O(m)$.

To this aim, we assume that $w \ge 1$ (otherwise the total dislocation is trivially $0$) and we partition the elements of $S$ into $k = \left\lfloor \frac{75}{1-\rho} \right\rfloor w > \frac{74}{1-\rho} w = 2 \cdot \frac{37}{1-\rho} w$ sets $P^{(0)},\dots,P^{(k-1)}$, 
where $P^{(j)} = \{x \in S \, : \, x \bmod k = j \}$. 
We define 
$D^{(j)} = \sum_{x \in P^{(j)}} \disl(x, S_0)$ as the total dislocation of the elements in $P^{(j)}$ in the returned sequence $S_0$, and we call its expected value $\mu^{(j)} = \E[D^{(j)}] = \sum_{x \in P^{(j)}} \E[ \disl(x, S_0) ]$.

Observe that the final position $\sigma_0(x)$ of a generic element $x$ in $P^{(j)}$ does not depend on the comparison results that involve other elements in $P^{(j)}$.
Indeed, Lemma~\ref{lem:dependent} guarantees that $\sigma_0(x)$ only depends on $C_w(\sigma_w(x), \frac{36}{1-\rho} w)$ and, using  $|x - \sigma_w(x)| \le w$, we have that $\sigma_0(x)$ is a function of $C_w(x, \frac{37}{1-\rho} w)$. 
Therefore, since $k > 2 \cdot \frac{37}{1-\rho}w$, the set of the elements considered in $C_w(x, \frac{37}{1-\rho} w)$ is disjoint from the set of elements considered in $C(y, \frac{37}{1-\rho} w)$ for any $y \in P^{(j)} \setminus \{x\}$.

We now use Hoeffding's inequality to prove that  $D^{(j)}  \le \mu^{(j)} + \frac{m}{k}$ with high probability. 
    Lemma~\ref{lem:move-afterward} ensures that $\disl(x, S_0) \le \disl(x, S_w) + | \sigma_w(x) - \sigma_0(x) | \le w + \frac{8}{1-\rho}w  < \frac{9}{1-\rho}w$ for all elements $x$, hence we have:\footnote{We only consider the case $|P^{(j)}| > 0$, since otherwise $D^{(j)}= \mu^{(j)} = 0$ and  $D^{(j)}  \le \mu^{(j)} + \frac{m}{k}$ is trivially true.}
\begin{multline*}  
\Pr\left( D^{(j)} - \mu^{(j)}  \ge \frac{m}{k} \right) 
\le \exp\left(- \frac{\frac{2m^2}{k^2}}{|P^{(j)}|  \cdot (\frac{9}{1-\rho}w)^2} \right) 
\le \exp\left(- \frac{\frac{2m^2}{k^2}}{ 2\frac{m}{k}  \cdot (\frac{9}{1-\rho}w)^2} \right) \\ 
= \exp\left(- \frac{m (1-\rho)^2}{81 kw^2} \right)  
\le \exp\left(- \frac{m (1-\rho)^3}{6075 w^3} \right).
\end{multline*}

Using the union bound, we have that all $j=0, \dots, k-1$ satisfy $D^{(j)}\le  \mu^{(j)} + \frac{m}{k}$ with probability at least $1 - \frac{75}{1-\rho} w \cdot \exp\left(- \frac{m (1-\rho)^3}{6075w^3} \right)$. Whenever this happens, the total dislocation of $S_0$ is $\sum_{j=0}^{k-1} D^{(j)} < \sum_{j=0}^{k-1} \mu^{(j)} +  k \cdot \frac{m}{k} \leq  (\gamma+1)\cdot m$ with $\gamma = \frac{1152}{\bsc^2 (1-\rho)} + \frac{311040}{(1-\rho)^4}$, where we invoked Lemma~\ref{lem:expected_dislocation_one_element}. 

Notice that such a lower bound already proves the desired claim whenever $w = o\big( \sqrt[3]{\frac{m}{\ln m}}\big)$. The rest of the proof is devoted to handling arbitrary initial window sizes $w_S$.

Consider an execution of  \basketsort in which $w_S$ is at least the dislocation of the input sequence $S$ and let $w'$ be the smallest window size that satisfies $w' \ge \frac{6}{\bsc} \ln m$ or, if no such window size exists (i.e., if $w_S < \frac{6}{\bsc}\ln m$), let $w' = w_S$. 
Notice that, since $\lfloor \rho w' \rfloor < \frac{6}{\bsc} \ln m$, we have $w' < \frac{6}{\rho \bsc} \ln m + \frac{1}{\rho}  <  \frac{6}{\rho \bsc} \ln m + \frac{1}{\rho \bsc}  \ln m\le \frac{14}{\bsc} \ln m$, where the second inequality follows from $\bsc \le \frac{1}{3} \le \ln m$ for $m \ge 2$ and the third inequality follows from $\rho \ge \frac{1}{2}$. 

Denote by $\mathcal{S}$  the set of all sequences having dislocation at most $w'$ that are attainable from the input sequence $S$ with some choice of the observed comparisons, and let $\text{BS}(S')$ be an indicator random variable for the event ``the sequence returned by \basketsort with input sequence $S'$ and initial window size $w'$ has a total dislocation of at most $\left(\gamma+ 1 \right) \ln m$''.
    From the discussion in the first part of this proof, all $S' \in \mathcal{S}$ satisfy:
    \begin{align*}
        \Pr(\text{BS}(S') = 1) &\ge 1 - \frac{75}{1-\rho} w' \cdot \exp\left(- \frac{m (1-\rho)^3}{6075 (w')^3} \right) \\
        &\ge 1 - \frac{1050 \ln m}{(1-\rho) \bsc} \cdot \exp\left( - \frac{m(1-\rho)^3 \bsc}{16669800 \ln^3 m} \right).
    \end{align*}

Moreover, by Lemma~\ref{lem:all_elements_good_until_wcrit}, $\Pr(S_{w'} \in \mathcal{S}) = 1 - O\Paren{\frac{\log m}{m^4 \log \frac{1}{\rho}}}$. 
By observing that an execution of \basketsort with initial window size $w_S$ can be broken down into (i) all the rounds corresponding to window sizes from $w_S$ (included) to $w'$ (excluded) resulting in the sequence $S_{w'}$ and (ii) an execution of \basketsort on $S_{w'}$ with initial window size $w'$, we can write:\footnote{The collection of rounds considered in (i) might be empty.}

\begin{align*}
    &  \Pr\left(\sum_{x \in S} \disl(x, S_1)  \le (\gamma +1) m \right)
    \ge   \sum_{S' \in \mathcal{S}} \Pr(\text{BS}(S_{w'}) = 1 \mid S_{w'} = S') \cdot \Pr(S_{w'} = S') \\ 
    &\ge \left( 1 - \frac{1050 \ln m}{(1-\rho) \bsc} \cdot \exp\left(- \frac{m (1-\rho)^3 \bsc}{16669800 \ln^3 m} \right) \right) \cdot \Pr(S_w \in \mathcal{S}) \\
    &\ge \left( 1 - \frac{1050 \ln m}{(1-\rho) \bsc} \cdot \exp\left(- \frac{m (1-\rho)^3 \bsc}{16669800 \ln^3 m} \right) \right) \cdot \left(1 - O\left( \frac{\log m}{m^4 \log \frac{1}{\rho}} \right) \right) \\
    &\ge 1 - \frac{1050\ln m}{(1-\rho) \bsc} \cdot \exp\Paren{- \frac{m (1-\rho)^3 \bsc}{16669800 \ln^3 m}}  - O\left( \frac{\log m}{m^4 \log \frac{1}{\rho}} \right). 
\tag*{\qedhere}
\end{align*}
\end{proof}

Combining \Cref{lem:basketsort-time,lemma:bs_maximum-dislocation,lem:total_dislocation_high} we obtain the main result of this section:

\thmbasketsort

\section{Lower Bounds}\label{sec:lower_bound}

In this section, we derive a lower bound on the maximum dislocation that can be achieved \emph{with high probability} by any sorting algorithm (\Cref{thm:lb_max_dislocation}), and a lower bound on the total dislocation that can be achieved  \emph{in expectation} by any sorting algorithm (\Cref{thm:lb_tot_dislocation}) even when the probability of comparison errors between pairs of distinct elements are uniform, i.e., when $q=p$ for some constant $0 < p < \frac{1}{2}$. 

Throughout this section, we let set $S = \{1, 2, \dots, n\}$ be the \emph{set} of elements to be sorted.
We can think of an instance of our sorting problem as a pair $\langle \pi, C \rangle$ where $\pi$ is a permutation of $S$ and $C = (c_{i,j} )_{i,j}$ is an $n \times n$ matrix that encodes the outcomes of the comparisons as seen by the sorting algorithm. 
More precisely, $c_{i,j}=\text{``$\prec$''}$ if $i$ is reported to be smaller than $j$ when comparing $i$ and $j$ and $c_{i,j}=\text{``$\succ$''}$ otherwise.
Note that $c_{i,j}=\text{``$\prec$''}$ if and only if $c_{j,i}=\text{``$\succ$''}$, and hence in what follows, we will only define $c_{i,j}$ for $i < j$.

The following lemma is a key ingredient in the derivations of our lower bounds.

\begin{lemma}\label{lem:swap_probability}
	Let $x,y \in S$ with $x < y$. Let $A$ be any (possibly randomized) algorithm. On a random instance, the probability that $A$ returns a permutation in which
	elements $x$ and $y$ appear the wrong relative order is at least $\frac{1}{2} \left( \frac{p}{1-p} \right)^{2(y-x)-1}$. 
\end{lemma}
\begin{proof}
We prove that, for any running time $t>0$, no algorithm $A$ can compute, within $t$ steps, a sequence in which $x$ and $y$ are in the correct relative order with a probability larger than $1-\frac{1}{2} \left( \frac{p}{1-p} \right)^{2(y-x)-1}$.

First of all, let us focus a deterministic version of algorithm $A$ by fixing a sequence $\lambda \in \{0, 1\}^t$ of random bits that can be used be the algorithm. We call the resulting algorithm $A_\lambda$, and we let $p_\lambda$ denote the probability of generating the sequence $\lambda$ of random bits (i.e., $2^{-t}$ if the random bits come from a fair coin).\footnote{Notice that $A_\lambda$ can use less than $t$ random bits, but not more (due to the limit on its running time).}
Notice that $A$ might already be a deterministic algorithm, in this case $A = A_{\lambda}$ for every $\lambda \in \{0, 1\}^t$.

Consider an instance $I = \langle \pi, C \rangle$, and let $\phi_I$ be the permutation of the elements in $S$ returned by $A_\lambda(I)$, i.e., $\phi_I(i)=j$ if the element $i \in S$ is the $j$-th element of the sequence returned by $A_\lambda$.
We define a new instance $I_{\pi,C} = \langle \pi', C' \rangle$ by ``swapping'' elements $x$ and $y$ of $I$ along with the results of all their comparisons. More formally we define: $\pi'(i)=\pi(i)$ for all $i \in S \setminus \{x, y\}$, $\pi'(x)=\pi(y)$, and $\pi'(y)=\pi(x)$; We define the comparison matrix $C' = (c'_{i,j})_{i,j}$ accordingly, i.e.:%
\indent\begin{multicols}{2}  
	\begin{itemize}
		\item $c'_{i,j} = c_{i,j}$ if $i,j \in S \setminus \{x, y\}$ and $j>i$,
		\item $c'_{i,x} = c_{i,y}$ if $i < x$,
		\item $c'_{x,j} = c_{y,j}$ if $j>x$ and $j \neq y$,
		\item $c'_{i,y} = c_{i,x}$ if $i < y$ and $i \neq x$,
		\item $c'_{y,j} = c_{x,j}$ if $j > y$,
		\item $c'_{x,y} = c_{y, x}$.
	\end{itemize}
\end{multicols}%
  \noindent Letting $R = ( r_{i,j} )_{i,j}$ be a random variable over the space of all comparison matrices.
  Fix a pair $i,j$ with $i<j$ and observe that $\Pr(r_{i,j} = c'_{i,j}) = \Pr(r_{i,j} = c_{i,j})$ whenever we are in one of the following cases: (i) $i \neq x$ and $j \neq y$, (ii) $i<x$ and $j \in \{x,y\}$, or (iii) $j>y$ and $i \in \{x,y\}$.
  In the remaining cases, $\Pr(r_{i,j} = c'_{i,j}) =\Pr(r_{i,j} = c_{i,j}) \cdot  \frac{\Pr(r_{i,j} = c'_{i,j})}{\Pr(r_{i,j} = c_{i,j})}
  \ge \Pr(r_{i,j} = c_{i,j}) \cdot \frac{p}{1-p}$. 
  Then, $\Pr(R = C') = Pr(R = C) \cdot \left( \frac{p}{1-p} \right)^{2(y-x)+1}$.

Let $\widetilde{\pi}$ be a random variable whose value is chosen u.a.r. over all the permutations of $S$.
Let $\Pr(I)$ be the probability that instance $I=\langle \pi, C \rangle$ is to be solved and notice that $\Pr(I) = \Pr(\widetilde{\pi} = \pi) \Pr(R = C)$.
It follows from the above discussion that $\Pr(I_{\pi,c}) \ge \Pr( I ) \cdot \left( \frac{p}{1-p} \right)^{2(y-x)-1}$.

Let $X_I$ be an indicator random variable that is $1$ if and only if algorithm $A_\lambda$ on instance $I$ either does not terminate within $t$ steps, or it terminates returning a sequence in which $x$ appears after $y$. Let $U$ be the set of all the possible instances and 
$U'$ be the set of the instances $I \in U$ such that $X_I = 0$ (i.e., $x$ and $y$ appear in the correct order in the output of $A$ on $I$).
Notice that there is a bijection between the instances $I = \langle \pi, C \rangle$ and their corresponding instances $I_{\pi, C}$ (i.e., our transformation is injective). Moreover, since $I$ and $I_{\pi, C}$ are indistinguishable by $A_\lambda$, we have that
either (i) $A_\lambda$ does not terminate within $t$ steps on both instances, or (ii) $\sigma_I(x) < \sigma_I(y) \iff  \sigma_{I_{\pi, C}}(x) > \sigma_{I_{\pi, C}}(y)$. As a consequence, if $X_I = 0$, then  $X_{I_{\pi, C}} = 1$.
Now, either $\sum_{I \in U} X_I \Pr(I) \ge \frac{1}{2}$ or we must have $\sum_{I \in U'} \Pr(I) \ge \frac{1}{2}$, which implies:
\begin{align*}
	\sum_{I \in U} X_I \Pr(I)  &\ge \sum_{\langle \pi, C \rangle \in U} X_{I_{\pi, C}} \Pr(I_{\pi, C}) 
	\ge \sum_{\langle \pi, C \rangle \in U'} X_{I_{\pi, C}} \Pr(I_{\pi, C}) 
	= \sum_{\langle \pi, C \rangle \in U'} \Pr(I_{\pi, C}) \\
	& \ge \sum_{ I \in U'} \Pr( I ) \cdot \left( \frac{p}{1-p} \right)^{2(y-x)-1} 
	\ge \frac{1}{2} \left( \frac{p}{1-p} \right)^{2(y-x)-1}.
\end{align*}

We let $Y$ (resp.\ $Y_\lambda$) be an indicator random variable which is $1$ if and only if either the execution of $A$ (resp.\ $A_\lambda$) on a random instance does not terminate within $t$ steps, or it terminates returning a sequence in which $x$ and $y$ appear in the wrong relative order. 
By the above calculations we know that $\Pr(Y_\lambda = 1) = \sum_{I \in U} X_I \Pr(I) \ge \frac{1}{2} \left( \frac{p}{1-p} \right)^{2(y-x)-1} \; \forall \lambda \in \{0, 1\}^t$. We are now ready to bound the probability that $Y=1$, indeed:

\begin{align*}
\Pr(Y = 1) & = \sum_{ \lambda \in  \{0, 1\}^t} \Pr(Y_\lambda=1) p_\lambda
           \ge \frac{1}{2} \left( \frac{p}{1-p} \right)^{2(y-x)-1} \sum_{ \lambda \in \{0, 1\}^t } p_\lambda 
            = \frac{1}{2} \left( \frac{p}{1-p} \right)^{2(y-x)-1},
\end{align*}
\noindent where the last equality follows from the fact that $p_\lambda$ is a probability distribution over the elements $\lambda \in \{0, 1\}^t$, and hence  $\sum_{ \lambda \in \{0, 1\}^t } p_\lambda = 1$.
\end{proof}

\noindent As a first consequence of the previous lemma, we obtain the following:
\begin{theorem}\label{thm:lb_max_dislocation}
	No (possibly randomized) algorithm can achieve maximum dislocation $o(\log n)$  with high probability.
\end{theorem}
\begin{proof}
By Lemma~\ref{lem:swap_probability}, any algorithm, when invoked on a random instance, must return a permutation $\rho$ in which elements $1$ and $h= \big\lfloor  \frac{\log n }{2 \log (1-p)/{p}} \big\rfloor$ appear in the wrong relative order with a probability larger than $\frac{1}{n}$.
When this happens, at least one of the following two conditions holds: (i) the position of element $1$ in $\rho $ is at least $\lceil \frac{h}{2} \rceil$; or (ii) the position of element $h$ in $\rho$ is at most $\lfloor \frac{h}{2} \rfloor$.
In any case, the maximum dislocation must be at least $\frac{h}{2} - 1 = \Omega(\log n)$.
\end{proof}

\noindent Finally, we are also able to prove a lower bound to the total dislocation.
\begin{theorem}\label{thm:lb_tot_dislocation}
	No (possibly randomized) algorithm can achieve an expected total dislocation of $o(n)$.
\end{theorem}
\begin{proof}
Let $A$ be any algorithm and let $\langle \pi,C \rangle$ be a random instance on an even number of elements $n$.
We define $X_k$ for $0 \le k < n/2$ to be an indicator random variable that is $1$ if and only if $A(\langle \pi,C \rangle)$ returns a permutation of $S$ in which elements $2k+1$ and $2k+2$ appear in the wrong relative order.

By Lemma~\ref{lem:swap_probability} we know that $\Pr(X_k=1) \ge \frac{1}{2}\frac{p}{1-p}$. We can hence obtain a lower bound on the expected total dislocation $D$ achieved by $A$ as follows:
\[
\E[D] \ge \sum_{k=0}^{n/2 - 1} \E[X_k] \ge \sum_{k=0}^{n/2 - 1} \frac{1}{2} \cdot \frac{p}{1-p} \ge \frac{n}{4} \cdot \frac{p}{1-p} = \Omega(n). \qedhere
\]
\end{proof}

\section{Derandomization}\label{sec:derandomization}

In this section, we consider the case in which the elements to be approximately sorted are subject to persistent random comparison errors with probability at least $q$ and at most $p$, for some positive constants $p$ and $q$ with $p \le \frac{9q+1}{8q+11}$, and the positions of the elements in the input sequence do not depend on the comparison errors.
Under these assumptions, we can exploit the intrinsic random nature of comparison faults to derandomize \rifflesort while maintaining its $O(n\log n)$ running time. In other words, the sequence returned by the derandomized version of \rifflesort is completely determined by the input sequence and by the observed comparison errors, and no access to an external source of random bits is required.

This section is organized as follows: first we show that (a slight variant of) $\rifflesort$ can be implemented so that the overall number of random bits required is $O(n)$ with high probability. Then, we show that we can use the outcomes of the noisy comparisons among some of the input elements to generate $\Omega(n)$ biased, but almost-uniform random bits, and that using these almost-uniform random bits in place of the truly uniform ones does not affect the asymptotic guarantees of \rifflesort. 

\subsection{Implementing \rifflesort using \texorpdfstring{$O(n)$}{O(n)} random bits w.h.p.}
\label{sub:rifflesort_n_bits}

The only randomized steps of \rifflesort are the computation of the sequence $S'$, which is obtained by shuffling the input sequence $S$ (see line~\ref{ln:rs_shuffle} of \Cref{alg:rifflesort}), and the partition of the $n$ input elements into the $k+1$ 
subsets $T_0, T_1, \dots, T_k$ of $T$, as described in \Cref{sec:rifflesort_description}.

Since we are considering the case in which the positions of the elements in the input sequence do not depend on the comparison errors, we can simply skip the shuffling step of line~\ref{ln:rs_shuffle}, and set $S' = S$ instead.
We now argue that the partition $T_0, \dots, T_k$ can be found while only using $O(n)$ random bits. In this regard, we will make use of the following more general result, whose proof does not depend on the details of \rifflesort nor on our error model and can be found in \Cref{apx:random_subset}.

\begin{restatable}{lemma}{lemmarandomsubset}
  \label{lemma:derand_set_split}
    There exists an algorithm that, given any set $A$ of $N$ elements and an integer $h \in \{0, \dots, N\}$, returns a set chosen uniformly at random among all subsets of $A$ with exactly $h$ elements.
    The algorithm requires $O(N)$ time and at most $60N$ random bits with probability at least $1 - O\left( e^{-\sqrt{N}} \right)$.
\end{restatable}

Then, using the fact that $|T_i| = \frac{1}{2}|T_{i+1}|$ for $i=0, \dots, k-1$, we can prove the following:

\begin{lemma}
  \label{lemma:partition_n_bits}
    The partition $T_0, T_1, \dots, T_k$ of \rifflesort can be found using $O(n)$ time and at most $999 n$ random bits with probability at least $1 - O(n^{-3})$.
\end{lemma}
\begin{proof}
    Recall from \Cref{sec:rifflesort_description} that $k=\left\lceil \log \frac{n}{\lceil \sqrt{n} \rceil} \right\rceil < 1 + \frac{1}{2} \log n$, that $|T_i| = 2^i \cdot \lceil \sqrt{n} \rceil$ for $i=0, \dots,k-1$, and that $T_k \le 2^k \cdot \lceil \sqrt{n} \rceil$. 
    
    To obtain the sought partition, we iteratively select $T_i$ for starting form $i=k$ down to $i=0$.
    More precisely, we keep track of the set $T'_i = S \setminus \bigcup_{j=i}^k T_i$ of all yet unselected elements (notice that $T'_{k+1} = S$), and we select $T_i$ u.a.r.\ among all subsets of $T'_{i+1}$ having $|T_i|$ elements using the algorithm of \Cref{lemma:derand_set_split}.

    We have $|T'_i| = \sum_{j=0}^{i-1} |T_i| < 2^{i} \lceil \sqrt{n} \rceil \le 2^{i+1} \sqrt{n}$, and, using $|T_i| \ge \sqrt{n}$, \Cref{lemma:derand_set_split} ensures that the selection of each $T_i$ requires no more than $60 |T'_{i+1}| \le 60 \cdot 2^{i+2} \sqrt{n}$ random bits with probability at least $1 - O( e^{- n^{1/4}} )$.
    Then, this holds jointly for all $T_i$ with probability at least $1 - O( k \cdot  e^{- n^{1/4}} ) = 1 - O(n^{-3})$ and, in such a case, the overall number of used random bits is upper bounded by:
    \[
        60 \sqrt{n} \sum_{i=0}^{k} 2^{i+2} < 60 \sqrt{n} \cdot 2^{k+3} < 60 \sqrt{n} \cdot 2^{4+\frac{1}{2}\log n} < 999 n, 
    \]
    and a similar calculation yields the claimed upper bound on the running time.
\end{proof}

\subsection{Derandomizing \rifflesort}

We start by formalizing a simple technique that combines multiple biased coin flips to simulate the outcome of a single coin with a smaller bias.

\begin{lemma}
  \label{lemma:xoring-bernoulli}
  Let $\delta \in [0, 1]$ and let $X_1, X_2, \dots, X_\eta$ be $\eta \ge 1$ independent Bernoulli random variables where $X_i$ has parameter $p_i \in \left[\frac{1}{2} - \frac{\delta}{2}, \frac{1}{2} + \frac{\delta}{2} \right]$.
  $\Pr( \bigoplus_{i=1}^\eta X_i = 1 ) \in \left[ \frac{1}{2} - \frac{\delta^\eta}{2}, \frac{1}{2} + \frac{\delta^\eta}{2} \right]$.
\end{lemma}
\begin{proof}
  The proof is by induction on $\eta$.
  The base case $\eta=1$ is trivially true. Suppose now that the claim is true for $\eta-1 \ge 1$ Bernoulli random variables. Defining $\rho = \Pr(\bigoplus_{i=1}^{\eta-1} X_i = 1)$, we have:
  \begin{align*}
    \Pr\left( \bigoplus_{i=1}^\eta X_i = 1 \right) &= 
    \Pr\left( \bigoplus_{i=1}^{\eta-1} X_i  = 0 \right) \Pr\left( X_\eta=1 \right) + \Pr\left( \bigoplus_{i=1}^{\eta-1} X_i  = 1  \right) \Pr\left( X_\eta=0 \right) \\  
    & = (1-\rho) p_\eta + \rho (1-p_\eta) 
     = \frac{1}{2} + 2 \left(\frac{1}{2} - p_\eta\right) \left(\rho - \frac{1}{2}\right).
  \end{align*}
  Therefore, using $\rho \in [\frac{1}{2} - \frac{\delta^{\eta-1}}{2}, \frac{1}{2} + \frac{\delta^{\eta-1}}{2}]$ as ensured by the induction hypothesis:
  \[
    \left| \Pr\left( \bigoplus_{i=1}^\eta X_i = 1 \right) - \frac{1}{2} \right|  =  2 \left|  \frac{1}{2} - p_\eta \right| \cdot \left| \rho - \frac{1}{2} \right| \le 2 \cdot \frac{\delta}{2} \cdot  \frac{\delta^{\eta-1}}{2}
    = \frac{\delta^\eta}{2}. \qedhere
  \]
\end{proof}

Since each comparison yields each of the two possible outcomes with probability at least $q \in (0, \frac{1}{2})$, we can associate these outcomes with the values $0$ and $1$ so that the comparison behaves as a Bernoulli random variable of some unknown parameter in $[q, 1-q] = [\frac{1}{2} - \frac{1-2q}{2}, \frac{1}{2} + \frac{1-2q}{2} ]$.
Then, \Cref{lemma:xoring-bernoulli} ensures that xor-ing together $\eta = \lceil 4 \log_{1-2q} \frac{1}{n}  \rceil = \left\lceil \frac{4 \log n}{\log \frac{1}{1-2q}} \right\rceil = O(\log n)$ independent comparison results yields an ``almost fair'' random bit $b$, i.e., $\Pr(b = 1) \in \left[\frac{1}{2} - \frac{1}{2n^4}, \frac{1}{2} + \frac{1}{2n^4}\right]$.


In the rest of this section, we consider values of $n$ larger than some suitable constant. We let $F$ be a set containing the first $1000 \eta$ elements from $S$, and we compare each element in $F$ with all the elements in $F'=S \setminus F$ to obtain a collection of $1000 \eta (n - \eta) \ge 999 \eta n$ independent comparison outcomes, from which we can generate $999 n$ almost-fair random bits. The overall time required to generate these almost-fair bits is $O(\eta n) = O(n \log n)$ and,
with probability at least $1 - \frac{999}{n^3}$, these bits behave exactly as unbiased random bits. 
By invoking \rifflesort on $F'$ using these almost-fair random bits as a source of randomness, we obtain an approximately sorted sequence $S'$. In the unlikely event that more random bits are required by \rifflesort, we simply halt the algorithm and return an arbitrary permutation of $F'$.
From the discussion in \Cref{sub:rifflesort_n_bits}, \Cref{lemma:partition_n_bits}, \Cref{lemma:xoring-bernoulli}, and \Cref{thm:rifflesort}, the overall running time is $O(n \log n)$ and, with probability at least $1 - \frac{999}{n^3} - \frac{1}{|F'|^{3}} - O\Big(\frac{1}{|F'| \sqrt{|F'|}}\Big) = 1 - O\big( \frac{1}{n \sqrt{n}} \big)$, the $999n$ almost-fair random bits suffice to run \rifflesort on $F'$, and the returned sequence $S'$ has a maximum (resp.\ total) dislocation of at most $O(\log n)$ (resp.\ $O(n)$). 

We are now left with the task of reinserting all the elements of $F$ into $S'$ without causing any asymptotic increase in the total dislocation and in the maximum dislocation. 
Although one might be tempted to use the result of Section~\ref{sec:noisy-search}, this is not actually possible since the comparison faults between the elements in $F$ and the elements in $F'$ now (indirectly) depend on the sequence $S'$.
However, a simple (but slower) strategy, which is similar to the one used in \cite{GeissmannLLP20}, works even when the sequence $S'$ is adversarially chosen as a function of the errors.
Let $c$ be an integer constant such that $c \ge \frac{51}{p}$ and $c > \frac{7-8p}{1-4p}$ (recall that our assumptions on $p$ and $q$ ensure that $p < \frac{1}{4}$), and let $d$ be an integer that satisfies $3 \log n \le d = O(\log n)$ and that is  an upper bound on the maximum dislocation of the sequence $S'$ returned by a successful execution of \rifflesort.

For each element $x\in F$, given a guess $\tilde{r}$ on $\rank(x, S')$, we can determine whether $\tilde{r}$ is a good estimate of $\rank(x, S')$ by comparing $x$ with all elements in the set $C_x = \{ y \in S' \mid \tilde{r} - cd \le \pos(y, S') < \tilde{r} + cd \}$,
and counting the number $m_x(\tilde{r})$ of \emph{mismatches}, i.e., the number of elements $y \in C_x$ for which either (i) $\pos(y, S') < \tilde{r}$ and $y \succ x$, or (ii) $\pos(y, S') \ge \tilde{r}$ and $y \prec x$.

Suppose that our guess of $\tilde{r}$ is much smaller than the true rank of $x$, say $\tilde{r} < \rank(x, S') - c d$. 
Then, all the elements $y \in S'$ such that $\tilde{r} + d \le \rank(y, S') < \tilde{r} + (c-1)d$
are smaller than $x$, satisfy $\pos(y, S') \ge \tilde{r}$, and belong to $C_x$.
Then, for any such $y$, the observed comparison result is $y \prec x$ with probability at least $1-p$ and the number $m_x(\tilde{r})$ of mismatches is at least $(c-2)d(1-p)$ in expectation.
A Chernoff bound shows that $\Pr( m_x(\tilde{r}) \le \frac{4}{5} (c-2)d(1-p) ) \le e^{-\frac{(c-2) d (1-p)}{50}}  \le e^{-d} \le \frac{1}{n^3}$, hence $m_x(\tilde{r}) > \frac{4}{5} (c-2) d (1-p)$ with probability at least $1 - \frac{1}{n^3}$.
A symmetric argument holds for the case in which $\tilde{r} \ge \rank(x, S') + c d$.

%
%
%

On the contrary, if $\rank(x, S') - d \le \tilde{r} < \rank(x, S') + d$, 
all the elements $y \in S'$ for which
$\tilde{r} - (c-1) d \le \rank(y, S') < \tilde{r} - d$
(resp.\ $\tilde{r} +  d \le \rank(y, S') < \tilde{r} + (c-1) d$)
are smaller (resp.\ larger) than $x$, satisfy $\pos(y, S') \le \rank(y, S') + d < \tilde{r}$ (resp.\ $\pos(y, S') \ge  \rank(y, S') -d \ge \tilde{r}$), and belong to $C_x$.
Then, for any such $y$, the observed comparison result is $y \succ x$ (resp.\ $y \prec x$) with probability at most $p$, and thus, the overall number $m'_x(\tilde{r})$ of mismatches involving these $y$s is stochastically dominated by a binomial random variable $M'$ of parameters $2(c-2)d$ and $p$.
A Chernoff bound shows that 
$\Pr(m'_x(\tilde{r}) \ge \frac{6}{5} 2(c-2)dp)  
\le
\Pr(M' \ge \frac{6}{5} \E[M'])  \le e^{-\frac{2(c-2)dp}{75}} \le e^{-d} \le \frac{1}{n^3}$.
Hence, with probability at least $1 - \frac{1}{n^3}$, we have $m_x(\tilde{r}) < m'_x(\tilde{r}) + 4d \le (\frac{12}{5} (c-2) p  + 4) d$.

%

We have therefore shown that, with probability at least $ 1- \frac{1}{n^2}$, the number of mismatches will be at least $\frac{4}{5} (c-2) d (1-p)$ whenever $\tilde{r} < \rank(x, S') - c d$ or $\tilde{r} \ge \rank(x, S') + c d$, and at most   $(\frac{12}{5} (c-2) p  + 4) d$ when $\rank(x, S') - d \le \tilde{r} < \rank(x, S') + d$.
Since our choice of $c$ ensures that  $\frac{4}{5} (c-2) (1-p) > (\frac{12}{5} (c-2) p  + 4)$, 
we can compute a rank $r_x$ satisfying $|r_x - \rank(x, S')| = O(d)$ in $O(d \cdot \frac{n}{d}) = O(n)$ time by first counting the number of mismatches $m_x(\tilde{r})$ for all choices of $\tilde{r} \in \{ 1, d + 1, 2d + 1, 3d + 1,\dots \}$ and then selecting the value of $\tilde{r}$ minimizing $m_x(\cdot)$. 
The total time required to compute all ranks $r_x$, for all $x \in F$, is $|F| \cdot O(n) = O(n \log n)$, and the success probability is at least $1 - |F| \cdot \frac{1}{n^2} = 1 - O( \frac{\log n}{n^2})$.
Combining this with the success probability of \rifflesort, we obtain an overall success probability of at least $1 - O( \frac{1}{n \sqrt{n}} )$.
Finally,  observe that, since the set $F$ only contains $O(\log n)$ elements, simultaneously reinserting them in $S'$ affects the maximum dislocation of $S'$ by at most an $O(\log n)$ additive term. Moreover, their combined contribution to the total dislocation can be at most $O(\log^2 n)$.
Overall, we have:

\begin{theorem}\label{thm:derandomized_rifflesort}
    Consider a set of $n$ elements subject to random persistent comparison faults with probability at least $q$ and at most $p$, for some constants $p,q$ with $0 < q \le p < \frac{9q+1}{8q+11}$. Let $S$ be a sequence of these elements chosen independently of the errors.
  There exists a deterministic algorithm that approximately sorts $S$ in $O(n \log n)$  worst-case time. The maximum dislocation and the total dislocation of the returned sequence are at most $O(\log n)$ 
  and at most $O(n)$, 
    respectively, with high probability.
\end{theorem}

\bibliographystyle{ACM-Reference-Format}
\bibliography{references}


\begin{thebibliography}{31}


\ifx \showCODEN    \undefined \def \showCODEN     #1{\unskip}     \fi
\ifx \showDOI      \undefined \def \showDOI       #1{#1}\fi
\ifx \showISBNx    \undefined \def \showISBNx     #1{\unskip}     \fi
\ifx \showISBNxiii \undefined \def \showISBNxiii  #1{\unskip}     \fi
\ifx \showISSN     \undefined \def \showISSN      #1{\unskip}     \fi
\ifx \showLCCN     \undefined \def \showLCCN      #1{\unskip}     \fi
\ifx \shownote     \undefined \def \shownote      #1{#1}          \fi
\ifx \showarticletitle \undefined \def \showarticletitle #1{#1}   \fi
\ifx \showURL      \undefined \def \showURL       {\relax}        \fi
\providecommand\bibfield[2]{#2}
\providecommand\bibinfo[2]{#2}
\providecommand\natexlab[1]{#1}
\providecommand\showeprint[2][]{arXiv:#2}

\bibitem[Addanki et~al\mbox{.}(2021)]%
        {AddankiGS21}
\bibfield{author}{\bibinfo{person}{Raghavendra Addanki},
  \bibinfo{person}{Sainyam Galhotra}, {and} \bibinfo{person}{Barna Saha}.}
  \bibinfo{year}{2021}\natexlab{}.
\newblock \showarticletitle{How to Design Robust Algorithms using Noisy
  Comparison Oracle}.
\newblock \bibinfo{journal}{\emph{Proceedings of the VLDB Endowment}}
  \bibinfo{volume}{14}, \bibinfo{number}{10} (\bibinfo{year}{2021}),
  \bibinfo{pages}{1703–1716}.
\newblock


\bibitem[Ailon et~al\mbox{.}(2008)]%
        {AilonCN08}
\bibfield{author}{\bibinfo{person}{Nir Ailon}, \bibinfo{person}{Moses
  Charikar}, {and} \bibinfo{person}{Alantha Newman}.}
  \bibinfo{year}{2008}\natexlab{}.
\newblock \showarticletitle{Aggregating inconsistent information: Ranking and
  clustering}.
\newblock \bibinfo{journal}{\emph{J. ACM}} \bibinfo{volume}{55},
  \bibinfo{number}{5} (\bibinfo{year}{2008}), \bibinfo{pages}{1--27}.
\newblock


\bibitem[Ajtai et~al\mbox{.}(2016)]%
        {AjtaiFHN16}
\bibfield{author}{\bibinfo{person}{Mikl{\'o}s Ajtai}, \bibinfo{person}{Vitaly
  Feldman}, \bibinfo{person}{Avinatan Hassidim}, {and} \bibinfo{person}{Jelani
  Nelson}.} \bibinfo{year}{2016}\natexlab{}.
\newblock \showarticletitle{Sorting and selection with imprecise comparisons}.
\newblock \bibinfo{journal}{\emph{ACM Transactions on Algorithms}}
  \bibinfo{volume}{12}, \bibinfo{number}{2} (\bibinfo{year}{2016}),
  \bibinfo{pages}{19}.
\newblock


\bibitem[Alon(2006)]%
        {Alon06}
\bibfield{author}{\bibinfo{person}{Noga Alon}.}
  \bibinfo{year}{2006}\natexlab{}.
\newblock \showarticletitle{Ranking Tournaments}.
\newblock \bibinfo{journal}{\emph{{SIAM} Journal on Discrete Mathematics}}
  \bibinfo{volume}{20}, \bibinfo{number}{1} (\bibinfo{year}{2006}),
  \bibinfo{pages}{137--142}.
\newblock


\bibitem[Alon et~al\mbox{.}(2009)]%
        {AlonLS09}
\bibfield{author}{\bibinfo{person}{Noga Alon}, \bibinfo{person}{Daniel
  Lokshtanov}, {and} \bibinfo{person}{Saket Saurabh}.}
  \bibinfo{year}{2009}\natexlab{}.
\newblock \showarticletitle{Fast {FAST}}. In
  \bibinfo{booktitle}{\emph{Proceedings of the 36th International Colloquium on
  Automata, Languages and Programming ({ICALP}09)}}. \bibinfo{pages}{49--58}.
\newblock


\bibitem[Ben{-}Or and Hassidim(2008)]%
        {Ben-OrH08}
\bibfield{author}{\bibinfo{person}{Michael Ben{-}Or} {and}
  \bibinfo{person}{Avinatan Hassidim}.} \bibinfo{year}{2008}\natexlab{}.
\newblock \showarticletitle{The Bayesian Learner is Optimal for Noisy Binary
  Search (and Pretty Good for Quantum as Well)}. In
  \bibinfo{booktitle}{\emph{Proceedings of the 49th Symposium on Foundations of
  Computer Science ({FOCS}08)}}. \bibinfo{pages}{221--230}.
\newblock


\bibitem[Beretta et~al\mbox{.}(2023)]%
        {BerettaNTV23}
\bibfield{author}{\bibinfo{person}{Lorenzo Beretta},
  \bibinfo{person}{Franco~Maria Nardini}, \bibinfo{person}{Roberto Trani},
  {and} \bibinfo{person}{Rossano Venturini}.} \bibinfo{year}{2023}\natexlab{}.
\newblock \showarticletitle{An Optimal Algorithm for Finding Champions in
  Tournament Graphs}.
\newblock \bibinfo{journal}{\emph{IEEE Transactions on Knowledge and Data
  Engineering}} \bibinfo{volume}{35}, \bibinfo{number}{10}
  (\bibinfo{year}{2023}), \bibinfo{pages}{10197--10209}.
\newblock


\bibitem[Bianchi and Penna(2021)]%
        {BianchiP21}
\bibfield{author}{\bibinfo{person}{Enrico Bianchi} {and} \bibinfo{person}{Paolo
  Penna}.} \bibinfo{year}{2021}\natexlab{}.
\newblock \showarticletitle{Optimal clustering in stable instances using
  combinations of exact and noisy ordinal queries}.
\newblock \bibinfo{journal}{\emph{Algorithms}} \bibinfo{volume}{14},
  \bibinfo{number}{2} (\bibinfo{year}{2021}), \bibinfo{pages}{55}.
\newblock


\bibitem[Braverman et~al\mbox{.}(2016)]%
        {BravermanMW16}
\bibfield{author}{\bibinfo{person}{Mark Braverman}, \bibinfo{person}{Jieming
  Mao}, {and} \bibinfo{person}{S.~Matthew Weinberg}.}
  \bibinfo{year}{2016}\natexlab{}.
\newblock \showarticletitle{Parallel algorithms for select and partition with
  noisy comparisons}. In \bibinfo{booktitle}{\emph{Proceedings of the
  Forty-eighth48th Symposium on Theory of Computing ({STOC}16)}}.
  \bibinfo{pages}{851--862}.
\newblock


\bibitem[Braverman and Mossel(2008)]%
        {BravermanM08}
\bibfield{author}{\bibinfo{person}{Mark Braverman} {and}
  \bibinfo{person}{Elchanan Mossel}.} \bibinfo{year}{2008}\natexlab{}.
\newblock \bibinfo{title}{Noisy sorting without resampling}.
\newblock , \bibinfo{numpages}{268-276}~pages.
\newblock


\bibitem[Brown(2011)]%
        {brown2011wasted}
\bibfield{author}{\bibinfo{person}{Daniel~G Brown}.}
  \bibinfo{year}{2011}\natexlab{}.
\newblock \showarticletitle{How I wasted too long finding a concentration
  inequality for sums of geometric variables}.
\newblock  (\bibinfo{year}{2011}).
\newblock
\urldef\tempurl%
\url{https://cs.uwaterloo.ca/~browndg/negbin.pdf}
\showURL{%
\tempurl}


\bibitem[Charbit et~al\mbox{.}(2007)]%
        {CharbitTY07}
\bibfield{author}{\bibinfo{person}{Pierre Charbit},
  \bibinfo{person}{St{\'{e}}phan Thomass{\'{e}}}, {and} \bibinfo{person}{Anders
  Yeo}.} \bibinfo{year}{2007}\natexlab{}.
\newblock \showarticletitle{The Minimum Feedback Arc Set Problem is {NP-Hard}
  for Tournaments}.
\newblock \bibinfo{journal}{\emph{Combinatorics, Probability and Computing}}
  \bibinfo{volume}{16}, \bibinfo{number}{1} (\bibinfo{year}{2007}),
  \bibinfo{pages}{1--4}.
\newblock


\bibitem[Cicalese(2013)]%
        {Cicalese13}
\bibfield{author}{\bibinfo{person}{Ferdinando Cicalese}.}
  \bibinfo{year}{2013}\natexlab{}.
\newblock \bibinfo{booktitle}{\emph{Fault-Tolerant Search Algorithms - Reliable
  Computation with Unreliable Information}}.
\newblock \bibinfo{publisher}{Springer}.
\newblock


\bibitem[Damaschke(2016)]%
        {Damaschke16}
\bibfield{author}{\bibinfo{person}{Peter Damaschke}.}
  \bibinfo{year}{2016}\natexlab{}.
\newblock \showarticletitle{The Solution Space of Sorting with Recurring
  Comparison Faults}. In \bibinfo{booktitle}{\emph{Proceedings of the 27th
  International Workshop on Combinatorial Algorithms ({IWOCA}16)}}.
  \bibinfo{pages}{397--408}.
\newblock


\bibitem[Feige et~al\mbox{.}(1994)]%
        {FeigeRPU94}
\bibfield{author}{\bibinfo{person}{Uriel Feige}, \bibinfo{person}{Prabhakar
  Raghavan}, \bibinfo{person}{David Peleg}, {and} \bibinfo{person}{Eli Upfal}.}
  \bibinfo{year}{1994}\natexlab{}.
\newblock \showarticletitle{Computing with Noisy Information}.
\newblock \bibinfo{journal}{\emph{SIAM J. Comput.}} \bibinfo{volume}{23},
  \bibinfo{number}{5} (\bibinfo{year}{1994}), \bibinfo{pages}{1001--1018}.
\newblock


\bibitem[Feller(1957)]%
        {feller1957introduction}
\bibfield{author}{\bibinfo{person}{William Feller}.}
  \bibinfo{year}{1957}\natexlab{}.
\newblock \bibinfo{booktitle}{\emph{An introduction to probability theory and
  its applications} (\bibinfo{edition}{2} ed.)}.
\newblock \bibinfo{publisher}{John Wiley \& Sons}.
\newblock


\bibitem[Geissmann et~al\mbox{.}(2017)]%
        {GeissmannLLP17}
\bibfield{author}{\bibinfo{person}{Barbara Geissmann}, \bibinfo{person}{Stefano
  Leucci}, \bibinfo{person}{Chih{-}Hung Liu}, {and} \bibinfo{person}{Paolo
  Penna}.} \bibinfo{year}{2017}\natexlab{}.
\newblock \showarticletitle{Sorting with Recurrent Comparison Errors}. In
  \bibinfo{booktitle}{\emph{Proceedings of the Twenty-Eighth International
  Symposium on Algorithms and Computation ({ISAAC}17)}}.
  \bibinfo{pages}{38:1--38:12}.
\newblock


\bibitem[Geissmann et~al\mbox{.}(2019)]%
        {GeissmannLLP19}
\bibfield{author}{\bibinfo{person}{Barbara Geissmann}, \bibinfo{person}{Stefano
  Leucci}, \bibinfo{person}{Chih{-}Hung Liu}, {and} \bibinfo{person}{Paolo
  Penna}.} \bibinfo{year}{2019}\natexlab{}.
\newblock \showarticletitle{Optimal Sorting with Persistent Comparison Errors}.
  In \bibinfo{booktitle}{\emph{Proceedings of the Twenty-seventh European
  Symposium on Algorithms ({ESA}19)}}. \bibinfo{pages}{49:1--49:14}.
\newblock


\bibitem[Geissmann et~al\mbox{.}(2020)]%
        {GeissmannLLP20}
\bibfield{author}{\bibinfo{person}{Barbara Geissmann}, \bibinfo{person}{Stefano
  Leucci}, \bibinfo{person}{Chih{-}Hung Liu}, {and} \bibinfo{person}{Paolo
  Penna}.} \bibinfo{year}{2020}\natexlab{}.
\newblock \showarticletitle{Optimal Dislocation with Persistent Errors in
  Subquadratic Time}.
\newblock \bibinfo{journal}{\emph{Theory Comput. Syst.}} \bibinfo{volume}{64},
  \bibinfo{number}{3} (\bibinfo{year}{2020}), \bibinfo{pages}{508--521}.
\newblock


\bibitem[Gu and Xu(2023)]%
        {GuX23}
\bibfield{author}{\bibinfo{person}{Yuzhou Gu} {and} \bibinfo{person}{Yinzhan
  Xu}.} \bibinfo{year}{2023}\natexlab{}.
\newblock \showarticletitle{Optimal Bounds for Noisy Sorting}. In
  \bibinfo{booktitle}{\emph{Proceedings of the 55th Symposium on Theory of
  Computing ({STOC}23)}}. \bibinfo{pages}{1502--1515}.
\newblock


\bibitem[Hopkins et~al\mbox{.}(2020)]%
        {HopkinsKLM20}
\bibfield{author}{\bibinfo{person}{Max Hopkins}, \bibinfo{person}{Daniel Kane},
  \bibinfo{person}{Shachar Lovett}, {and} \bibinfo{person}{Gaurav Mahajan}.}
  \bibinfo{year}{2020}\natexlab{}.
\newblock \showarticletitle{Noise-tolerant, reliable active classification with
  comparison queries}. In \bibinfo{booktitle}{\emph{Proceedings of 33-rd
  Conference on Learning Theory ({COLT}20)}}. PMLR,
  \bibinfo{pages}{1957--2006}.
\newblock


\bibitem[Kenyon-Mathieu and Schudy(2007)]%
        {Kenyon-MathieuS07}
\bibfield{author}{\bibinfo{person}{Claire Kenyon-Mathieu} {and}
  \bibinfo{person}{Warren Schudy}.} \bibinfo{year}{2007}\natexlab{}.
\newblock \showarticletitle{How to rank with few errors}. In
  \bibinfo{booktitle}{\emph{Proceedings of the Thirty-nineth Symposium on
  Theory of Computing ({STOC}07)}}. \bibinfo{pages}{95--103}.
\newblock


\bibitem[Klein et~al\mbox{.}(2011)]%
        {KleinPSW11}
\bibfield{author}{\bibinfo{person}{Rolf Klein}, \bibinfo{person}{Rainer
  Penninger}, \bibinfo{person}{Christian Sohler}, {and}
  \bibinfo{person}{David~P. Woodruff}.} \bibinfo{year}{2011}\natexlab{}.
\newblock \showarticletitle{Tolerant algorithms}. In
  \bibinfo{booktitle}{\emph{Proceedings of the Nineteenth European Symposium on
  Algorithms ({ESA}11)}}. \bibinfo{pages}{736--747}.
\newblock


\bibitem[Knuth(1998)]%
        {Knuth98fisheryates}
\bibfield{author}{\bibinfo{person}{Donald~Ervin Knuth}.}
  \bibinfo{year}{1998}\natexlab{}.
\newblock \bibinfo{booktitle}{\emph{The art of computer programming, Volume
  {II:} Seminumerical Algorithms, 3rd Edition}}.
\newblock \bibinfo{publisher}{Addison-Wesley}. 12--15 pages.
\newblock
\showISBNx{0201896842}
\urldef\tempurl%
\url{https://www.worldcat.org/oclc/312898417}
\showURL{%
\tempurl}


\bibitem[Kunisky et~al\mbox{.}(2025)]%
        {KuniskySWY25}
\bibfield{author}{\bibinfo{person}{Dmitriy Kunisky}, \bibinfo{person}{Daniel~A.
  Spielman}, \bibinfo{person}{Alexander~S. Wein}, {and} \bibinfo{person}{Xifan
  Yu}.} \bibinfo{year}{2025}\natexlab{}.
\newblock \showarticletitle{Statistical Inference of a Ranked Community in a
  Directed Graph}. In \bibinfo{booktitle}{\emph{Proceedings of the 57th Annual
  {ACM} Symposium on Theory of Computing, {STOC} 2025, Prague, Czechia, June
  23-27, 2025}}, \bibfield{editor}{\bibinfo{person}{Michal Kouck{\'{y}}} {and}
  \bibinfo{person}{Nikhil Bansal}} (Eds.). \bibinfo{publisher}{{ACM}},
  \bibinfo{pages}{2107--2117}.
\newblock
\urldef\tempurl%
\url{https://doi.org/10.1145/3717823.3718280}
\showDOI{\tempurl}


\bibitem[Kunisky et~al\mbox{.}(2024)]%
        {KuniskyInference}
\bibfield{author}{\bibinfo{person}{Dmitriy Kunisky}, \bibinfo{person}{Daniel~A.
  Spielman}, {and} \bibinfo{person}{Xifan Yu}.}
  \bibinfo{year}{2024}\natexlab{}.
\newblock \showarticletitle{Inference of rankings planted in random
  tournaments}.
\newblock \bibinfo{journal}{\emph{CoRR}}  \bibinfo{volume}{abs/2407.16597}
  (\bibinfo{year}{2024}).
\newblock
\urldef\tempurl%
\url{https://doi.org/10.48550/ARXIV.2407.16597}
\showDOI{\tempurl}
\showeprint[arXiv]{2407.16597}


\bibitem[Pelc(2002)]%
        {Pelc02}
\bibfield{author}{\bibinfo{person}{Andrzej Pelc}.}
  \bibinfo{year}{2002}\natexlab{}.
\newblock \showarticletitle{Searching games with errors - fifty years of coping
  with liars}.
\newblock \bibinfo{journal}{\emph{Theoretical Computer Science}}
  \bibinfo{volume}{270}, \bibinfo{number}{1-2} (\bibinfo{year}{2002}),
  \bibinfo{pages}{71--109}.
\newblock


\bibitem[R{\'e}nyi(1961)]%
        {Renyi61}
\bibfield{author}{\bibinfo{person}{Alfr{\'e}d R{\'e}nyi}.}
  \bibinfo{year}{1961}\natexlab{}.
\newblock \showarticletitle{On a problem of information theory}.
\newblock \bibinfo{journal}{\emph{MTA Mat. Kut. Int. Kozl. B}}
  \bibinfo{volume}{6} (\bibinfo{year}{1961}), \bibinfo{pages}{505--516}.
\newblock


\bibitem[Rivest et~al\mbox{.}(1980)]%
        {RivestMKWS80}
\bibfield{author}{\bibinfo{person}{Ronald~L. Rivest},
  \bibinfo{person}{Albert~R. Meyer}, \bibinfo{person}{Daniel~J. Kleitman},
  \bibinfo{person}{Karl Winklmann}, {and} \bibinfo{person}{Joel Spencer}.}
  \bibinfo{year}{1980}\natexlab{}.
\newblock \showarticletitle{Coping with Errors in Binary Search Procedures}.
\newblock \bibinfo{journal}{\emph{J. Comput. System Sci.}}
  \bibinfo{volume}{20}, \bibinfo{number}{3} (\bibinfo{year}{1980}),
  \bibinfo{pages}{396--404}.
\newblock


\bibitem[Skala(2013)]%
        {Skala13}
\bibfield{author}{\bibinfo{person}{Matthew Skala}.}
  \bibinfo{year}{2013}\natexlab{}.
\newblock \showarticletitle{Hypergeometric tail inequalities: ending the
  insanity}.
\newblock \bibinfo{journal}{\emph{CoRR}}  \bibinfo{volume}{abs/1311.5939}
  (\bibinfo{year}{2013}).
\newblock
\showeprint[arxiv]{1311.5939}~[math.PR]
\urldef\tempurl%
\url{https://arxiv.org/abs/1311.5939}
\showURL{%
\tempurl}


\bibitem[Ulam(1976)]%
        {Ulam76}
\bibfield{author}{\bibinfo{person}{Stanislav~M. Ulam}.}
  \bibinfo{year}{1976}\natexlab{}.
\newblock \bibinfo{title}{Adventures of a Mathematician}.
\newblock
\newblock


\end{thebibliography}

\clearpage
\appendix
\section{Sampling a subset of a given size using a linear number of random bits}
\label{apx:random_subset} 

This appendix is devoted to proving the following lemma, which was used in Section~\ref{sub:rifflesort_n_bits} as part of our derandomization of \rifflesort.

\lemmarandomsubset*

To this aim, we design a recursive algorithm named \randomsubset, whose pseudocode is given in Algorithm~\ref{alg:randomsubset}. To sample a random subset of cardinality $h \in {1, \dots, N}$ from an input set $A$ of $N$ elements,\footnote{The case $h=0$ is trivial and is handled in line~\ref{ln:randomsubset_zero} by returning the empty set.}
\randomsubset starts by choosing an element $x$ u.a.r.\ from $A$, and a set $B$ that contains each element of $A \setminus \{x\}$ independently with probability $\frac{1}{2}$. Then, if $|B|$ is smaller than $h$, the algorithm returns a set containing $x$, all the elements in $B$, and all the elements in a subset of size $h-|B|-1$ that is recursively sampled from $A \setminus (B \cup \{x\})$. Otherwise, when $B$ contains at least $h$ elements, the algorithm returns a subset of size $h$ that is recursively sampled from $B$ (notice that it might be $|B|=h$).

\begin{algorithm2e}[t]
    \caption{\randomsubset\unskip($A$, $h$)}
    \label{alg:randomsubset}

        \lIf{$A = \emptyset$}{\Return $\emptyset$}\label{ln:randomsubset_zero}

        \BlankLine
        $x \gets$ An element chosen u.a.r. from $A$\label{ln:select_x}\;
        $B \gets$ a set obtained by selecting each element of $A \setminus \{x \}$ independently with probability $1/2$\label{ln:select_subset_B}\;

        \BlankLine
        
        \If{$|B| \le h - 1$}{ \Return $B  \cup \{x\} \cup \randomsubset( A \setminus (B \cup \{x\}), h- |B|-1)$\label{ln:B_small}}
        \Else {
            \Return \randomsubset$(B, h)$
        }
\end{algorithm2e}

Next lemma shows that the subset returned by \randomsubset$(A, h)$ is chosen u.a.r.\ among all subsets of $A$ with $h$ elements, as desired.
\begin{lemma}
    \label{lemma:randomsubset_unbiased}
    For any set $A$ of $N$ elements, any integer $h \in \{0, \dots, N\}$, and any $Z \in \binom{A}{h}$, it holds that $\Pr(\randomsubset(A,h) = Z) = \binom{|A|}{h}^{-1}$.
\end{lemma}
\begin{proof}
    The proof is by induction on $N$.
    Line~\ref{ln:randomsubset_zero} of Algorithm~\ref{alg:randomsubset} ensures that the claim holds for the base case $N=0$.
    Hence, in the rest of the proof, we consider $N>0$, we assume that the claim holds for all sets with less than $N$ elements, and we show that it holds for all sets with $N$ elements.
    
    Focus on the execution of \randomsubset$(A, h)$, let $R$ be the returned subset, and consider the element $x$ and the set $B$ selected in lines~\ref{ln:select_x} and \ref{ln:select_subset_B}, respectively. Using $0 \le |B| < N$ and $\Pr(|B| = i) = \frac{\binom{N-1}{i}}{2^{N-1}}$ for $i=0, \dots, N-1$, we have: 
    \begin{multline*}
        \Pr(R = Z) = \sum_{i = 0}^{N-1} \Pr(R = Z \mid |B|=i) \cdot \Pr(|B| = i) \\
        = \frac{1}{2^{N-1}} \left( \sum_{i = 0}^{h-1} \binom{N-1}{i}\Pr(R = Z \mid |B|=i) + \sum_{i = h}^{N-1} \binom{N-1}{i}\Pr(R = Z \mid |B|=i) \right). 
    \end{multline*}

    We now separately study the probabilities in the two summations above. 
    We start by considering $\Pr(R = Z \mid |B|=i)$ when $0 \le i \le h-1$, and we observe that the events $R = Z$ and $|B|=i$ jointly imply $B \cup \{x\} \subseteq Z$ since the returned set always includes all elements in $B \cup \{x\}$. Then, denoting with $R'$ the set returned by the call to \randomsubset$(A \setminus (B \cup  \{x\}), h - |B|- 1)$ in line~\ref{ln:B_small}, we have:
    \begin{align*}
        \Pr(R = Z  \mid |B|=i) &=
        \Pr(R = Z \mid B \cup \{x\} \subseteq Z, |B|=i) \cdot \Pr(B \cup \{x\} \subseteq Z \mid |B|=i) \\
        &= \Pr(R' = Z \setminus (B \cup \{x\}) \mid B \cup \{x\} \subseteq Z, |B|=i) \\
        & \quad\quad \cdot \Pr(x \in Z \mid |B|=i) \Pr(B \subseteq Z \setminus \{x\} \mid |B| = i, x \in Z) \\
        &= \binom{N-i-1}{h-i-1}^{-1} \cdot \frac{h}{N}  \binom{h-1}{i} \binom{N-1}{i}^{-1} \\
        &= \frac{(h-i-1)! \, (N-h)!}{(N-i-1)!} \cdot \frac{h}{N} \cdot \frac{(h-1)!}{i! \, (h-i-1)!} \cdot \frac{i! \, (N-i-1)!}{(N-1)!} 
        = \binom{N}{h}^{-1}.
    \end{align*}
    
    Next, we consider $\Pr(R = Z \mid |B|=i)$ when $h \le i < n$.
    Observe that if $R=Z$ and $|B|=i$, then we must have $B \supseteq Z$ and, in particular, $x \not\in Z$. There are exactly $\binom{N-1}{i}$ subsets of $A \setminus \{x\}$ of size $i$ and, whenever $x \not\in Z$, $\binom{N-h-1}{i-h}$ of such sets are supersets of $Z$. Then, we can write $\Pr(R = Z \mid |B|=i)$ as:
    \begin{multline*}
        \Pr(R = Z \mid B \supseteq Z, x \not\in Z, |B|=i) \cdot \Pr(B \supseteq Z \mid x \not \in Z, |B|=i) \cdot \Pr(x \not\in Z \mid |B|=i) \\
         = \binom{i}{h}^{-1} \cdot \binom{N-h-1}{i-h} \binom{N-1}{i}^{-1} \cdot \frac{N-h}{N} \\ 
        = \frac{h! \, (i-h)! }{i!} \cdot \frac{(N-h-1)!}{(i-h)! \, (N-i-1)! } \cdot \frac{i! \, (N-i-1)!}{(N-1)!} \cdot \frac{N-h}{N}
        = \binom{N}{h}^{-1}.
    \end{multline*}

    Hence, for any choice of $i \in \{0, \dots, N-1\}$, we have $\Pr(R=Z \mid |B|=i) = \binom{N}{h}^{-1}$ which can be substituted back in the above formula for $\Pr(R = Z)$ to obtain the claimed probability. Indeed:
    \[
        \frac{1}{2^{N-1}} \binom{N}{h}^{-1} \left( \sum_{i = 0}^{h-1} \binom{N-1}{i} + \sum_{i = h}^{N-1} \binom{N-1}{i} \right)
        = \binom{N}{h}^{-1} \frac{1}{2^{N-1}} \sum_{i=0}^{N-1} \binom{N-1}{i} = \binom{N}{h}^{-1}. \tag*{\qedhere}
    \] 
\end{proof}

We now show that \randomsubset$(A, h)$ can be implemented so that it only requires a linear number of random bits, except for a probability of at most $O( e^{-\sqrt{|S|}} )$.
\begin{lemma}
    \label{lemma:randomsubset_numbits}
    Let $A$ be any set of $N$ elements and $h \in \{0, \dots, N\}$.
    The number of random bits required by \randomsubset$(A, h)$ is at most $60N$ with probability at least $1 - O(e^{- \sqrt{N}})$.
\end{lemma}
\begin{proof}
    We focus on an execution of \randomsubset$(A, h)$ for a sufficiently large value of $N = |A|$, and we classify all the resulting recursive calls \randomsubset$(A',h')$ into two regimes: we say that a call is \emph{big} if $|A'| \ge \sqrt{N} + 1$ and \emph{small} otherwise. 
    We separately analyze these two regimes to provide finer (probabilistic) upper bounds on the size of set $S$ within each regime.

    We start by considering big calls. Let $x$ and $B$ be the element and the subset selected in lines~\ref{ln:select_x} and $\ref{ln:select_subset_B}$ of Algorithm~\ref{alg:randomsubset}, respectively. For each element $y \in A \setminus \{x\}$, let $I_y$ be an indicator random variable that is $1$ iff $y \in B$, so that $|B| = \sum_{y \in A \setminus \{x\}} I_y$. Using $\E[|B|] = \sum_{y \in S \setminus \{x\}} \E[ I_y ] = \frac{|A|-1}{2}$, Hoeffding's inequality allows us to write: 
    \begin{align*}
        \Pr\left(  \frac{|S|-1}{3} \le |A|  \le \frac{2(|S|-1)}{3} \right) = 
        1- \Pr\left( \Big| |A| - \E\big[|A|\big] \Big| \ge \frac{|S|-1}{6} \right) \\
        \ge 1-  2 \exp\left( - \frac{(|S|-1)^2}{18 (|S|-1)} \right) 
        \ge 1 -  2 \exp\left(- \frac{\sqrt{N}}{18} \right).
    \end{align*}

    Since the algorithm is invoked recursively on a set having a size of either $|B'|$ or $|A'| - |B| - 1$, there is a probability of at most $1 - 2^{19} \exp(- (1+ \frac{1}{18}) \sqrt{N})$ that (i) none of next $19$ recursive calls is small, and (ii) the input sets all such calls have a size larger than $2|S|/3$.
    Hence, with probability at least $1 - 2^{19} \exp(- (1 + \frac{1}{18}) \sqrt{N})\log_{3/2} N$, there are at most $\log_{3/2} N$ consecutive groups of $19$ recursive calls that start with a big call, and the size of the set $A'$ in the first call of the $i$-th such group is at most $(\frac{2}{3})^{i-1} N$.

    Regarding small calls, notice that the size of the set $B$ sampled in each such call of \randomsubset$(A', h')$ is at most $\left\lceil \frac{|A'|-1}{2} \right\rceil \le \frac{|A'|}{2}$ with probability at least $\frac{1}{2}$. 
    Then, with a probability of at least $1 - 2^{2 \sqrt{N}}$, we have that within a group of $\lceil 2 \sqrt{N} \rceil $ small calls, either the recursion reaches the final call, or the size of the input sets becomes at most $|A'|/2$. Hence, with probability at least $1- 2^{2 \sqrt{N}} \log N$, there are at most $\log N$ such groups and the size of $A'$ in the first call of the $i$-th such group is at most $ (\frac{1}{2})^{i-1} (\sqrt{N}+1)$.

    We now assume that aforementioned upper bounds on the number of groups and on the sizes of $A'$ hold for both regimes and we upper bound the number random bits used to select all the sets $A'$ and all the elements $x$ (see lines \ref{ln:select_x} and \ref{ln:select_subset_B}). For the sets $A'$, the random bits used for their selection are at most
    $
        19 N \sum_{i=1}^{+\infty} \left(\frac{2}{3} \right)^{i-1} + \left\lceil \sqrt{N} \right\rceil (\sqrt{N} + 1) \sum_{i=1}^{+\infty} \left( \frac{1}{2} \right)^{i-1} = 59 N + o(N)$.

    We now deal with the elements $x$. Observe that a uniformly random element from a set $A$ can be sampled by arbitrarily indexing the elements of $A$ with the integers in $\{0, \dots, |A|-1\}$, choosing a uniformly random integer in $j \in \{0, \dots, |A|\}$, and returning the element with index $j$.
    A simple rejection strategy allows to sample a uniformly random integer in the interval $[0, \ell-1]$ using at most $c \lceil \log \ell \rceil$ random bits with probability at least $1 - 2^{-c}$. In details, such a strategy samples an integer between $0$ and $2^{\lceil \log \ell \rceil}$ by choosing the value of each bit of its binary representation independently and u.a.r. If such an integer lies in $[0, \ell-1]$, it is accepted and returned, otherwise it is rejected and the procedure repeats.
    
    Moreover, if (up to) $\eta$ successive integers $j_1, \dots, j_\eta$ are sampled with the above strategy, where $j_i$ is chosen u.a.r.\ from $\{0, \dots, \ell_i - 1\}$ for some $\ell_i \le \ell$, the overall number of needed random bits is at most $2c \eta \lceil \log \ell \rceil$ with probability at least $\exp(- \frac{c}{2} \eta (1- \frac{1}{c} )^2)$ for $c>1$.
    This can be seen by considering the random variables $t_1, \dots, t_\eta$, where $t_i$ counts the number of trials (i.e., repetitions of the above procedure) needed to sample $j_i$. Then, $\sum_{i=1}^\eta t_i$ is stochastically dominated by a sum of $\eta$ independent shifted geometric random variables of parameter $\frac{1}{2}$, i.e., by a negative binomial random variable of parameters $\eta$ and $\frac{1}{2}$. We can use a tail bound for negative binomial random variables (see, e.g., \cite{brown2011wasted}) to obtain $\Pr\left(\sum_{i}^N t_i > 2 c \eta \right) \le \exp(- \frac{c}{2} \eta (1- \frac{1}{c} )^2)$ for any $c>1$.

Since, for sufficiently large values of $N$, there are at most $\eta = 4 \lceil \sqrt{N} \rceil \log N$ recursive calls of \randomsubset$(A, h)$ with $h>0$, each of which chooses an element $x$ from the set $A$ which contains at most $N$ elements, we have that no more than $16 \lceil \sqrt{N} \rceil \log N \cdot \lceil \log N \rceil   = o(N)$ random bits are needed with probability at least $1-e^{-\sqrt{N} \log N}$.

    Overall, when $N$ is larger than a suitable constant, the number of used random bits is at most $60\eta$ with probability at least $1 -  2^{19} \exp(-(1 + \frac{1}{18})\sqrt{N}) \log_{3/2} N - 2^{-2\sqrt{N}} \log N - e^{- \sqrt{N} \log N} = 1  - O(e^{-\sqrt{N}})$, as claimed.
\end{proof}

Lemma~\ref{lemma:derand_set_split} follows from Lemma~\ref{lemma:randomsubset_unbiased}, Lemma~\ref{lemma:randomsubset_numbits}, and from observing that the running time of \randomsubset is asymptotically dominated by the number of used random bits.

\end{document}